\PassOptionsToPackage{unicode}{hyperref}
\PassOptionsToPackage{hyphens}{url}
\PassOptionsToPackage{dvipsnames,svgnames,x11names}{xcolor}
\documentclass[
  12pt,
  letterpaper,
]{article}
\usepackage{amsmath,amssymb}
\usepackage{iftex}
\ifPDFTeX
  \usepackage[T1]{fontenc}
  \usepackage[utf8]{inputenc}
  \usepackage{textcomp} 
\else 
  \usepackage{unicode-math} 
  \defaultfontfeatures{Scale=MatchLowercase}
  \defaultfontfeatures[\rmfamily]{Ligatures=TeX,Scale=1}
\fi
\usepackage{lmodern}
\ifPDFTeX\else
\fi
\IfFileExists{upquote.sty}{\usepackage{upquote}}{}
\IfFileExists{microtype.sty}{
  \usepackage[]{microtype}
  \UseMicrotypeSet[protrusion]{basicmath} 
}{}
\makeatletter
\@ifundefined{KOMAClassName}{
  \IfFileExists{parskip.sty}{%
    \usepackage{parskip}
  }{
    \setlength{\parindent}{0pt}
    \setlength{\parskip}{6pt plus 2pt minus 1pt}}
}{
  \KOMAoptions{parskip=half}}
\makeatother
\usepackage{xcolor}
\usepackage[margin = 1in]{geometry}
\usepackage{longtable,booktabs,array}
\usepackage{calc} 
\usepackage{etoolbox}
\makeatletter
\patchcmd\longtable{\par}{\if@noskipsec\mbox{}\fi\par}{}{}
\makeatother
\IfFileExists{footnotehyper.sty}{\usepackage{footnotehyper}}{\usepackage{footnote}}
\makesavenoteenv{longtable}
\usepackage{graphicx}
\makeatletter
\def\maxwidth{\ifdim\Gin@nat@width>\linewidth\linewidth\else\Gin@nat@width\fi}
\def\maxheight{\ifdim\Gin@nat@height>\textheight\textheight\else\Gin@nat@height\fi}
\makeatother
\setkeys{Gin}{width=\maxwidth,height=\maxheight,keepaspectratio}
\makeatletter
\def\fps@figure{htbp}
\makeatother
\makeatletter
\g@addto@macro\normalsize{%
  \setlength\abovedisplayskip{3pt}
  \setlength\belowdisplayskip{3pt}
  \setlength\abovedisplayshortskip{3pt}
  \setlength\belowdisplayshortskip{3pt}
  \setlength{\textfloatsep}{10pt plus 2pt minus 2pt}
  \setlength{\intextsep}{10pt plus 2pt minus 2pt}
  \setlength{\abovecaptionskip}{3pt}
  \setlength{\belowcaptionskip}{4pt}
}
\makeatother

\usepackage{setspace}
\usepackage{etoolbox}
\usepackage{microtype}
\AtBeginEnvironment{thebibliography}{%
  \singlespacing
  \fontsize{8.2}{8.8}\selectfont
  \setlength{\parskip}{0pt}
  \setlength{\bibsep}{1.5pt}
}

\AtBeginEnvironment{table}{\singlespacing}
\AtBeginEnvironment{longtable}{\singlespacing}
\AtBeginEnvironment{algorithm}{\singlespacing}
\AtBeginEnvironment{thebibliography}{\singlespacing}
\usepackage{booktabs}
\usepackage{adjustbox}
\usepackage{multirow}
\usepackage{longtable}
\usepackage{paralist}
\usepackage{fancyhdr}
\usepackage{dcolumn}
\usepackage{graphicx}
\usepackage{amsthm}
\usepackage{etoolbox}
\usepackage[font=small,labelfont=it]{caption}
\usepackage{algorithm2e}
\newcommand{\minipagebegin}{\begin{minipage}}
\newcommand{\minipageend}{\end{minipage}}
\RestyleAlgo{ruled}
\newcounter{algoline}

\AtBeginEnvironment{algorithm}{\setcounter{algoline}{0}}

\theoremstyle{plain}
\newtheorem{theorem}{Theorem}

\newtheorem{lemma}{Lemma}
\theoremstyle{definition}

\newtheorem{example}{Example}
\newtheorem{assumption}{Assumption}
\usepackage{mathtools}
\usepackage{bbm}
\usepackage{booktabs}
\usepackage{longtable}
\usepackage{array}
\usepackage{multirow}
\usepackage{wrapfig}
\usepackage{float}
\usepackage{colortbl}
\usepackage{pdflscape}
\usepackage{tabu}
\usepackage{threeparttable}
\usepackage{threeparttablex}
\usepackage[normalem]{ulem}
\usepackage{makecell}
\usepackage{xcolor}
\ifLuaTeX
  \usepackage{selnolig}  
\fi
\usepackage[round,authoryear]{natbib}
\usepackage{bookmark}
\IfFileExists{xurl.sty}{\usepackage{xurl}}{} 
\hypersetup{
  pdfauthor={Duong Trinh},
  colorlinks=true,
  linkcolor={blue},
  filecolor={Maroon},
  citecolor={Blue},
  urlcolor={blue},
  pdfcreator={LaTeX via pandoc}}

\title{Endogenous Selection and Spillovers:

Bayesian Inference for Policy-Relevant Causal Effects}
\author{Duong Trinh\footnote{University of Graz, Department of Economics. Universitätsstraße 15 Bauteil, Graz, Austria, 8010. E-mail address: \href{mailto:duong.trinh@uni-graz.at}{\nolinkurl{duong.trinh@uni-graz.at}}.}}
\date{}

\begin{document}
\maketitle

\thispagestyle{empty}

\begin{abstract}
\begin{spacing}{1.15}
This paper develops a new econometric framework to identify and estimate policy-relevant causal effects in contexts with endogenous selection into treatment and spillovers within single large networks or spatial settings. Conventional causal inference methods relying on either unconfoundedness or no-interference assumptions are generally inadequate in these scenarios. We introduce a Spillover Roy model that jointly models endogenous treatment selection and potential outcomes while allowing spillovers through a low-dimensional exposure mapping of neighbors' treatments. The model captures heterogeneous treatment responses across levels of latent resistance to treatment and neighborhood exposure.  Within this framework, we define policy-relevant direct, spillover, and total effects under feasible policy changes and show that the total effect decomposes into a direct component from policy-induced participation and a spillover component from policy-induced changes in neighborhood treatment exposure. For estimation and inference, we develop a Bayesian data-augmentation algorithm with parameter expansion that enables efficient posterior computation and coherent uncertainty quantification for heterogeneous causal effects and policy counterfactuals. An application to the U.S. Opportunity Zones program finds positive direct effects on housing development but limited spillover benefits, while counterfactual policy analysis reveals diminishing returns from program expansion.

\smallskip


\medskip

\end{spacing}
\end{abstract}

\newpage
\pagenumbering{arabic}

\section{Introduction}\label{introduction}

Spillovers, often referred to as interference, and endogenous selection into treatment are pervasive features of economic settings that complicate conventional causal inference approaches. Spillovers arise when a unit's outcome depends not only on its own treatment status but also on the treatments received by others within its social network or geographic area \citep[see, e.g.,][]{forastiere2021identification, giffin2022generalized}. For example, participants in after-school programs may influence the behavior or academic performance of non-participating peers, while place-based policies such as regional development incentives can affect neighboring communities through migration or business relocation. Simultaneously, participation in many programs is not randomly assigned \citep[see, e.g.,][]{abbring2007econometric}. Individuals or regions often enter treatment based on unobserved characteristics, such as motivation or growth potential, which are also related to potential outcomes, generating endogenous selection. In such environments, policy changes influence outcomes through two interconnected mechanisms: by altering who participates in the program and by reshaping the treatment exposure faced by others. Ignoring either mechanism can lead to biased estimates of program effectiveness and misleading conclusions regarding policy design and evaluation.

Nonetheless, the issue of policy-relevant causal inference in settings where both endogenous selection and spillover effects occur remains insufficiently explored due to the intertwined challenges. Existing approaches typically focus on identifying causal effects at fixed treatment or exposure states may not correspond directly to feasible policy changes. In practice, policymakers are interested in the consequences of modifying program rules, such as expanding eligibility, changing subsidies, or altering designation criteria, beyond the scope of existing treatment effects. As a result, conventional treatment-effect parameters do not fully capture multiple channels of policy interventions in such settings.

This paper develops a new econometric framework for policy-relevant causal effects that simultaneously handle endogenous selection and spillovers in large non-clustered networks or spatial settings. We extend the Generalised Roy model to accommodate network or spatial interactions by allowing potential outcomes to depend on both individual treatment status and an exposure measure summarizing neighbors' treatment. This structure preserves the economic interpretation of selection into treatment while incorporating spillovers in a tractable way. Our resulting Spillover Roy model captures heterogeneous treatment responses across the latent resistance distribution and across exposure levels. Within this framework, we define policy-relevant causal effects as contrasts between expected outcomes under alternative feasible policy regimes. A policy shift may affect outcomes through two distinct channels: induced changes in own treatment participation and induced changes in exposure to treatments experienced by neighbors. Our framework therefore decomposes the total policy impact into direct and spillover components, providing a transparent characterization of the mechanisms through which policies operate.

Our methodology contributes to three strands of research. \emph{First}, it relates to the literature on causal inference under interference. One approach imposes partial interference, whereby spillovers operate within exogenously defined clusters \citep[e.g.,][]{hudgens2008toward, sobel2006randomized, manski2013identification}, with recent work allowing for noncompliance or endogenous treatment take-up \citep{ditraglia2023identifying, vazquez2023causal}. Our setting is closer to work on general interference in a single network or spatial environment, where a low-dimensional exposure mapping summarizes the relevant treatment configuration. Existing methods typically rely on randomized treatment or unconfoundedness conditional on observed covariates \citep[e.g.,][]{aronow2017estimating, leung2020treatment, forastiere2021identification, forastiere2022estimating}. Recent studies relax treatment exogeneity: \citet{hoshino2024causal} use instrumental exposure mappings to identify local direct and indirect effects under noncompliance, while \citet{chen2025heterogeneous} model treatment choices as a network equilibrium and identify heterogeneous marginal exposure effects. Related work studies policy effects when interference is mediated through market-equilibrium variables under randomized treatment \citep{munro2025treatment}. Our setting instead combines endogenous treatment selection with network or spatial exposure and focuses on counterfactual changes in the rule governing treatment participation.

\emph{Second}, our policy estimands build on the literature concerning marginal and policy-relevant treatment effects under endogenous selection \citep[e.g.,][]{heckman2005structural, heckman2007econometric, carneiro2011estimating, mogstad2018using, sasaki2023estimation, opper2024late}. In the Generalized Roy framework, treatment is governed by a latent-index selection rule, and the treatment effects may vary with the unobserved determinants of treatment choice. The marginal treatment effect (MTE) characterizes treatment gains along this latent resistance margin and provides a building block for policy evaluation. A central insight of this literature is that policy evaluation requires specifying the policy-relevant target population: the individuals whose treatment choices would change under the counterfactual policy may not coincide with those whose choices are shifted by the available instrument. Among various measures or treatment effects, the PRTE has the advantage of directly evaluating alternative policy scenarios under consideration and can be represented as policy-specific weighted averages of the underlying MTE. PRTE evaluates a change from a baseline to a counterfactual treatment-selection regime and, when normalized by the change in participation, measures the average outcome gain per net participant induced by the policy. We extend this policy-evaluation approach to settings with interference.
A policy-induced change in treatment selection now affects outcomes through both own participation and the resulting change in neighbors' treatment exposure. We therefore define policy-relevant direct and spillover effects and show how the total policy impact decomposes into these two components, linking policy evaluation to heterogeneity along both the latent resistance and neighborhood-exposure margins.

\emph{Third}, our estimation strategy relates to Bayesian methods for endogenous selection and latent-index models, including inference on potential-outcome distributions in the presence of unidentified dependence parameters \citep{poirier2003predictive} and parameter-expanded data-augmentation techniques for handling covariance restrictions and identifying normalizations \citep[e.g.,][]{ding2014bayesian, dougan2018bayesian, zhang2026parameter}. Building on this literature, we develop a parameter-expanded Gibbs sampler tailored to the richer latent structure of the Spillover Roy model, which jointly estimates endogenous treatment selection and regime-specific potential outcomes and delivers posterior inference for heterogeneous treatment, spillover, and policy-relevant effects. Monte Carlo experiments show that our proposed proposed procedure yeilds valid inference, while naive approaches that ignore either endogenous selection or spillovers can exhibit substantial bias.

We apply the proposed framework to evaluate the causal effects of the Opportunity Zones (OZ) program on housing growth in U.S. census tracts. We model the designation of treated areas as an endogenous selection process influenced by local economic characteristics and political decisions. Our empirical results reveal substantial heterogeneity in treatment gains consistent with selection-on-gains behavior: areas with higher expected returns are more likely to receive the program. We find positive direct effects of designation on housing development but limited evidence of beneficial spillovers for neighboring non-designated areas. Furthermore, policy counterfactual analysis shows that expanding the program induces diminishing returns as marginal entrants generate smaller --- and eventually negative --- direct gains, while spillover benefits increase but remain insufficient to sustain positive net effects under large expansions.

The remainder of this paper is structured as follows. In Section \ref{BCIES-Section2}, we present Spillover Roy Model and define causal estimands with key identification assumptions. In Section \ref{BCIES-Section3}, we propose Bayesian data-augmentation approach to estimate the model and conduct inference. We then evaluate our method using simulations in Section \ref{BCIES-Section4} and investigate the causal impact of the U.S. Opportunity Zones (OZ) program on economic outcomes in Section \ref{BCIES-Section5}. Finally, Section \ref{BCIES-Section6} concludes with brief remarks and policy recommendations.

\section{The Spillover Roy Model}\label{BCIES-Section2}

\subsection{General Model Setup}\label{BCIES-Section2_1}

We consider a general setting for \(n\) agents (\(i = 1,\ldots,n\)) which involves treatment selection and outcome determination with spillovers.

\textbf{Treatment selection}

Let \(D_i\) be the observed binary treatment decision, which takes the value of \(1\) if the unit receives the treatment and \(0\) otherwise. This can be regarded as individual treatment and determined by a latent-index representation as follows
\begin{equation}
\begin{split}
  D_i^* &= \mu(Z_i,X_i) - U_i,\\
  D_i &= 1 \text{ if } D_i^* \geq 0, \quad D_i = 0 \text{ otherwise},
\end{split}
\end{equation}
where \(D_i^*\) denotes the net benefit, or latent utility, from receiving the treatment and \(U_i\) captures an unobserved component of the treatment choice. Vector \(X_i\) contains observed characteristics that may jointly influence treatment participation and potential outcomes. To identify causal effects under endogenous selection, we additionally observe a vector of excluded variables \(Z_i\), which shifts treatment participation without directly affecting potential outcomes. Throughout the paper, \(Z_i\) plays the role of an instrumental variable in the structural selection equation. Variation in \(Z_i\) provides exogenous changes in treatment propensity while satisfying the exclusion restriction imposed later in the identification analysis.

Assume that \(U_i\) is continuously distributed with a strictly increasing cumulative distribution function \(F_U\). Define \(V_i \coloneqq F_U(U_i)\), then it has uniformly distribution and indicates different quantile level of \(U_i\). Also define \(\nu(Z_i,X_i) \coloneqq F_U(\mu(Z_i,X_i))\), which is the mean scale utility function in discrete choice theory, we can thereby rewriting the treatment rule as:
\begin{equation}
D_i = \mathbbm{1} \{\mu(Z_i,X_i) \geq U_i\} = \mathbbm{1} \{\nu(Z_i,X_i) \geq V_i\}.
\end{equation}
The variable \(V_i\) indexes the unit's latent resistance to treatment: units with higher values of are less likely to participate for a given value of \(\nu(Z_i,X_i)\). This representation is standard in the Generalized Roy model \citep{heckman2005structural} and will be central for defining marginal treatment effects and policy-relevant effects. Because the latent resistance \(V_i\) may be statistically dependent on the potential-outcome disturbances, treatment selection mechanism is generally endogenous.

\textbf{Outcome determination with spillovers}

Under interference, potential outcomes of unit \(i\) may depend not only on its individual treatment but also on the treatment assignments of other units connected to \(i\) through a network or spatial interaction structure. Let \(\mathbf{D} = (D_1,\ldots,D_n)\) denote the population treatment vector, and let \(Y_i(D_i,\mathbf{D}_{-i})\) denote the potential outcome of unit \(i\) under the treatment assignment \(\mathbf{D} = (D_1,\ldots,D_n)\), where \(\mathbf{D}_{-i} = (D_1, \ldots, D_{i-1}, D_{i+1}, \ldots, D_n)\) collects the treatment assignments of all units other than \(i\). Without additional restrictions, unit \(i\) may have a distinct potential outcome for each possible realization of \(\mathbf{D}\). Because the number of potential outcomes grows exponentially with the number of units, such a framework is generally infeasible for both identification and estimation in large networks. Following the network causal inference literature \citep[see, e.g.,][]{aronow2017estimating, manski2013identification}, we impose an exposure-mapping restriction that summarizes the aspects of neighbors' treatment assignments relevant for unit \(i\)'s outcome.

\begin{assumption}{[Exposure mapping]}
\label{Assp1}
There exists a known measurable mapping $T:\{1,\ldots,n\}\times\{0,1\}^{n-1}\times\mathcal W\rightarrow\mathcal T$ such that, for each unit $i$,
$$
T_i=T(i,\mathbf{D}_{-i},\mathbf{W}),
$$
where $\mathbf W\in\mathcal W$ denotes the observed network or spatial structure and $\mathcal T$ is the exposure space. Depending on the application, $\mathbf W$ may be either a network adjacency matrix or a spatial weights matrix. For every $d\in\{0,1\}$ and every pair of treatment vectors $\mathbf d_{-i}$ and $\mathbf d'_{-i}$,
$$
T(i,\mathbf d_{-i},\mathbf W) = T(i,\mathbf d'_{-i},\mathbf W) \quad \Longrightarrow \quad Y_i(d,\mathbf d_{-i}) = Y_i(d,\mathbf d'_{-i}).
$$
\end{assumption}

Thus, conditional on unit \(i\)'s own treatment status, the potential outcome depends on the treatment assignments of all other units only through the exposure state \(T_i\), rather than through the full treatment vector \(\mathbf D_{-i}\). Consequently, there exists a function \(Y_i^{(d)}: \mathcal T \rightarrow \mathbbm{R}\) such that
\[
Y_i(d,\mathbf D_{-i}) = Y_i^{(d)}(T_i),\qquad d\in\{0,1\}.
\]
Throughout the paper, the interaction matrix \(\mathbf{W}\) is assumed to be known and predetermined with respect to the latent disturbances in the treatment and outcome equations. This matrix characterizes the pattern of either network interference or spatial interference. For the remainder of the paper, we focus on the weighted neighborhood treatment as the exposure mapping
\[
T_i = \bar D_{\mathcal N i} = \sum_{j\neq i}w_{ij}D_j, 
\qquad w_{ii}=0, 
\qquad \sum_{j\neq i} w_{ij} = 1,
\]
The scalar \(\bar D_{\mathcal N i}\) summarizes the treatment intensity in unit \(i\)'s neighborhood.\footnote{When \(w_{ij} = 1/N_i\) for all neighbors \(j\in \mathcal N (i)\), where \(N_i\) denotes the number of neighbors of unit \(i\), \(\bar D_{\mathcal N i}\) reduces to the proportion of treated neighbors. More generally, the weights may reflect heterogeneous interaction strengths, geographic proximity, or other measures of network influence, yielding a weighted neighborhood treatment intensity.} Other commonly used exposure mappings include binary exposure indicators (e.g., at least one treated neighbor), unweighted treated-neighbor counts, higher-order neighborhood summaries, and other nonlinear exposure measures.

Under the weighted neighborhood exposure mapping, we specify the treated and untreated potential outcomes as
\begin{equation}
\begin{split}
Y_i^{(1)}(\bar{D}_{\mathcal{N}i})  &= \mu_1\left(\bar{D}_{\mathcal{N}i}, X_i\right) + \varepsilon_i^{(1)}, \quad \text{ and }\\
Y_i^{(0)}(\bar{D}_{\mathcal{N}i})  &= \mu_0\left(\bar{D}_{\mathcal{N}i}, X_i\right) + \varepsilon_i^{(0)},
\end{split}
\end{equation}
where \(X_i\) denotes a vector of observed individual characteristics, \(\mu_d\left(\bar{D}_{\mathcal{N}i},X_i\right) = \mathbbm{E} \left[ Y_i^{(d)} \mid \bar{D}_{\mathcal{N}i}, X_i\right]\) is the regime-specific mean response function, and \(\varepsilon_i^{(d)}\) is the corresponding idiosyncratic disturbance for each \(d\in\{0,1\}\).

Let \(Y_i\) be the revealed outcome, which equals the treated potential outcome when \(i\) is treated (\(D_i = 1\)) and equals untreated potential outcome when \(i\) is untreated (\(D_i = 0\))
\begin{equation}
Y_i = D_i Y_i^{(1)} + (1-D_i)  Y_i^{(0)}.
\end{equation}
Combining endogenous treatment selection with spillovers operating through the specified exposure mapping, the general framework is
\begin{equation}
\label{KeyModel1}
\begin{split}
D_i &= \mathbbm{1} \{\nu(Z_i,X_i) \geq V_i\},
\\
\bar{D}_{\mathcal{N}i} &= \sum_{j=1, j\neq i}^{n} w_{ij} D_j, \quad \sum_{j=1, j\neq i}^{n} w_{ij} = 1,
\\
Y_i^{(1)}  &= \mu_1\left(\bar{D}_{\mathcal{N}i}, X_i\right) + \varepsilon_i^{(1)},
\\
Y_i^{(0)}  &= \mu_0\left(\bar{D}_{\mathcal{N}i}, X_i\right) + \varepsilon_i^{(0)},
\\
Y_i &= D_i Y_i^{(1)} + (1-D_i)  Y_i^{(0)}.
\end{split}
\end{equation}

\begin{example} 
(Place-based policies with spatial spillovers) Consider the Opportunity Zones (OZ) program, where $D_i$ indicates OZ designation of census tract $i$ and $Y_i$ denotes a local economic outcome, such as housing development. Designation may be endogenous because unobserved local characteristics, such as growth potential or political support, can affect both designation and potential outcomes. Thus, latent resistance $V_i$ may be correlated with the unobserved determinants of potential outcomes, while an excluded variable $Z_i$ provides exogenous variation in designation. OZ designation may also affect nearby tracts through housing-market, migration, or business-location responses, generating spatial spillovers through the neighborhood exposure defined above.
\end{example}

\begin{example} 
(Educational programs with network spillovers) Consider an after-school program (ASP), where $D_i$ indicates participation of student $i$ and $Y_i$ denotes an outcome such as academic performance or social-emotional development. Because participation is self-selected, latent resistance $V_i$ may be correlated with unobserved determinants of potential outcomes, generating endogenous treatment selection. An excluded cost shifter $Z_i$ provides exogenous variation in participation. Program participation may also affect students indirectly through interactions with participating peers, generating network spillovers through the neighborhood exposure defined above.
\end{example}

\subsection{Parametric Spillover Roy Model and Identification}\label{BCIES-Section2_2}

We now introduce the assumptions under which the structural objects in \eqref{KeyModel1} are identified.

\begin{assumption}{[Parametric Spillover Roy model]}
\label{Assp2}
Suppose the data are generated by
\begin{equation}
\label{KeyModel2}
\begin{split}
D_i   &= \mathbbm{1}\{Z_i\alpha + X_i\beta^{(D)} + \varepsilon_i^{(D)} > 0\},\\
\bar{D}_{\mathcal{N}i} &= \sum_{j=1, j\neq i}^{n} w_{ij} D_j, \quad \sum_{j=1, j\neq i} w_{ij} = 1,\\
Y_i^{(1)} &= \delta^{(1)}\bar{D}_{\mathcal{N}i} + X_i\beta^{(1)} + \varepsilon_i^{(1)},\\ 
Y_i^{(0)} &= \delta^{(0)}\bar{D}_{\mathcal{N}i} + X_i\beta^{(0)} + \varepsilon_i^{(0)}, \\
Y_i   &= D_i Y_i^{(1)} + (1 - D_i) Y_i^{(0)}.\\
\end{split}
\end{equation}
\end{assumption}

This specification imposes linearity in the mean response functions and the treatment index while accommodating endogenous selection through unrestricted dependence between the treatment-selection disturbance, \(\varepsilon_i^{(D)}\) and the potential-outcome disturbances, \(\bigl(\varepsilon_i^{(1)},\varepsilon_i^{(0)}\bigr)\). Substituting the potential outcomes into the switching equation yields
\[
Y_i = X_i \beta^{(0)} + D_i X_i (\beta^{(1)} - \beta^{(0)}) + \delta^{(0)} \bar{D}_{\mathcal{N}i} + D_i(\delta^{(1)} - \delta^{(0)})\bar{D}_{\mathcal{N}i} + (1-D_i)\varepsilon_i^{(0)} + D_i\varepsilon_i^{(1)},
\]
which makes clear that the model allows the individual outcome to depend on own treatment, neighborhood treatment, and their interaction.

\emph{Remark}. The specification in Assumption \ref{Assp2} extends the Generalized Roy
framework \citep{heckman2005structural} by allowing potential outcomes to depend on neighborhood treatment exposure, with potentially different spillover effects across treatment states, \(\delta^{(1)}\) and \(\delta^{(0)}\). When \(\delta^{(0)}=\delta^{(1)}=0\), the model reduces to the canonical Generalized Roy model.

\begin{assumption}{[Instrument validity]\\}
\label{Assp3}
\textnormal{(i) Instrument exogeneity.}
$$
\left(
\varepsilon_i^{(D)},
\varepsilon_i^{(1)},
\varepsilon_i^{(0)}
\right)
\perp\!\!\!\perp
(X_i,Z_i),
$$
where $\perp\!\!\!\perp$ denotes the statistical independence.

\textnormal{(ii) Instrument relevance.}
The excluded variable $Z_i$ generates nondegenerate variation in the treatment-selection index
$$
\nu_i = Z_i\alpha+X_i\beta^{(D)}.
$$
\end{assumption}

Assumption \ref{Assp3} requires the instrumental variable to satisfy both exogeneity and relevance. Part (i) states that the observed covariates and the instrumental variable are jointly independent of the latent disturbances governing treatment selection and potential outcomes. Together with the structural specification in Assumption \ref{Assp2}, this implies that the instrumental variable affects outcomes only through its effect on treatment selection. Part (ii) requires the instrument to generate sufficient variation in the latent treatment-selection index, thereby ensuring identification of the treatment-selection equation. These conditions are standard in structural latent-index models of endogenous treatment selection and heterogeneous treatment effects \citetext{\citealp[see, e.g.,][]{carneiro2011estimating}; \citealp{brinch2017beyond}; \citealp[and][]{cornelissen2018benefits}}.

\begin{assumption}{[Finite-mixture distribution and cross-unit independence]}
\label{Assp4}
Conditional on $\mathbf{W}$, the disturbance vectors are independently distributed across units and,
for $i = 1,\ldots,n$
\begin{equation}
\label{AS-finitemixture}
\begin{split}
\boldsymbol{\varepsilon}_i = \begin{bmatrix}
\varepsilon_i^{(D)} \\ \varepsilon_i^{(1)} \\ \varepsilon_i^{(0)}
\end{bmatrix} &\overset{ind}{\sim} \sum_{g=1}^G \pi_g N\left(0, \Sigma_g \right) \\ 
\end{split}
\end{equation}
$$
\begin{aligned}
\text{where} \quad \sum_{g=1}^G \pi_g = 1; 
\quad \Sigma_g = \begin{bmatrix}
1 & \sigma_{1D,g} & \sigma_{0D,g} \\
& \sigma_{1,g}^2 & \sigma_{10,g} \\
& & \sigma_{0,g}^2
\end{bmatrix}
\end{aligned}
$$
\end{assumption}

Assumption \ref{Assp4} models the joint distribution of the latent disturbances as a finite mixture of multivariate normal distributions. The first diagonal element of each component covariance matrix is normalized to one, reflecting the standard scale normalization in binary latent-index models \citep[see, e.g.,][]{cameron2005microeconometrics, chan2019bayesian}. Without this normalization, the parameters in the selection equation are identified only up to scale. This finite-mixture specification flexibly approximates the joint distribution of the latent disturbances and allows for non-Gaussian heterogeneity while preserving tractability for estimation and inference. In addition, we note that the latent resistance \(V_i\) in \eqref{KeyModel1} can now be represented as
\[
V_i = \Phi(-\varepsilon_i^{(D)}),
\]
where \(\Phi(\cdot)\) denotes the standard normal cdf. Under Assumption \(\ref{Assp4}\) and the normalization \(\mathop{\mathrm{Var}}(\varepsilon_i^{(D)})=1\), it follows that \(V_i\sim U(0,1)\).

The conditional independence across units implies that, conditional on the predetermined network structure \(\mathbf{W}\), neighborhood treatment carries no additional information about unit \(i\)'s latent disturbances beyond that contained in its own treatment decision and observed covariates. Accordingly,
\[
\mathop{\mathrm{\mathbbm{E}}}[\varepsilon_i^{(d)} \mid D_i, \bar{D}_{\mathcal{N}i}, X_i, Z_i, \mathbf{W}] = \mathop{\mathrm{\mathbbm{E}}}[\varepsilon_i^{(d)} \mid D_i, X_i, Z_i, \mathbf{W}], \quad d\in\{0,1\}.
\]

\begin{theorem}[Identification of the Spillover Roy Model]
\label{Theorem1}
Suppose Assumptions \ref{Assp1}--\ref{Assp4} hold. Then the parameters of the Spillover Roy model in \eqref{KeyModel2},
$$
\alpha,\;\beta^{(D)},\;\beta^{(1)},\;\beta^{(0)},\;\delta^{(1)},\;\delta^{(0)},
$$
and the aggregate covariance parameters
$$
\sigma_{1D}
\coloneqq
\mathop{\mathrm{Cov}}\!\left(\varepsilon_i^{(1)},\varepsilon_i^{(D)}\right),
\qquad
\sigma_{0D}
\coloneqq
\mathop{\mathrm{Cov}}\!\left(\varepsilon_i^{(0)},\varepsilon_i^{(D)}\right),
$$
are identified.
\end{theorem}

\emph{Proof.} See Appendix \ref{BCIES-APDX_ProofProp1}.

\subsection{Causal Estimands}\label{causal-estimands}

Our targeted estimands include a hierarchy of causal parameters: marginal (structural) effects, average effects, and policy-relevant effects.

\textbf{Marginal Structural Objects}

We begin with the primitive objects that characterize heterogeneity in both selection into treatment and neighborhood exposure.

We define the \emph{Marginal Treatment Effect (MTE)} under interference as
\begin{equation}
\label{eq:MTE-bcies}
\mathop{\mathrm{MTE}}(\bar d_{\mathcal N},v,x)
\coloneqq
\mathop{\mathrm{\mathbbm{E}}}\!\left[
Y_i^{(1)}(\bar d_{\mathcal N})-Y_i^{(0)}(\bar d_{\mathcal N})
\mid V_i=v,\; X_i=x
\right].
\end{equation}
This generalizes the classical marginal treatment effect to settings with interference by allowing treatment effects to vary with both latent selection heterogeneity \(v\) and the neighborhood exposure level \(\bar d_{\mathcal N}\).

We define the \emph{Marginal Spillover Effect (MSE)} as
\begin{equation}
\label{eq:MSE-bcies}
\mathop{\mathrm{MSE}}(d,\bar d_{\mathcal N},v,x)
\coloneqq
\frac{\partial}{\partial \bar d_{\mathcal N}}
\mathop{\mathrm{\mathbbm{E}}}\!\left[
Y_i^{(d)}(\bar d_{\mathcal N})
\mid V_i=v,\; X_i=x
\right],
\qquad d\in\{0,1\},
\end{equation}
This measures the local causal response of potential outcomes to a marginal increase in neighborhood treatment exposure for individuals with treatment status \(d\) and latent resistance \(v\).

\begin{theorem}[Identification of Marginal Structural Objects]
\label{Theorem2}
Suppose Assumptions \ref{Assp1}--\ref{Assp4} hold. Then the marginal structural objects
$$
\mathop{\mathrm{MTE}}(\bar d_{\mathcal N},v,x)
\quad\text{and}\quad
\mathop{\mathrm{MSE}}(d,\bar d_{\mathcal N},v,x),\; d\in\{0,1\},
$$
are identified. In particular, 
\begin{align}
\label{eq:MTE-identified}
\mathop{\mathrm{MTE}}(\bar d_{\mathcal N},v,x)
&=
(\delta^{(1)}-\delta^{(0)})\bar d_{\mathcal N}
+
x(\beta^{(1)}-\beta^{(0)})
+
\mathop{\mathrm{\mathbbm{E}}}\!\left[
\varepsilon_i^{(1)}-\varepsilon_i^{(0)}
\mid V_i=v
\right],
\\
\label{eq:MSE-identified}
\mathop{\mathrm{MSE}}(d,\bar d_{\mathcal N},v,x)
&=\delta^{(d)}, \qquad d\in\{0,1\}.
\end{align}
Moreover, under Assumption \ref{Assp4},
\begin{equation}
\label{eq:eps-diff-generic}
\mathop{\mathrm{\mathbbm{E}}}\!\left[
\varepsilon_i^{(1)}-\varepsilon_i^{(0)}
\mid V_i=v
\right] = - (\sigma_{1D}-\sigma_{0D})\Phi^{-1}(v),
\end{equation}
is identified from the finite-mixture distribution of
$\left(\varepsilon_i^{(D)},\varepsilon_i^{(1)},\varepsilon_i^{(0)}\right)$.
Hence both $\mathop{\mathrm{MTE}}(\bar d_{\mathcal N},v,x)$ and $\mathop{\mathrm{MSE}}(d,\bar d_{\mathcal N},v,x)$ are identified functions of the model primitives.
\end{theorem}

\emph{Proof.} See Appendix \ref{BCIES-APDX_ProofProp2}.

Evaluated at mean values of the covariates \(x\), MTE would exhibit heterogeneity in treatment effects due to \(\bar{d}_{\mathcal{N}}\) if \(\delta^{(1)} - \delta^{(0)} \neq 0\). Furthermore, \(\delta^{(1)} - \delta^{(0)}\) implies the patterns of interaction effects between individual treatment and neighborhood treatment: Positive interaction (\(\delta^{(1)} - \delta^{(0)} > 0\)) means the treatment is more valuable when more of neighbors are treated. In contrast, negative interaction (\(\delta^{(1)} - \delta^{(0)} < 0\)) means the treatment is more valuable when less of neighbors are treated.

\textbf{Policy-Relevant Effects}

Policy changes typically modify eligibility rules, subsidies, or program intensity, thereby shifting the probability of treatment participation. Such changes do not affect all individuals equally: they primarily induce participation among individuals who are marginal with respect to treatment choice. In contexts of social or spatial interactions, policy changes may also alter outcomes indirectly through changes in neighborhood treatment exposure. Our aim is therefore to evaluate \emph{policy-relevant average effects per induced participant}. In particular, we decompose the impact of a policy change into two components: a \emph{direct effect}, capturing the gain for individuals who are induced into treatment by the policy, and a \emph{spillover effect}, capturing the gain generated through the induced change in neighborhood treatment exposure. This framework extends the policy-rerelevant treatment effect (PRTE) concept of \citet{heckman2005structural} to settings with network or spatial interactions. The normalization by the induced participation share ensures that all policy-relevant effects are interpreted as average gains per additional participant generated by the policy change, which is the natural metric for policy evaluation.

Let \(a\) and \(a'\) denote two policy regimes, with \(a'\) representing the more generous policy. We assume that each policy modifies the treatment-selection rule through a known transformation of the identified selection equation, while the latent resistance \(V_i\) remains invariant across policy regimes Thus, for each policy regime \(r \in \{a,a'\}\),
\[
D_i^r = \mathbbm{1}\{P^r(Z_i,X_i)\ge V_i\},
\]
where \(P^r(Z_i,X_i)\) denotes the counterfactual treatment propensity under policy \(r\). The corresponding neighborhood treatment exposure is
\[
\bar D_{\mathcal N i}^r = \sum_{j\neq i} w_{ij} D^r_j,
\]
and the realized outcome is
\[
Y_i^r = Y_i^{(D_i^r)}(\bar D_{\mathcal N i}^r).
\]
We assume that the policy change is \emph{pointwise monotone}, namely,
\[
P^{a'}(z,x) \ge P^{a}(z,x), \text{ for almost every } (z,x).
\]
Thus, policy \(a'\) weakly expands participation relative to policy \(a\).

Define the conditional share of \emph{induced participants} by
\[
\Delta P(x;a,a')
\coloneqq
\mathop{\mathrm{\mathbbm{E}}}[P^{a'}(Z_i,X_i) - P^{a}(Z_i,X_i)\mid X_i=x],
\]
and assume throughout that \(\Delta P(x;a,a')>0\). This quantity represents the expected increase in treatment participation generated by the policy among individuals with covariates \(X_i=x\).

\emph{Policy-Relevant Direct Effect (PRDE)}

The Policy-Relevant Direct Effect is defined as
\[
\mathop{\mathrm{PRDE}}(x;a,a') \coloneqq \mathop{\mathrm{\mathbbm{E}}}\left[Y_i^{(1)}(\bar D_{\mathcal N i}^{a})- Y_i^{(0)}(\bar D_{\mathcal N i}^{a}) \mid P^{a}(Z_i,X_i)<V_i\le P^{a'}(Z_i,X_i), X_i=x \right].
\]
This estimand measures the average treatment gain for individuals induced into treatment by the policy, evaluated at the neighborhood treatment exposure that would prevail under the baseline policy \(a\).

\emph{Policy-Relevant Spillover Effect (PRSE)}

The Policy-Relevant Spillover Effect is defined as
\[
\mathop{\mathrm{PRSE}}(x;a,a') \coloneqq
\frac{\mathop{\mathrm{\mathbbm{E}}}\left[Y_i^{(D_i^{a'})}(\bar D_{\mathcal N i}^{a'}) - Y_i^{(D_i^{a'})}(\bar D_{\mathcal N i}^{a}) \mid X_i=x\right]}{\Delta P(x;a,a')}.
\]
This estimand isolates the contribution of the policy-induced change in neighborhood treatment exposure while holding each individual's treatment status fixed at its value under the new policy \(a'\).

\emph{Policy-Relevant Total Effect (PRTOT)}

The overall policy effect per induced participant is
\[
\mathop{\mathrm{PRTOT}}(x;a,a') \coloneqq \frac{\mathop{\mathrm{\mathbbm{E}}}[Y_i^{a'}-Y_i^{a}\mid X_i=x]}{\Delta P(x;a,a')}.
\]
Using the decomposition,
\[
Y_i^{a'}-Y_i^{a}
=
\underbrace{
\Big(
Y_i^{(D_i^{a'})}(\bar D_{\mathcal N i}^{a'})
-
Y_i^{(D_i^{a'})}(\bar D_{\mathcal N i}^{a})
\Big)
}_{\text{spillover component}}
+
\underbrace{
\Big(
Y_i^{(D_i^{a'})}(\bar D_{\mathcal N i}^{a})
-
Y_i^{(D_i^{a})}(\bar D_{\mathcal N i}^{a})
\Big)
}_{\text{direct component}},
\]
it follows that
\[
\mathop{\mathrm{PRTOT}}(x;a,a') = \mathop{\mathrm{PRDE}}(x;a,a') + \mathop{\mathrm{PRSE}}(x;a,a').
\]

To characterize the policy-relevant direct effect, define the mean neighborhood exposure under policy \(a\) among policy-induced participants with latent resistance \(v\) by
\[
\bar d_{\mathcal N}^{a,S}(v,x)
\coloneqq
\mathop{\mathrm{\mathbbm{E}}}\!\left[
\bar D_{\mathcal N i}^{a}
\;\middle|\;
P^a(Z_i,X_i)<v\le P^{a'}(Z_i,X_i),
V_i=v,
X_i=x
\right].
\]

\begin{theorem}[Identification of Policy-Relevant Direct, Spillover, and Total Effects]
\label{Theorem3}

Suppose Assumptions \ref{Assp1}--\ref{Assp4} and the policy-counterfactual conditions stated above hold. Then, 
\[
\mathop{\mathrm{PRDE}}(x;a,a')
=
\int_0^1
\mathop{\mathrm{MTE}}\!\left(
\bar d_{\mathcal N}^{a,S}(v,x),
v,
x
\right)
h_{PR}(v\mid x;a,a')\,dv,
\]
where the policy weights are
\[
h_{PR}(v\mid x;a,a')
\coloneqq
\frac{
F_{P^a\mid X}(v\mid x)
-
F_{P^{a'}\mid X}(v\mid x)
}{
\Delta P(x;a,a')
}.
\]
Moreover,
$$
\mathop{\mathrm{PRSE}}(x;a,a')
=
\frac{
\mathop{\mathrm{\mathbbm{E}}}\!\left[
\delta^{(D_i^{a'})}
\left(
\bar D_{\mathcal N i}^{a'}
-
\bar D_{\mathcal N i}^{a}
\right)
\mid X_i=x
\right]
}{
\Delta P(x;a,a')
},
$$
where $\delta^{(D_i^{a'})} \coloneqq D_i^{a'}\delta^{(1)} + \left(1-D_i^{a'}\right)\delta^{(0)}.$

Consequently,
$$
\mathop{\mathrm{PRTOT}}(x;a,a') = \mathop{\mathrm{PRDE}}(x;a,a') + \mathop{\mathrm{PRSE}}(x;a,a'), 
$$
and the policy-relevant direct, spillover, and total effects are identified.
\end{theorem}

\emph{Proof.} See Appendix \ref{BCIES-APDX_ProofProp3}.

\section{Bayesian Estimation and Inference}\label{BCIES-Section3}

\subsection{Bayesian data augmentation}\label{bayesian-data-augmentation}

We conduct Bayesian inference for the structural parameters and the causal estimands defined in Section \ref{BCIES-Section2_2}. Posterior computation involves two latent-data features. First, treatment status reveals only the sign of the latent treatment-selection index \(D_i^*\). Second, only one of the two potential outcomes is observed for each unit. We therefore employ data augmentation, treating the latent selection index and the missing potential outcome as auxiliary variables.

Define \(\mathbf{P}_i= \bigl[Z_i \ X_i \bigr]\), \(\mathbf{Q}_i=\bigl[ \bar D_{\mathcal N i} \ X_i \bigr]\), \(\boldsymbol{\gamma}=\bigl[\alpha \ \beta^{(D)}\bigr]\), \(\boldsymbol{\kappa}_1=\bigl[\delta^{(1)} \ \beta^{(1)}\bigr]\), \(\boldsymbol{\kappa}_0=\bigl[\delta^{(0)} \ \beta^{(0)}\bigr]\). The Spillover Roy model can then be written as
\begin{equation}
\label{eq:GGRM-mixture-main}
\begin{aligned}
D_i^*
    &=\mathbf P_i^\top\boldsymbol{\gamma}
      +\varepsilon_i^{(D)},\\
Y_i^{(1)}
    &=\mathbf Q_i^\top\boldsymbol{\kappa}_1
      +\varepsilon_i^{(1)},\\
Y_i^{(0)}
    &=\mathbf Q_i^\top\boldsymbol{\kappa}_0
      +\varepsilon_i^{(0)},\\
D_i&=\mathbbm{1}\{D_i^*>0\},\\
Y_i&=D_iY_i^{(1)}+(1-D_i)Y_i^{(0)}.
\end{aligned}
\end{equation}

As in Section \ref{BCIES-Section2_2}, the joint distribution of the unobservables is represented by a finite mixture of multivariate normal distributions. Let \(c_i\in\{1,\ldots,G\}\) denote the latent mixture component for unit \(i\), with
\[
\Pr(c_i=g\mid\boldsymbol{\pi})=\pi_g,\qquad
\sum_{g=1}^G\pi_g=1,
\]
and
\begin{equation}
\label{eq:mixture-main}
\boldsymbol{\varepsilon}_i\mid c_i=g
\sim
\mathcal N(\mathbf 0,\mathbf\Sigma_g),
\end{equation}
where
\begin{equation}
\label{eq:Sigma-main}
\boldsymbol{\varepsilon}_i
= \begin{bmatrix}
\varepsilon_i^{(D)}\\
\varepsilon_i^{(1)}\\
\varepsilon_i^{(0)}
\end{bmatrix},
\qquad
\mathbf\Sigma_g=
\begin{bmatrix}
1 & \sigma_{1D,g} & \sigma_{0D,g}\\
\sigma_{1D,g} & \sigma_{1,g}^2 & \sigma_{10,g}\\
\sigma_{0D,g} & \sigma_{10,g} & \sigma_{0,g}^2
\end{bmatrix}.
\end{equation}
The normalization \(\Sigma_{g,11}=1\) fixes the scale of the latent
treatment-selection equation.

Let \(Y_i^{\mathrm{mis}}\) denote the unobserved potential outcome and
define the augmented outcome vector
\begin{equation*}
\mathbf{L}_i^* \coloneqq
\begin{bmatrix}
D_i^*\\
Y_i^{(1)}\\
Y_i^{(0)}
\end{bmatrix}
=
\begin{bmatrix}
D_i^*\\
D_iY_i+(1-D_i)Y_i^{miss}\\
D_iY_i^{miss}+(1-D_i)Y_i
\end{bmatrix},
\end{equation*}
With the corresponding block-diagonal design matrix and parameter vector
\[
\mathbf R_i=
\begin{bmatrix}
\mathbf P_i^\top & 0 & 0\\
0 & \mathbf Q_i^\top & 0\\
0 & 0 & \mathbf Q_i^\top
\end{bmatrix}, \qquad
\boldsymbol{\theta}
=
\begin{bmatrix}
\boldsymbol{\gamma}\\
\boldsymbol{\kappa}_1\\
\boldsymbol{\kappa}_0
\end{bmatrix},
\]
the augmented model is
\begin{equation}
\label{eq:augmented-model-main}
\mathbf L_i^*
=
\mathbf R_i\boldsymbol{\theta}
+
\boldsymbol{\varepsilon}_i,
\qquad
\boldsymbol{\varepsilon}_i\mid c_i=g
\sim\mathcal N(\mathbf0,\mathbf\Sigma_g).
\end{equation}

Conditional on the augmented data and mixture allocations, \eqref{eq:augmented-model-main} has a Gaussian regression representation. This representation is the basis of the posterior sampler developed below. Details of the complete-data likelihood and the conditional distributions of the augmented variables are provided in Appendix \ref{BCIES-APDX_sampling}.

\subsection{Prior specification and parameter expansion}\label{prior-specification-and-parameter-expansion}

We complete the model by assigning priors to the regression parameters, mixture probabilities, and component-specific covariance matrices. Specifically,
\begin{align}
\boldsymbol{\theta}&\sim\mathcal N(\underline{\boldsymbol{\mu}}_{\theta},\underline{\mathbf V}_{\theta}),
\\
\boldsymbol{\pi}&\sim\mathcal{D}ir(\underline{\omega}_1,\ldots,\underline{\omega}_G).
\end{align}
A complication arises in posterior simulation of the covariance matrices \(\mathbf\Sigma_g\). Because treatment depends only on the sign of \(D_i^*\), the scale of the latent selection equation is not identified and its disturbance variance is normalized to one, \(\Sigma_{g,11}=1\) (see Assumption \ref{Assp4}). Directly sampling \(\mathbf\Sigma_g\) subject to this restriction complicates covariance updating.

We address this problem using parameter expansion. For each mixture component \(g\), introduce a positive expansion parameter \(\tau_g\) and define
\begin{equation}
\mathbf A_g
\coloneqq
\mathop{\mathrm{diag}}(\tau_g,1,1),
\qquad
\widetilde{\mathbf\Sigma}_g
\coloneqq
\mathbf A_g\mathbf\Sigma_g\mathbf A_g.
\end{equation}
Unlike \(\mathbf\Sigma_g\), the expanded covariance matrix \(\widetilde{\mathbf\Sigma}_g\) is unrestricted and is assigned an inverse-Wishart prior,
\begin{equation}
\label{eq:IW-expanded-main}
\widetilde{\mathbf\Sigma}_g \sim\mathcal W^{-1}(\mathbf{I}_3,\underline{\nu}).
\end{equation}
Posterior simulation proceeds on this expanded parameter space. At each covariance update, an auxiliary value of \(\tau_g^2\) is first drawn from its conditional prior implied by equation \eqref{eq:IW-expanded-main}. This auxiliary draw rescales the selection-equation residuals and enters the expanded residual cross-product matrix. Conditional on the transformed residuals, the unrestricted covariance matrix \(\widetilde{\mathbf\Sigma}_g\) is then drawn from its inverse-Wishart conditional posterior. After this update, the scale associated with the new expanded covariance draw is \(\tau_g^{2}=\widetilde{\Sigma}_{g,11},\)
and the covariance matrix in the identified parameterization is recovered as
\begin{equation}
\label{eq:recoverSigma-main}
\mathbf\Sigma_g = \mathbf A_g^{-1}\widetilde{\mathbf\Sigma}_g\mathbf A_g^{-1},\qquad \text{where }\mathbf A_g = \mathop{\mathrm{diag}}(\tau_g,1,1) 
\end{equation}
which restores the identifying normalization \(\Sigma_{g,11}=1\). The parameter expansion therefore permits standard inverse-Wishart updating on an unrestricted covariance space while enforcing the normalization through deterministic rescaling. Parameter expansion can also improve the mixing of data-augmentation algorithms for latent-variable and sample-selection models \citep{ding2014bayesian, dougan2018bayesian}. The induced prior on \((\tau_g^2,\mathbf\Sigma_g)\) and the corresponding derivations are provided in Appendix \ref{BCIES-APDX_PX}.

\subsection{Posterior computation}\label{posterior-computation}

Let \(\Theta=\left(\boldsymbol{\theta},\boldsymbol{\pi},\{\mathbf\Sigma_g\}_{g=1}^G\right)\) denote the model parameters. Augmenting the observed data with \(\mathbf D^*\), \(\mathbf Y^{\mathrm{mis}}\), and the mixture allocations
\(\mathbf c\), the posterior distribution is proportional to
\begin{equation}
\label{eq:posterior-mixture-main}
p(\Theta,\mathbf D^*,\mathbf Y^{\mathrm{mis}},\mathbf c
   \mid\mathbf Y,\mathbf D) \propto
p(\mathbf Y,\mathbf D,\mathbf D^*,\mathbf Y^{\mathrm{mis}},\mathbf c
   \mid\Theta)
p(\boldsymbol{\theta})
p(\boldsymbol{\pi})
\prod_{g=1}^G p(\widetilde{\mathbf\Sigma}_g).
\end{equation}
We construct a parameter-expanded Gibbs sampler that alternates between latent-data augmentation and parameter updating. Given the current parameter values, the missing potential outcomes are sampled from Gaussian conditional distributions, the latent treatment indices from Gaussian distributions truncated according to observed treatment status, and the component allocations from multinomial distributions. Conditional on the resulting complete data, the regression parameters have a Gaussian posterior and the mixture probabilities have a Dirichlet posterior. Component-specific covariance matrices are updated on the expanded scale and subsequently transformed back to the identified parameterization using \eqref{eq:recoverSigma-main}. Algorithm \ref{alg:GGRM_mixture_short} summarizes the resulting sampler. Closed-form expressions for all conditional posterior distributions and further implementation details are provided in Appendix \ref{BCIES-APDX_MCMC}.

\begin{algorithm}[H]
\caption{Parameter-expanded Gibbs sampler}
\label{alg:GGRM_mixture_short}

Initialize
$\boldsymbol{\theta}$,
$\boldsymbol{\pi}$,
$\{\mathbf\Sigma_g\}_{g=1}^G$,
and $\mathbf c$.

\For{$s=1,\ldots,S$}{

\textbf{1. Latent-data augmentation}

\Indp
Sample missing potential outcomes
$\mathbf Y^{\mathrm{mis}}$;

Sample latent treatment indices
$\mathbf D^*$ subject to
$D_i=\mathbbm 1\{D_i^*>0\}$;

Sample mixture allocations $\mathbf c$.

\Indm

\textbf{2. Parameter updates}

\Indp
Sample regression parameters $\boldsymbol{\theta}$;

Sample mixture probabilities $\boldsymbol{\pi}$;

For each $g$, sample the auxiliary expansion scale $\tau_g^2$,
update $\widetilde{\boldsymbol{\Sigma}}_g$ on the expanded scale,
and normalize to recover $\boldsymbol{\Sigma}_g$.

\Indm
}

\Return retained posterior draws of
$\boldsymbol{\theta}$,
$\boldsymbol{\pi}$,
and $\{\mathbf\Sigma_g\}_{g=1}^G$.
\end{algorithm}

\subsection{Posterior inference for causal and policy-relevant effects}\label{posterior-inference-for-causal-and-policy-relevant-effects}

The structural and policy-relevant causal quantities introduced in Section \ref{BCIES-Section2_2} are functions of the model parameters and, for policy counterfactuals, of the treatment and exposure distributions induced by alternative policy regimes. Bayesian inference for these quantities follows directly from the retained posterior draws. For each posterior draw \(\Theta^{[s]}\), we evaluate the corresponding marginal treatment and
spillover effects,
\[
{\mathop{\mathrm{MTE}}}^{[s]}(\bar d_{\mathcal N},v,x),
\qquad
{\mathop{\mathrm{MSE}}}^{[s]}(d,\bar d_{\mathcal N},x),
\]
using the expressions derived in Section \ref{BCIES-Section2_2}. Posterior means are used as point estimates, and posterior quantiles provide credible intervals. Thus, posterior uncertainty about treatment selection, outcome responses, and their dependence is propagated jointly to the heterogeneous causal effects. For a counterfactual policy regime \(\mathcal P\), we additionally construct the treatment decisions and neighborhood exposures implied by the policy at each posterior draw. Comparing these quantities across the baseline and counterfactual regimes yields posterior draws of the policy-relevant direct and spillover effects,
\[
{\mathop{\mathrm{PRDE}}}^{[s]},
\qquad
{\mathop{\mathrm{PRSE}}}^{[s]},
\qquad
{\mathop{\mathrm{PRTOT}}}^{[s]}
=
{\mathop{\mathrm{PRDE}}}^{[s]}+{\mathop{\mathrm{PRSE}}}^{[s]}.
\]
Hence, uncertainty in the structural parameters is propagated through both endogenous treatment participation and the neighborhood exposure generated by counterfactual policies.

\section{Simulation Study}\label{BCIES-Section4}

We assess the finite-sample performance of the proposed framework through a series of Monte Carlo experiments. The baseline design incorporates both selection on unobservables and spillovers, the two features that motivate the Spillover Roy model. We evaluate recovery of the structural parameters and heterogeneous marginal treatment effects and compare the correctly specified \emph{Spillover Roy model (SRM)} with a misspecified \emph{Non-Spillover Roy model (NSRM)} that omits neighborhood exposure. Additional robustness designs are reported in Appendix \ref{BCIES-APDX_Simulations}.

\subsection{Data Generating Processes}\label{BCIES-Section4_1}

For each replication, we generate five exogenous variables \(\widetilde{\mathbf X}_k\), \(k=1,\ldots,5\) independently from the standard normal distribution and set \(\mathbf X=\left[\boldsymbol{\iota}_n^\top,\widetilde{\mathbf X}^\top\right]^\top.\). The instrumental variable \(\mathbf Z\) is also generated independently from the same distribution. We construct the interaction matrix \(\mathbf W\) following the interaction design described in \citet{liu2010gmm}. The matrix is block diagonal, with each block representing a group-specific interaction network. The sample is partitioned into \(G = 30\) groups. Group sizes \(m_g\) are allowed to vary around \(n/G\): for the first \(G-1\) groups, \(m_g\) is drawn from
\[
\left\{
\lfloor n/G\rfloor-2,\ldots,\lfloor n/G\rfloor+3
\right\},
\]
and the size of the final group is chosen so that \(\sum_{g=1}^G m_g=n\). Within group \(g\), the interaction matrix \(\mathbf W_g\) is generated as follows. For each row \(i=1,\ldots,m_g\), we draw \(\tau_{ig}\) uniformly from \(\{1,2,3,4\}\) and connect unit \(i\) to the subsequent \(\tau_{ig}\) units, wrapping around the group boundary when necessary. We then symmetrize the group-specific matrices and construct
\[
\mathbf{W} \coloneqq \mathop{\mathrm{diag}}(\bf{W}_1^\top+\bf{W}_1,\ldots,\bf{W}_g^\top+\bf{W}_g)
\]
After row normalization, neighborhood exposure is given by \(\bar{\mathbf{D}}_{\mathcal{N}} =\mathbf{W}\mathbf{D}\). Treatment selection and potential outcomes follow the Spillover Roy model described in \eqref{KeyModel2}. We set
\[
\beta^{(D)} = \left[0,0.5,-0.5,0.3,-0.3,-0.2\right]^\top, 
\]
\[
\beta^{(1)} = \left[2,0.4,-0.4,0.3,-0.3,0.2\right]^\top, \quad \beta^{(0)} = \left[1,0.4,-0.4,0.3,-0.3,0.2\right]^\top, 
\]
with instrument strength \(\alpha = 1.5\). Spillovers are present in both potential-outcome regimes, with \(\delta^{(1)} = 1.5\) and \(\delta^{(0)} = 0.5\). The disturbance vector
\(\epsilon_i = \left[\epsilon_i^{(D)} , \epsilon_i^{(1)} , \epsilon_i^{(0)}\right]^\top\) is independently distributed across units as
\[
\begin{aligned}
\epsilon_i\overset{iid}{\sim}\mathcal N(0,\Sigma),
\qquad
\Sigma=
\begin{bmatrix}
1 & 0.9 & 0.7\\
0.9 & 1 & 0.6\\
0.7 & 0.6 & 1
\end{bmatrix},
\end{aligned}
\]
Thus, \(\sigma_{D}^2 = \sigma_{1}^2 = \sigma_{0}^2 = 1\) and \((\rho_{1D},\rho_{0D},\rho_{10}) = (0.9,0.7,0.6)\). The nonzero correlations between the selection disturbance and the potential-outcome disturbances generate selection on unobservables, while \(\rho_{10}>0\) induces positive cross-regime dependence between \(\bigl(Y^{(1)},Y^{(0)}\bigr)\).

We consider sample sizes \(n\in\{500,1{,}000,2{,}000\}\). For each simulated sample, we estimate two specifications. The \emph{SRM} includes neighborhood exposure and corresponds to the correctly specified model. The \emph{NSRM} omits \(\bar D_{\mathcal Ni}\) and therefore provides a benchmark for assessing the consequences of ignoring spillovers. For each specification, the MCMC sampler is run for \(11{,}000\) iterations, with the first \(1{,}000\) draws discarded as burn-in. The prior hyperparameters are
\[\underline{\boldsymbol{\mu}}_{\theta} = \mathbf{0}_{21};\quad \underline{\mathbf V}_{\theta} = 10^2\mathbf{I}_{21};\quad \underline{\nu} = 4;\quad   \underline{\omega}_1 = \ldots = \underline{\omega}_G = 1/G.\]
Across \(N_{\mathrm{sim}}=1{,}000\) Monte Carlo replications, posterior means are used as point estimates and 95\% posterior credible intervals are used for interval estimation. We report Monte Carlo bias, root mean squared error (RMSE), and empirical coverage of the 95\% credible intervals.

\subsection{Simulation Results}\label{BCIESsection4.2}

Table \ref{tab:tab-SMCresults-Mar26-DGP2-SI} reports finite-sample performance for the structural parameters. Under the correctly specified SRM, biases are generally small, RMSEs decline with the sample size, and empirical coverage is close to the nominal 95\% level for most parameters. The results indicate that the proposed Bayesian procedure accurately recovers the principal features of the Spillover Roy model in samples of the sizes considered. The NSRM produces a markedly different pattern. By construction, it sets the spillover coefficients \(\delta^{(1)}\) and \(\delta^{(0)}\) to zero even though both are nonzero in the data-generating process. More importantly, this omission contaminates estimation of other structural parameters. In particular, the outcome coefficients and several covariance parameters exhibit persistent bias and substantial coverage distortions. These discrepancies do not disappear as \(n\) increases, consistent with misspecification bias rather than finite-sample variability.

\begin{spacing}{1.5}
\begingroup
\setlength{\tabcolsep}{3.5pt}
\renewcommand{\arraystretch}{0.88}
\begin{table}[H]
\centering
\caption{\label{tab:tab-SMCresults-Mar26-DGP2-SI}Simulation Results for Model Parameters}
\centering
\resizebox{\ifdim\width>\linewidth\linewidth\else\width\fi}{!}{
\begin{threeparttable}
\begin{tabular}[t]{cccccccccccccccc}
\toprule
\multicolumn{3}{c}{ } & \multicolumn{4}{c}{Quantities of Interest} & \multicolumn{9}{c}{Other Parameters} \\
\cmidrule(l{3pt}r{3pt}){4-7} \cmidrule(l{3pt}r{3pt}){8-16}
Model & Metric & n & $\delta^{(1)}$ & $\delta^{(0)}$ & $\delta^{(1)}-\delta^{(0)}$ & $\sigma_{1D}-\sigma_{0D}$ & $\alpha$ & $\beta^{(D)}$ & $\beta_1^{(1)}$ & $\beta_1^{(0)}$ & $\sigma_1^2$ & $\sigma_0^2$ & $\rho_{1D}$ & $\rho_{0D}$ & $\rho_{10}$\\
\midrule
 &  & True Value & 1.500 & 0.500 & 1.000 & 0.200 & 1.500 & 0.000 & 2.000 & 1.000 & 1.000 & 1.000 & 0.900 & 0.700 & 0.600\\
\cmidrule{1-16}
 &  & 500 & 0.003 & -0.002 & 0.006 & -0.040 & 0.064 & -0.020 & 0.014 & -0.001 & 0.009 & 0.007 & -0.046 & -0.006 & 0.105\\

 &  & 1000 & 0.007 & 0.000 & 0.007 & -0.024 & 0.028 & -0.012 & 0.005 & 0.000 & 0.002 & 0.002 & -0.024 & -0.001 & 0.110\\

 & \multirow{-3}{*}{\centering\arraybackslash Bias} & 2000 & -0.006 & -0.001 & -0.006 & -0.013 & 0.017 & -0.005 & 0.010 & 0.000 & 0.000 & 0.001 & -0.012 & 0.000 & 0.110\\

 &  & 500 & 0.220 & 0.240 & 0.325 & 0.119 & 0.142 & 0.076 & 0.130 & 0.148 & 0.105 & 0.098 & 0.061 & 0.082 & 0.126\\

 &  & 1000 & 0.152 & 0.176 & 0.233 & 0.087 & 0.085 & 0.054 & 0.089 & 0.108 & 0.067 & 0.070 & 0.035 & 0.058 & 0.122\\

 & \multirow{-3}{*}{\centering\arraybackslash RMSE} & 2000 & 0.108 & 0.119 & 0.163 & 0.063 & 0.059 & 0.037 & 0.064 & 0.072 & 0.048 & 0.049 & 0.023 & 0.042 & 0.119\\

 &  & 500 & 0.968 & 0.961 & 0.954 & 0.956 & 0.900 & 0.928 & 0.964 & 0.960 & 0.966 & 0.954 & 0.929 & 0.961 & 0.972\\

 &  & 1000 & 0.948 & 0.937 & 0.953 & 0.950 & 0.920 & 0.941 & 0.954 & 0.941 & 0.976 & 0.953 & 0.943 & 0.961 & 0.967\\

\multirow{-9}{*}[1\dimexpr\aboverulesep+\belowrulesep+\cmidrulewidth]{\centering\arraybackslash SRM} & \multirow{-3}{*}{\centering\arraybackslash Coverage} & 2000 & 0.942 & 0.948 & 0.956 & 0.952 & 0.921 & 0.941 & 0.954 & 0.952 & 0.979 & 0.957 & 0.958 & 0.960 & 0.960\\
\cmidrule{1-16}
 &  & 500 & -1.500 & -0.500 & -1.000 & -0.029 & 0.060 & -0.013 & 0.757 & 0.249 & 0.138 & 0.021 & -0.089 & -0.014 & 0.114\\

 &  & 1000 & -1.500 & -0.500 & -1.000 & -0.015 & 0.027 & -0.007 & 0.754 & 0.249 & 0.131 & 0.016 & -0.070 & -0.008 & 0.124\\

 & \multirow{-3}{*}{\centering\arraybackslash Bias} & 2000 & -1.500 & -0.500 & -1.000 & -0.011 & 0.016 & -0.002 & 0.755 & 0.251 & 0.124 & 0.016 & -0.061 & -0.005 & 0.130\\

 &  & 500 & 1.500 & 0.500 & 1.000 & 0.124 & 0.144 & 0.075 & 0.763 & 0.264 & 0.180 & 0.102 & 0.103 & 0.084 & 0.132\\

 &  & 1000 & 1.500 & 0.500 & 1.000 & 0.091 & 0.087 & 0.054 & 0.757 & 0.257 & 0.152 & 0.073 & 0.078 & 0.058 & 0.135\\

 & \multirow{-3}{*}{\centering\arraybackslash RMSE} & 2000 & 1.500 & 0.500 & 1.000 & 0.067 & 0.061 & 0.037 & 0.756 & 0.254 & 0.136 & 0.053 & 0.066 & 0.042 & 0.138\\

 &  & 500 & 0.000 & 0.000 & 0.000 & 0.963 & 0.891 & 0.932 & 0.000 & 0.173 & 0.805 & 0.962 & 0.676 & 0.964 & 0.987\\

 &  & 1000 & 0.000 & 0.000 & 0.000 & 0.957 & 0.924 & 0.941 & 0.000 & 0.016 & 0.656 & 0.944 & 0.508 & 0.969 & 0.986\\

\multirow{-9}{*}[1\dimexpr\aboverulesep+\belowrulesep+\cmidrulewidth]{\centering\arraybackslash NSRM} & \multirow{-3}{*}{\centering\arraybackslash Coverage} & 2000 & 0.000 & 0.000 & 0.000 & 0.957 & 0.922 & 0.935 & 0.000 & 0.000 & 0.422 & 0.946 & 0.296 & 0.965 & 0.963\\
\bottomrule
\end{tabular}
\begin{tablenotes}[para]
\item \textit{Notes:} This table displays the average bias (Bias), the Root Mean Squared Error (RMSE), and the coverage rate (Coverage) across $R=100$ replicates; where $\text{Bias}=R^{-1}\sum_{r=1}^R (\hat{\alpha}_r-\alpha),\text{ RMSE}=\sqrt{R^{-1}\sum_{r=1}^R (\hat{\alpha}_r-\alpha)^2},$ and $\text{ Coverage}=R^{-1}\sum_{r=1}^R \mathbbm{1}\{\alpha\in \widehat{CI}_{0.95,r}\}$. The rows contain results for models with/without spatial interference and for various sample size $n$.
\end{tablenotes}
\end{threeparttable}}
\end{table}

\endgroup
\end{spacing}

We next evaluate estimation of the marginal treatment effect over both latent resistance to treatment and neighborhood exposure. To keep the main presentation concise, Table \ref{tab:tab-SMCresults-DGP2-MTE-combined} reports results at three representative resistance values, \(v\in\{0.1,0.5,0.9\}\), for low, medium, and high neighborhood exposure, \(\bar d_{\mathcal N}\in\{0.1,0.5,0.9\}\). Complete results over \(v\in\{0.1,\ldots,0.9\}\) are reported in Supplementary Tables \ref{tab:tab-SMCresults-Mar26-DGP2-MTE1s}---\ref{tab:tab-SMCresults-Mar26-DGP2-MTE3s}.

\begin{table}[!ht]
\centering
\caption{\label{tab:tab-SMCresults-DGP2-MTE-combined}Finite-sample performance of MTE estimation across neighborhood exposure regimes.}
\centering
\fontsize{9}{11}\selectfont
\begin{threeparttable}
\begin{tabular}[t]{>{\raggedright\arraybackslash}p{1.1cm}>{\centering\arraybackslash}p{0.7cm}rrrrrrrrr}
\toprule
\multicolumn{2}{c}{ } & \multicolumn{3}{c}{$n=500$} & \multicolumn{3}{c}{$n=1000$} & \multicolumn{3}{c}{$n=2000$} \\
\cmidrule(l{3pt}r{3pt}){3-5} \cmidrule(l{3pt}r{3pt}){6-8} \cmidrule(l{3pt}r{3pt}){9-11}
Model & $v$ & Bias & RMSE & Coverage & Bias & RMSE & Coverage & Bias & RMSE & Coverage\\
\midrule
\addlinespace[0.3em]
\hline
\multicolumn{11}{l}{\textbf{Panel A. Low exposure}}\\
 & 0.1 & 0.372 & 0.429 & 0.603 & 0.386 & 0.416 & 0.267 & 0.390 & 0.406 & 0.044\\

 & 0.5 & 0.409 & 0.428 & 0.103 & 0.405 & 0.414 & 0.006 & 0.404 & 0.408 & 0.000\\

\multirow{-3}{*}{\raggedright\arraybackslash \hspace{1em}NSRM} & 0.9 & 0.446 & 0.482 & 0.452 & 0.424 & 0.443 & 0.154 & 0.418 & 0.428 & 0.012\\
\cmidrule{1-11}
 & 0.1 & -0.036 & 0.237 & 0.961 & -0.024 & 0.176 & 0.958 & -0.008 & 0.124 & 0.951\\

 & 0.5 & 0.016 & 0.168 & 0.963 & 0.006 & 0.118 & 0.961 & 0.009 & 0.085 & 0.964\\

\multirow{-3}{*}{\raggedright\arraybackslash \hspace{1em}SRM} & 0.9 & 0.068 & 0.217 & 0.974 & 0.037 & 0.149 & 0.969 & 0.026 & 0.110 & 0.953\\
\cmidrule{1-11}
\addlinespace[0.3em]
\hline
\multicolumn{11}{l}{\textbf{Panel B. Medium exposure}}\\
 & 0.1 & -0.028 & 0.217 & 0.958 & -0.014 & 0.155 & 0.948 & -0.010 & 0.113 & 0.947\\

 & 0.5 & 0.009 & 0.126 & 0.946 & 0.005 & 0.084 & 0.955 & 0.004 & 0.061 & 0.952\\

\multirow{-3}{*}{\raggedright\arraybackslash \hspace{1em}NSRM} & 0.9 & 0.046 & 0.187 & 0.974 & 0.024 & 0.131 & 0.970 & 0.018 & 0.097 & 0.959\\
\cmidrule{1-11}
 & 0.1 & -0.034 & 0.204 & 0.963 & -0.021 & 0.148 & 0.951 & -0.010 & 0.108 & 0.949\\

 & 0.5 & 0.018 & 0.114 & 0.958 & 0.009 & 0.076 & 0.961 & 0.007 & 0.055 & 0.960\\

\multirow{-3}{*}{\raggedright\arraybackslash \hspace{1em}SRM} & 0.9 & 0.070 & 0.175 & 0.968 & 0.040 & 0.121 & 0.962 & 0.024 & 0.086 & 0.961\\
\cmidrule{1-11}
\addlinespace[0.3em]
\hline
\multicolumn{11}{l}{\textbf{Panel C. High exposure}}\\
 & 0.1 & -0.428 & 0.479 & 0.514 & -0.414 & 0.442 & 0.276 & -0.410 & 0.425 & 0.053\\

 & 0.5 & -0.391 & 0.410 & 0.117 & -0.395 & 0.404 & 0.004 & -0.396 & 0.401 & 0.000\\

\multirow{-3}{*}{\raggedright\arraybackslash \hspace{1em}NSRM} & 0.9 & -0.354 & 0.398 & 0.610 & -0.376 & 0.397 & 0.237 & -0.382 & 0.394 & 0.030\\
\cmidrule{1-11}
 & 0.1 & -0.032 & 0.247 & 0.955 & -0.018 & 0.175 & 0.960 & -0.012 & 0.129 & 0.931\\

 & 0.5 & 0.020 & 0.177 & 0.948 & 0.012 & 0.122 & 0.954 & 0.005 & 0.086 & 0.951\\

\multirow{-3}{*}{\raggedright\arraybackslash \hspace{1em}SRM} & 0.9 & 0.072 & 0.220 & 0.959 & 0.043 & 0.156 & 0.960 & 0.022 & 0.106 & 0.961\\
\bottomrule
\end{tabular}
\begin{tablenotes}[para]
\item \textit{Notes:}  Panels A--C report posterior bias, RMSE, and 95\% credible interval coverage for representative values of the latent resistance to treatment ($v = 0.1, 0.5, 0.9$). Results are based on $R=1{,}000$ Monte Carlo replications.
\end{tablenotes}
\end{threeparttable}
\end{table}

\vspace{-0.5\baselineskip}

The SRM performs well across all three exposure levels. Bias is small, RMSE generally declines with \(n\), and empirical coverage remains close to 95\% throughout most of the resistance distribution. The NSRM behaves differently. At low and high exposure, its MTE estimates exhibit substantial and persistent bias, and coverage deteriorates sharply as the sample size increases. At medium exposure, the misspecification is less consequential because the omitted exposure component happens to generate much smaller distortion under this design. The contrast across exposure levels illustrates an important feature of the problem: a model that ignores spillovers may appear adequate at particular exposure values while failing severely elsewhere. Increasing the sample size therefore improves precision under the correctly specified model but does not eliminate the distortions generated by omitting neighborhood exposure. The resulting undercoverage is particularly pronounced in regions of the exposure space where the omitted spillover component is economically important. These results demonstrate that correctly modeling interference is necessary for reliable inference on heterogeneous treatment effects.

This finding is also relevant for the policy analysis in our proposed framework. Policy-relevant direct effects are constructed by averaging MTEs over individuals induced into treatment under a policy change and over the neighborhood exposures generated by the baseline policy. Reliable policy evaluation therefore requires accurate estimation of the MTE over both the resistance and exposure dimensions. The simulation evidence shows that the SRM provides such recovery in the baseline design, whereas an analysis that omits spillovers can substantially distort the MTE surface. Accordingly, the results support using the estimated SRM as an input to the policy counterfactual analysis considered later in the paper.
\section{Empirical Application}\label{BCIES-Section5}

\subsection{Institutional Context and Empirical Design}\label{BCIES-Section5_1}

To demonstrate the empirical relevance of the proposed framework, we investigate the effects of the Opportunity Zones (OZ) program, a major U.S. place-based tax incentive introduced by the Tax Cuts and Jobs Act of 2017. This program offers preferential tax treatment for investments in designated census tracts with the objective of stimulating local economic activity. The OZ setting is particularly well suited to our framework because designation was potentially endogenous and its effects may extend beyond designated tracts. Eligibility for OZ designation was determined primarily using pre-program socioeconomic conditions from the 2011--2015 American Community Survey (ACS). Census tracts generally qualified if their poverty rate exceeded \(20\%\) or their median family income was below \(80\%\) of the area median income. Approximately \(40\%\) of U.S. census tracts were eligible. Importantly, eligibility did not imply designation. State governors were given substantial discretion to nominate up to \(25\%\) of eligible tracts within their states, after which the nominations were certified by the U.S. Treasury. This two-stage process-rule-based eligibility followed by discretionary selection among eligible tracts creates scope for endogenous selection into OZ designation. Our outcome of interest is growth in the number of housing units at the census-tract level. Tax incentives may stimulate construction and other investment within designated tracts, but the resulting effects need not stop at tract boundaries. Designation may generate positive spillovers if investment in an OZ raises demand for development in nearby areas, or negative spillovers if investment is reallocated from neighboring tracts toward tax-advantaged locations. Existing empirical studies report mixed evidence on the effects of the OZ program \citep{corinth2024opportunity, freedman2023jue, chen2023jue, wheeler2022locally}. These features motivate an empirical specification that allows both endogenous selection into designation and spatial spillovers.

We apply the Spillover Roy model in Section \ref{BCIES-Section2} to the OZ setting by letting \(D_i = QOZ_i, Z_i = Political_i, X_i = Demographic_i,\) and \(\bar{D}_{\mathcal{N}i} = \overline{QOZ}_i\), with the disturbance vector following the finite-mixture specification in \eqref{eq:mixture-main}. The outcome \(Y_i\) is the housing-unit growth between 2017 and 2022. Individual treatment \(QOZ_i\) is an indicator variable equal to one if the tract was designated as a Qualified Opportunity Zone and zero if it was eligible but not designated, so the analysis focuses on the designation margin among OZ-eligible tracts. Neighborhood treatment \(\overline{QOZ} = \sum_{j \neq i}QOZ_i\) is the share of neighboring tracts designated as QOZs, based on a row-normalized spatial adjacency matrix \(\mathbf{W} = (w_{ij})\). OZ eligibility and designation are obtained from the Urban Institute, and tract boundaries used to construct \(\mathbf{W}\) are obtained from the U.S. Census Bureau's TIGER/Line Shapefiles. The demographic covariates include the poverty rate, median earnings, and employment rate constructed from the ACS 2013--2017 five-year estimates and enter both the selection and outcome equations to account for observed characteristics associated with OZ designation and housing development.

We use partisan alignment between a tract's state legislative representative and the governor as the excluded variable in the treatment-selection equation. Specifically, \(Political_i\) equals one if the representative of tract \(i\) in the state's lower legislative chamber and the governor belong to the same political party, and zero otherwise. Previous studies document that political alignment is associated with the likelihood of OZ designation \citep{alm2021land, frank2022determines, eldar2022does}, supporting instrument relevance. The identifying restriction is that, conditional on the included pre-treatment tract characteristics, partisan alignment affects the housing-unit growth between 2017 and 2022 only through OZ designation but does directly affects the potential outcomes. Appendix \ref{BCIES-APDX_EmpiricalApp} provides additional details on instrument construction and sensitivity analyses.

We focus on California, for which we can assemble comprehensive tract-level data on OZ designation, housing outcomes, demographic characteristics, political affiliation, and spatial linkages. The final sample contains \(3{,}699\) OZ-eligible census tracts, comprising \(727\) designated QOZs and \(2{,}972\) eligible but non-designated tracts (Non-QOZs). Supplementary Figure \ref{fig:fig-OZ-California} displays their spatial distribution illustrating the close geographic proximity of designated and non-designated tracts. Detailed variable definitions and data sources are provided in Supplementary Tables \ref{tab-definition}--\ref{tab-source}. Designated tracts are systematically more disadvantaged along several pre-treatment socioeconomic dimensions; detailed summary statistics are reported in Supplementary Table \ref{tab:tab-OZ-summarystats}. This reinforces the importance of accounting for nonrandom selection into OZ designation.

\subsection{Estimation Results}\label{BCIES-Section5_2}

\begin{table}[!ht]
\centering
\caption{\label{tab:tab-OZ-est-results}Posterior Estimates of Key Model Parameters}
\centering
\fontsize{9}{11.5}\selectfont
\begin{threeparttable}
\begin{tabular}[t]{>{\raggedright\arraybackslash}p{5.2cm}>{\centering\arraybackslash}p{2.6cm}>{\centering\arraybackslash}p{1.8cm}>{\centering\arraybackslash}p{3.4cm}}
\toprule
 & Posterior Mean & SD & 90\% Credible Interval\\
\midrule
\addlinespace[0.3em]
\multicolumn{4}{l}{\textbf{Treatment Selection}}\\
\hspace{1em}Partisan alignment ($\alpha$) & 0.162 & 0.070 & {}[0.049, 0.277]\\
\addlinespace[0.3em]
\multicolumn{4}{l}{\textbf{Neighborhood Exposure}}\\
\hspace{1em}QOZs ($\delta^{(1)}$) & 0.032 & 0.014 & {}[0.009, 0.055]\\
\hspace{1em}Non-QOZs ($\delta^{(0)}$) & 0.009 & 0.009 & {}[-0.006, 0.024]\\
\hspace{1em}$\delta^{(1)}-\delta^{(0)}$ & 0.023 & 0.016 & {}[-0.004, 0.051]\\
\addlinespace[0.3em]
\multicolumn{4}{l}{\textbf{Endogenous Selection}}\\
\hspace{1em}$\rho_{1D}$ & 0.182 & 0.134 & {}[-0.063, 0.412]\\
\hspace{1em}$\rho_{0D}$ & -0.130 & 0.052 & {}[-0.218, -0.045]\\
\addlinespace[0.3em]
\multicolumn{4}{l}{\textbf{Selection on Unobserved Gains}}\\
\hspace{1em}$\sigma_{1D}-\sigma_{0D}$ & 0.052 & 0.022 & {}[0.018, 0.092]\\
Observations & 3,699 &  & \\
\bottomrule
\end{tabular}
\begin{tablenotes}[para]
\item \textit{Notes:} Posterior means, standard deviations, and 90\% credible intervals are reported for the preferred specification with pre-treatment demographic controls. $\rho_{dD}$ denotes the correlation between the treatment-selection disturbance and the potential-outcome disturbance under treatment state $d$. $\sigma_{1D}-\sigma_{0D}$ measures selection on unobserved treatment gains. Full parameter estimates and alternative specifications are reported in Supplementary Table S.X.
\end{tablenotes}
\end{threeparttable}
\end{table}

Table \ref{tab:tab-OZ-est-results} reports posterior estimates of the key parameters from our preferred specification with pre-treatment demographic controls; full parameter estimates and alternative specifications are reported in Appendix \ref{BCIES-APDX_EmpiricalApp-estimates}. Partisan alignment is positively associated with OZ designation, with a 90\% credible interval excluding zero, supporting instrument relevance. The neighborhood-treatment coefficient is positive for QOZs but smaller and imprecisely estimated for non-QOZs. The estimated dependence between the treatment-selection and potential-outcome disturbances provides evidence of endogenous selection, while the positive estimate of \(\sigma_{1D}-\sigma_{0D}\) implies indicates selection on treatment gains. We therefore next examine how the MTE varies jointly with latent resistance and neighborhood exposure.

\begin{figure}[H]

{\centering \includegraphics[width=0.8\linewidth]{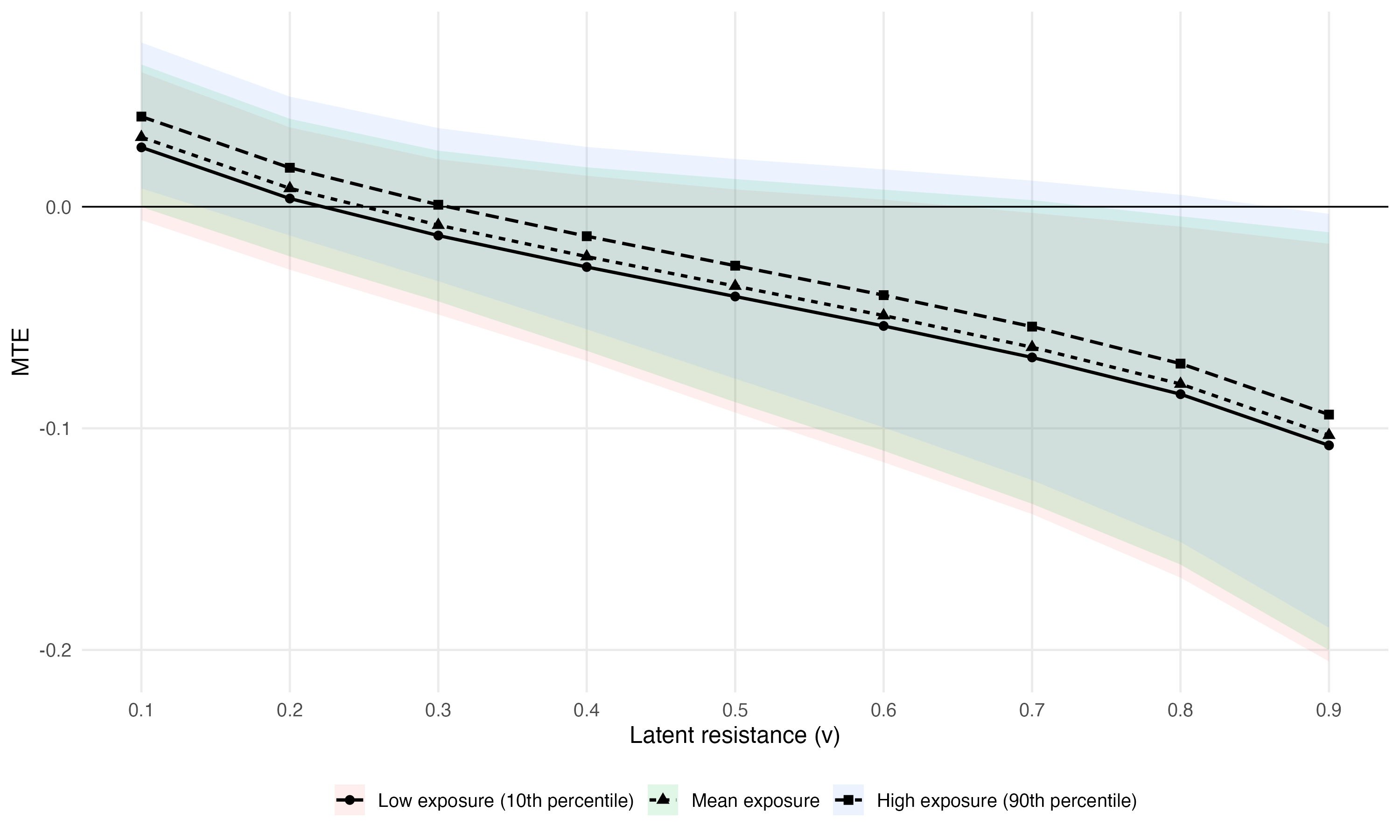} 

}

\caption[Marginal treatment effects by latent resistance and exposure.]{Marginal Treatment Effects by Latent Resistance and Exposure. The estimated MTE declines with latent resistance $v$ at all exposure levels, indicating negative selection on gains. Higher exposure shifts the MTE upward, consistent with positive spillover effects, although the shift is modest. Shaded bands show 90\% credible intervals.}\label{fig:fig-OZ-MTE-combined}
\end{figure}

\vspace{-0.75\baselineskip}

Figure \ref{fig:fig-OZ-MTE-combined} reveals substantial heterogeneity in the effect of OZ designation. The MTE declines with latent resistance \(v\), implying that tracts more likely to be designated experience larger gains, whereas effects become negative toward the upper end of the resistance distribution. Higher neighborhood OZ exposure shifts the MTE upward, although the magnitude of this exposure-related heterogeneity is modest. Thus, treatment gains vary systematically with both endogenous selection and the surrounding treatment environment. Detailed posterior estimates of the MTE are reported in Supplementary Table \ref{tab:tab-OZ-est-MTE}.

Under the realized OZ assignment, the average direct effect on treated tracts is approximately \(4.5\) percentage points, (\(90\% \text{CI}: [1.4, 7.8]\)), while neighborhood spillovers add approximately \(1.4\) percentage points, yielding a positive average total effect of \(5.9\) percentage points. By contrast, the average spillover effect on untreated tracts is small and statistically insignificant. Full estimates of average causal effects are reported in Supplementary Table \ref{tab:tab-OZ-est-CP}. Thus, the estimated gains are concentrated primarily among designated tracts, with comparatively limited spillover benefits to non-QOZs.

\subsection{Policy Counterfactual Analysis}\label{BCIES-Section5_3}

We next evaluate counterfactual expansions of OZ designation using the policy-relevant effects defined in Section \ref{BCIES-Section2}. We consider policies that increase the baseline treatment probability \(P_a\) according to
\[
P_a^{\tau} = P_a + \tau (1 - P_a), \quad \tau \in [0,1],
\]
where \(\tau\) closes a fraction of the remaining gap between the baseline treatment probability and one. Thus, larger values of \(\tau\) induce progressively broader expansions while preserving treatment probabilities within the unit interval. Because an expansion not only changes census tracts' own designation status but also neighborhood OZ exposure, its total effect reflects both direct gains for newly treated units and indirect gains arising from changes in the surrounding treatment environment.

Figure \ref{fig:fig-OZ-PRTE-shifta} reports the policy-relevant effects across counterfactual expansions, showing pronounced diminishing returns to OZ expansion. The PRDE declines steadily with \(\tau\), becoming negative under sufficiently large expansions. This pattern follows from the declining MTE profile: broader policies induce tracts farther along the latent-resistance margin, for which expected gains from designation are progressively smaller. By contrast, the PRSE increases with policy intensity as additional designations raise neighborhood OZ exposure. These spillover gains only partially offset the declining direct gains, however, so the PRTOT also falls and eventually becomes negative. Thus, the estimated benefits of expanding OZ designation depend importantly on the scale of expansion. Posterior estimates and the corresponding shares of induced tracts are reported in Supplementary Table \ref{tab:tab-OZ-PRTE-shifta}.

\begin{figure}[H]

{\centering \includegraphics[width=0.8\linewidth]{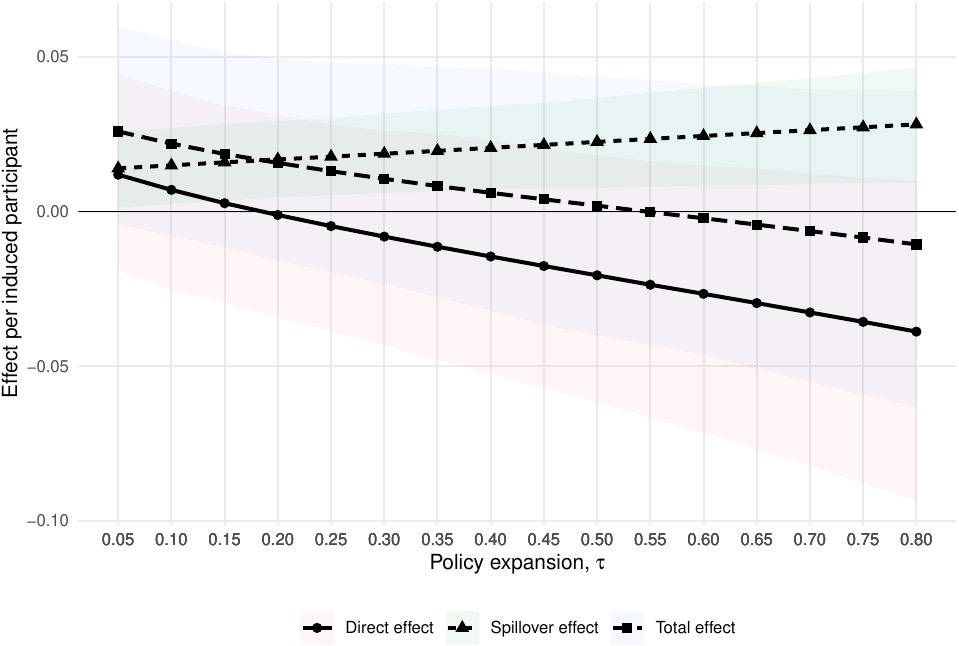} 

}

\caption{Policy-Relevant Effects under OZ Expansion. The figure reports posterior means of the policy-relevant direct (PRDE), spillover (PRSE), and total (PRTOT) effects per induced tract across counterfactual policy expansions. Shaded regions denote 90\% credible intervals.}\label{fig:fig-OZ-PRTE-shifta}
\end{figure}

\vspace{-0.75\baselineskip}

Figure \ref{fig:fig-OZ-PRSE-shifta-bygroup} further decomposes the spillover channel by treatment response type. Spillover gains are largest for policy-induced tracts and also increase for always-treated tracts as expansion raises neighborhood exposure. By contrast, estimated spillover effects for never-treated tracts remain small and imprecisely, with credible intervals including zero across the expansions considered. The estimated spillover benefits therefore accrue primarily to induced and already-treated tracts rather than broadly to tracts that remain untreated. Corresponding estimates are reported in Supplementary Table \ref{tab:tab-OZ-PRSE-shifta-bygroup}.

\begin{figure}[H]

{\centering \includegraphics[width=0.8\linewidth]{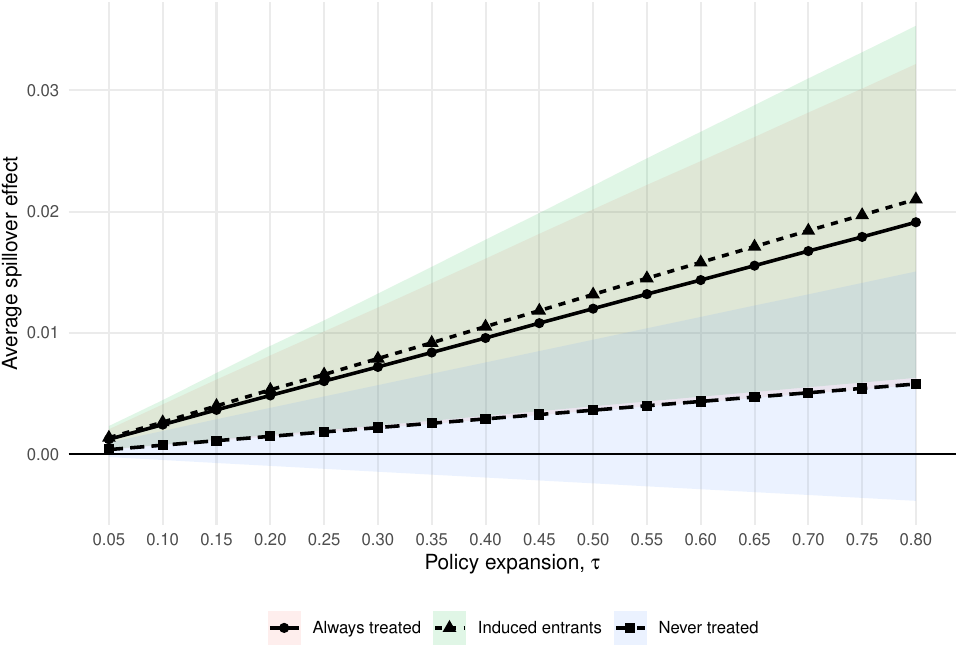} 

}

\caption{Group-Specific Spillover Effects under OZ Expansion. The figure reports posterior mean spillover effects by policy shift for always-treated units, induced entrants, and never-treated units. Shaded regions denote 90\% credible intervals.}\label{fig:fig-OZ-PRSE-shifta-bygroup}
\end{figure}

\vspace{-0.75\baselineskip}

Taken together, the counterfactual results are informative for recurring OZ designation decisions. The recent permanent extension of the program introduces new rounds of tract designation beginning in 2027, requiring states to select among eligible low-income communities. Our counterfactuals do not evaluate the new designation rules directly, but they illustrate an important trade-off relevant to such decisions. As designation expands, additional neighborhood spillovers coexist with diminishing direct gains as the policy reaches tracts with greater latent resistance; under sufficiently large expansions, the spillover gains are insufficient to offset the declining direct returns. Thus, the consequences of expanding a place-based program depend not only on how many additional areas are designated, but also on which areas are induced into treatment and how those designations alter surrounding treatment exposure.

\section{Conclusion}\label{BCIES-Section6}

This paper develops a framework for policy-relevant causal inference when treatment is endogenously selected and outcomes are subject to spillovers in a large network or spatial setting. The proposed Spillover Roy model extends the Generalized Roy framework by allowing potential outcomes to depend on both own treatment and neighborhood treatment exposure. This structure accommodates heterogeneity along the latent resistance-to-treatment margin and across exposure levels. We characterize the consequences of feasible policy changes that jointly alter treatment participation and neighborhood exposure. The resulting total policy effect decomposes into a direct effect operating through induced participation and a spillover effect operating through policy-induced changes in neighborhood exposure.

We develop a Bayesian data-augmentation approach for estimation and inference, using parameter expansion to accommodate the normalization of the latent selection equation and facilitate posterior computation. Simulations demonstrate reliable recovery of structural and heterogeneous causal effects and show that ignoring spillovers can substantially distort inference. In the application to the U.S. Opportunity Zones program, we find positive direct effects of designation on housing growth and heterogeneous treatment gains consistent with selection on gains. Spillover benefits are concentrated among designated and policy-induced tracts, whereas we find little evidence of benefits for neighboring tracts that remain untreated. Counterfactual policy experiments further indicate diminishing direct returns to program expansion, with spillover gains insufficient to offset these declines under large expansions. 

Several extensions merit further study. One is to allow treatment choices themselves to interact strategically, so that policy interventions propagate through equilibrium participation responses as well as outcome spillovers. A second direction is to relax the parametric structure and develop semiparametric or nonparametric identification and inference for policy-relevant effects under endogenous selection and interference, thereby broadening the robustness and applicability of the approach.

\clearpage

\bibliography{BCIES.bib}

\clearpage
\appendix

\renewcommand{\thesection}{S\arabic{section}}
\renewcommand{\thesubsection}{S\arabic{section}.\arabic{subsection}}
\renewcommand{\theequation}{S\arabic{equation}}
\renewcommand{\thefigure}{S\arabic{figure}}
\renewcommand{\thetable}{S\arabic{table}}
\renewcommand{\thealgocf}{S\arabic{algocf}}
\setcounter{section}{0}
\setcounter{subsection}{0}
\setcounter{equation}{0}
\setcounter{figure}{0}
\setcounter{table}{0}
\setcounter{algocf}{0}

\section*{Supplementary Appendices}\label{supplementary-appendices}
\addcontentsline{toc}{section}{Supplementary Appendices}

\section{Proofs for the Identification Results}\label{BCIES-APDX_Proofs}

\subsection{Proof of Theorem 1}\label{BCIES-APDX_ProofProp1}

We use the following auxiliary results in proving Theorem \ref{Theorem1}:

\begin{lemma}
\label{lemma:IMR}
Let $\varepsilon \sim \mathcal{N}(0,1)$ and $D = \mathbbm{1}\{\nu + \varepsilon > 0\}$. Then
$$
\mathop{\mathrm{\mathbbm{E}}}[\varepsilon \mid D=1] = \lambda(\nu) := \frac{\phi(\nu)}{\Phi(\nu)}, \quad
\mathop{\mathrm{\mathbbm{E}}}[\varepsilon \mid D=0] = -\lambda(-\nu) := -\frac{\phi(\nu)}{1-\Phi(\nu)}.
$$
where $\phi$ and $\Phi$ are the standard normal pdf and cdf.
The function $\lambda(\cdot)$ is called the inverse Mills ratio.
\end{lemma}

\begin{lemma}
\label{lemma:mix}
Under Assumption \ref{Assp4}, for each mixture component $g$,
$$
\mathop{\mathrm{\mathbbm{E}}}[\varepsilon^{(1)}_i \mid \varepsilon^{(D)}_i = \varepsilon, g] = \sigma_{1D,g}\varepsilon
\quad \text{and} \quad
\mathop{\mathrm{\mathbbm{E}}}[\varepsilon^{(0)}_i \mid \varepsilon^{(D)}_i = \varepsilon, g] = \sigma_{0D,g}\varepsilon.
$$
Since $\varepsilon^{(D)}_i$ is standard normal in every mixture component, $\Pr(g \mid \varepsilon^{(D)}_i = \varepsilon) = \pi_g$. Thus, marginalizing over $g$ using the law of total expectation,
$$
\mathop{\mathrm{\mathbbm{E}}}[\varepsilon^{(1)}_i \mid \varepsilon^{(D)}_i = \varepsilon] = \sigma_{1D}\varepsilon
\quad \text{and} \quad
\mathop{\mathrm{\mathbbm{E}}}[\varepsilon^{(0)}_i \mid \varepsilon^{(D)}_i = \varepsilon] = \sigma_{0D}\varepsilon,
$$
where
$\sigma_{1D} = \sum_{g=1}^{G} \pi_g \sigma_{1D,g}$ and
$\sigma_{0D} = \sum_{g=1}^{G} \pi_g \sigma_{0D,g}$.
\end{lemma}

\subsection*{\texorpdfstring{Proof of Theorem \ref{Theorem1}}{Proof of Theorem }}\label{proof-of-theorem}
\addcontentsline{toc}{subsection}{Proof of Theorem \ref{Theorem1}}

\subsubsection*{(i) Identification of the treatment selection equation}\label{i-identification-of-the-treatment-selection-equation}
\addcontentsline{toc}{subsubsection}{(i) Identification of the treatment selection equation}

The scale normalization \(\mathop{\mathrm{Var}}(\varepsilon^{(D)}_i) = 1\) in Assumption \ref{Assp4} removes the scale indeterminacy of the binary latent-index model. Combined with the instrument relevance condition in Assumption \ref{Assp3}(ii), the conditional treatment probability
\[
\Pr(D_i=1 \mid X_i, Z_i) = \Phi(\nu(X_i, Z_i)) = \Phi(Z_i \alpha + X_i \beta^{(D)})
\]
identifies the treatment-selection coefficients \((\alpha, \beta^{(D)})\).

It follows that the treatment-selection index \(\nu_i \coloneqq Z_i\alpha+X_i\beta^{(D)}\) and the corresponding inverse Mills ratios \(\lambda(\nu_i)\) and \(\lambda(-\nu_i)\) are identified.

\subsubsection*{(ii) Identification of the treated-regime parameters}\label{ii-identification-of-the-treated-regime-parameters}
\addcontentsline{toc}{subsubsection}{(ii) Identification of the treated-regime parameters}

For the treated regime,
\[
D_i=1
\quad\Longleftrightarrow\quad
\varepsilon_i^{(D)}>-\nu_i,
\qquad
\nu_i\coloneqq Z_i\alpha+X_i\beta^{(D)}.
\]
By Assumption \ref{Assp4},
\[
\mathop{\mathrm{\mathbbm{E}}}[\varepsilon^{(1)}_i \mid D_i=1, \bar{D}_{\mathcal{N}i}, X_i, Z_i, \mathbf{W}] 
= \mathop{\mathrm{\mathbbm{E}}}[\varepsilon^{(1)}_i \mid D_i=1, X_i, Z_i, \mathbf{W}]
\]
Using the law of iterated expectations,
\begin{align*}
\mathop{\mathrm{\mathbbm{E}}}[\varepsilon^{(1)}_i \mid D_i=1, X_i, Z_i, \mathbf{W}]
&= \mathop{\mathrm{\mathbbm{E}}}\left[\mathop{\mathrm{\mathbbm{E}}}[\varepsilon^{(1)}_i \mid \varepsilon^{(D)}_i, D_i=1, X_i, Z_i, \mathbf{W}] \mid D_i=1, X_i, Z_i, \mathbf{W}\right]\\
&= \mathop{\mathrm{\mathbbm{E}}}\left[\sigma_{1D} \varepsilon^{(D)}_i \mid D_i=1, X_i, Z_i, \mathbf{W}\right]\\
&= \sigma_{1D} \mathop{\mathrm{\mathbbm{E}}}\left[\varepsilon^{(D)}_i \mid D_i=1, X_i, Z_i, \mathbf{W}\right]\\
&= \sigma_{1D} \mathop{\mathrm{\mathbbm{E}}}\left[\varepsilon^{(D)}_i \mid \varepsilon^{(D)}_i > -\nu_i\right]\\
&= \sigma_{1D} \lambda(\nu_i),
\end{align*}
where the second equality follows from Assumption \ref{Assp3}(i) and Lemma \ref{lemma:mix}, and the final equality follows from Lemma \ref{lemma:IMR}.

Therefore,
\[
\begin{aligned}
\mathop{\mathrm{\mathbbm{E}}}[Y_i \mid D_i=1, \bar{D}_{\mathcal{N}i}, X_i, Z_i, \mathbf{W}]
&= \mathop{\mathrm{\mathbbm{E}}}[Y_i^{(1)} \mid D_i=1, \bar{D}_{\mathcal{N}i}, X_i, Z_i, \mathbf{W}] \\
&= \delta^{(1)} \bar{D}_{\mathcal{N}i} + X_i \beta^{(1)} + \mathop{\mathrm{\mathbbm{E}}}[\varepsilon^{(1)}_i \mid D_i=1, \bar{D}_{\mathcal{N}i}, X_i, Z_i, \mathbf{W}] \\
&= \delta^{(1)} \bar{D}_{\mathcal{N}i} + X_i \beta^{(1)} + \sigma_{1D}\lambda(\nu_i).
\end{aligned}
\]
Under standard full-rank conditions for \(\left(\bar D_{\mathcal N i}, X_i, \lambda(\nu_i)\right)\), the conditional mean uniquely identifies \(\delta^{(1)}, \beta^{(1)},\sigma_{1D}.\)

\subsubsection*{(iii) Identification of the untreated-regime parameters}\label{iii-identification-of-the-untreated-regime-parameters}
\addcontentsline{toc}{subsubsection}{(iii) Identification of the untreated-regime parameters}

For untreated units,
\[
\mathop{\mathrm{\mathbbm{E}}}[\epsilon_i^{(0)} \mid D_i=0, \bar{D}_{\mathcal{N}i}, X_i, Z_i, \mathbf{W}] = -\sigma_{0D}\lambda(-\nu_i),
\]
which yields
\[
\mathop{\mathrm{\mathbbm{E}}}[Y_i \mid D_i=0, \bar{D}_{\mathcal{N}i}, X_i, Z_i, \mathbf{W}] = \delta^{(0)} \bar{D}_{\mathcal{N}i} + X_i \beta^{(0)} - \sigma_{0D}\lambda(-\nu_i).
\]
Under the corresponding full-rank condition, this identifies \(\delta^{(0)}, \beta^{(0)},\sigma_{0D}.\)

This completes the proof.

\subsection{Proof of Theorem 2}\label{BCIES-APDX_ProofProp2}

Under Assumption \ref{Assp2},
\[
Y_i^{(1)}(\bar d_{\mathcal N}) 
= \delta^{(1)} \bar d_{\mathcal N} 
+ X_i \beta^{(1)} 
+ \varepsilon_i^{(1)},
\qquad
Y_i^{(0)}(\bar d_{\mathcal N}) 
= \delta^{(0)} \bar d_{\mathcal N} 
+ X_i \beta^{(0)} 
+ \varepsilon_i^{(0)} .
\]
Hence,
\[
Y_i^{(1)}(\bar d_{\mathcal N})
-
Y_i^{(0)}(\bar d_{\mathcal N})
=
(\delta^{(1)}-\delta^{(0)}) \bar d_{\mathcal N}
+
X_i(\beta^{(1)}-\beta^{(0)})
+
(\varepsilon_i^{(1)}-\varepsilon_i^{(0)}).
\]
Taking conditional expectations given \(V_i=v\) and \(X_i=x\) yields
\[
\mathop{\mathrm{MTE}}(\bar d_{\mathcal N}, v, x)
=
(\delta^{(1)}-\delta^{(0)}) \bar d_{\mathcal N}
+
x(\beta^{(1)}-\beta^{(0)})
+
\mathbb{E}
\!\left[
\varepsilon_i^{(1)}-\varepsilon_i^{(0)}
\mid
V_i=v
\right].
\]
Further, under Assumption \ref{Assp4},
\[
\varepsilon_i^{(D)} \sim \mathcal{N}(0,1),
\qquad
V_i = \Phi(-\varepsilon_i^{(D)}),
\qquad
\varepsilon_i^{(D)} = -\Phi^{-1}(V_i).
\]
Therefore,
\[
\mathbb{E}
\!\left[
\varepsilon_i^{(1)}
\mid
V_i=v
\right]
=
-
\sum_{g=1}^{G}
\pi_g
\sigma_{1D,g}
\Phi^{-1}(v),
\qquad
\mathbb{E}
\!\left[
\varepsilon_i^{(0)}
\mid
V_i=v
\right]
=
-
\sum_{g=1}^{G}
\pi_g
\sigma_{0D,g}
\Phi^{-1}(v),
\]
so that
\begin{align*}
\mathop{\mathrm{MTE}}(\bar d_{\mathcal N}, v, x)
&= (\delta^{(1)}-\delta^{(0)}) \bar d_{\mathcal N}
+ x(\beta^{(1)}-\beta^{(0)})
- \sum_{g=1}^{G}\pi_g(\sigma_{1D,g}-\sigma_{0D,g})\Phi^{-1}(v)\\
&= (\delta^{(1)}-\delta^{(0)}) \bar d_{\mathcal N}
+ x(\beta^{(1)}-\beta^{(0)})
- (\sigma_{1D}-\sigma_{0D})\Phi^{-1}(v)
\end{align*}
Since Theorem \ref{Theorem1} establishes identification of the parameters
\(\delta^{(1)}\), \(\delta^{(0)}\), \(\beta^{(1)}\), \(\beta^{(0)}\), \(\sigma_{1D}\), and \(\sigma_{0D}\), the marginal treatment effect \(\mathop{\mathrm{MTE}}(\bar d_{\mathcal N}, v, x)\) is identified.

Next, for \(d \in \{0,1\}\),
\[
\mathbb{E}
\!\left[
Y_i^{(d)}(\bar d_{\mathcal N})
\mid
V_i=v,
X_i=x
\right]
=
\delta^{(d)} \bar d_{\mathcal N}
+
x \beta^{(d)}
+
\mathbb{E}
\!\left[
\varepsilon_i^{(d)}
\mid
V_i=v
\right].
\]
Differentiating with respect to \(\bar d_{\mathcal N}\) gives
\[
\mathop{\mathrm{MSE}}(d,\bar d_{\mathcal N}, v, x)
=
\delta^{(d)}.
\]
Hence, the marginal spillover effect \(\mathop{\mathrm{MSE}}(d,\bar d_{\mathcal N}, v, x)\) is also identified.

\subsection{Proof of Theorem 3}\label{BCIES-APDX_ProofProp3}

\textbf{Direct component}

Define the indicator for policy-induced participants by
\[
S_i(a,a') \coloneqq \mathbbm{1}\!\left\{P^a(Z_i,X_i)<V_i\le P^{a'}(Z_i,X_i)\right\}.
\]
Because the same latent resistance \(V_i\) governs treatment decisions under both policies,
\[
D_i^r = \mathbbm{1}\!\left\{ P^r(Z_i,X_i)\ge V_i\right\},\qquad r\in\{a,a'\}.
\]
Pointwise policy monotonicity implies
\[
D_i^{a'}\ge D_i^a \qquad\text{and}\qquad D_i^{a'} - D_i^a = S_i(a,a').
\]
By definition,
\[
\mathop{\mathrm{PRDE}}(x;a,a')
=
\mathop{\mathrm{\mathbbm{E}}}\left[
Y_i^{(1)}(\bar D_{\mathcal Ni}^{a})
-
Y_i^{(0)}(\bar D_{\mathcal Ni}^{a})
\mid
S_i(a,a')=1,
X_i=x
\right].
\]
Equivalently,
\[
\mathop{\mathrm{PRDE}}(x;a,a') =\frac{\mathop{\mathrm{\mathbbm{E}}}\left[\left\{Y_i^{(1)}(\bar D_{\mathcal Ni}^{a})-Y_i^{(0)}(\bar D_{\mathcal Ni}^{a})\right\}S_i(a,a')\mid X_i=x\right]}{\Pr(S_i(a,a')=1\mid X_i=x)}.
\]
Since \(V_i\mid X_i=x\sim U(0,1)\) and \(V_i\perp Z_i\mid X_i\), the denominator satisfies
\begin{align*}
\Pr(S_i(a,a')=1\mid X_i=x)
&=
\mathop{\mathrm{\mathbbm{E}}}\!\left[
\Pr\!\left(
P^a(Z_i,X_i)<V_i\le P^{a'}(Z_i,X_i)
\mid Z_i,X_i=x
\right)
\mid X_i=x
\right]
\\
&=
\mathop{\mathrm{\mathbbm{E}}}\!\left[
P^{a'}(Z_i,X_i)-P^a(Z_i,X_i)
\mid X_i=x
\right]
\\
&=
\Delta P(x;a,a').
\end{align*}
Next, applying the law of iterated expectations with respect to
\(V_i\) gives
\begin{align*}
&\mathop{\mathrm{\mathbbm{E}}}\!\left[
\left\{
Y_i^{(1)}(\bar D_{\mathcal N i}^{a})
-
Y_i^{(0)}(\bar D_{\mathcal N i}^{a})
\right\}
S_i(a,a')
\mid X_i=x
\right]
\\
&\quad=
\int_0^1
\mathop{\mathrm{\mathbbm{E}}}\!\left[
Y_i^{(1)}(\bar D_{\mathcal N i}^{a})
-
Y_i^{(0)}(\bar D_{\mathcal N i}^{a})
\mid
S_i(a,a')=1,
V_i=v,
X_i=x
\right] \times
\Pr(S_i(a,a')=1\mid V_i=v,X_i=x)
\,dv.
\end{align*}

Under Assumption \ref{Assp2},
\begin{align*}
Y_i^{(1)}(\bar D_{\mathcal N i}^{a}) - Y_i^{(0)}(\bar D_{\mathcal N i}^{a}) = (\delta^{(1)}-\delta^{(0)})\bar D_{\mathcal N i}^{a} + X_i(\beta^{(1)}-\beta^{(0)})
+ \varepsilon_i^{(1)}-\varepsilon_i^{(0)}.
\end{align*}
Therefore,
\begin{align*}
&\mathop{\mathrm{\mathbbm{E}}}\!\left[Y_i^{(1)}(\bar D_{\mathcal N i}^{a}) - Y_i^{(0)}(\bar D_{\mathcal N i}^{a}) \mid S_i(a,a')=1, V_i=v, X_i=x\right]
\\
&\quad=
(\delta^{(1)}-\delta^{(0)})\mathop{\mathrm{\mathbbm{E}}}\!\left[\bar D_{\mathcal N i}^{a} \mid S_i(a,a')=1, V_i=v, X_i=x \right]
\\
&\qquad+ x(\beta^{(1)}-\beta^{(0)}) + \mathop{\mathrm{\mathbbm{E}}}\!\left[\varepsilon_i^{(1)}-\varepsilon_i^{(0)}\mid S_i(a,a')=1, V_i=v, X_i=x \right].
\end{align*}

Conditional on \((V_i,X_i)\), the event \(S_i(a,a')=1\) depends only on the excluded variation entering the policy propensity scores. By Assumption \ref{Assp3},
\[
\mathop{\mathrm{\mathbbm{E}}}\!\left[\varepsilon_i^{(1)}-\varepsilon_i^{(0)}\mid S_i(a,a')=1, V_i=v, X_i=x \right]
= \mathop{\mathrm{\mathbbm{E}}}\!\left[\varepsilon_i^{(1)}-\varepsilon_i^{(0)} \mid V_i=v\right].
\]
Using the definition
\[
\bar d_{\mathcal N}^{a,S}(v,x) = \mathop{\mathrm{\mathbbm{E}}}\!\left[\bar D_{\mathcal N i}^{a} \mid S_i(a,a')=1, V_i=v, X_i=x\right],
\]
and the affine form of the MTE established in Theorem \ref{Theorem2}, it follows that
\[
\mathop{\mathrm{\mathbbm{E}}}\!\left[Y_i^{(1)}(\bar D_{\mathcal N i}^{a}) - Y_i^{(0)}(\bar D_{\mathcal N i}^{a})\mid S_i(a,a')=1, V_i=v, X_i=x \right] = \mathop{\mathrm{MTE}}\!\left(\bar d_{\mathcal N}^{a,S}(v,x), v, x\right).
\]
Moreover, because \(V_i\perp Z_i\mid X_i\),
\begin{align*}
\Pr(S_i(a,a')=1\mid V_i=v,X_i=x) 
&= \Pr\!\left(P^a(Z_i,X_i)<v\le P^{a'}(Z_i,X_i) \mid X_i=x\right)\\
&=\Pr\!\left(P^a(Z_i,X_i)<v\mid X_i=x\right) -\Pr\!\left(P^{a'}(Z_i,X_i)<v\mid X_i=x\right)\\
&=F_{P^a\mid X}(v\mid x)-F_{P^{a'}\mid X}(v\mid x),
\end{align*}
where the second equality follows from pointwise policy monotonicity.

Combining the preceding expressions yields
\[
\mathop{\mathrm{PRDE}}(x;a,a') =
\int_0^1 \mathop{\mathrm{MTE}}\!\left(
\bar d_{\mathcal N}^{a,S}(v,x), v, x \right) h_{PR}(v\mid x;a,a')\,dv,
\]
where
\[
h_{PR}(v\mid x;a,a') =
\frac{F_{P^a\mid X}(v\mid x)-F_{P^{a'}\mid X}(v\mid x)}{\Delta P(x;a,a')}.
\]

It remains to verify that \(h_{PR}\) is a proper weight function. Indeed, pointwise policy monotonicity implies
\[
F_{P^a\mid X}(v\mid x) \ge F_{P^{a'}\mid X}(v\mid x),\]
so \(h_{PR}(v\mid x;a,a')\ge0\). Furthermore, for any random
variable \(Q\in[0,1]\),
\[
\int_0^1 F_{Q\mid X}(v\mid x)\,dv
=
1-\mathop{\mathrm{\mathbbm{E}}}[Q\mid X_i=x] \quad (\text{integration by parts}). 
\]
Applying this identity with \(Q = P^{a}(Z_i)\) and \(Q = P^{a'}(Z_i)\),
\begin{align*}
\int_0^1 h_{PR}(v\mid x;a,a')\,dv 
= \frac{\mathop{\mathrm{\mathbbm{E}}}[P^{a'}(Z_i,X_i)\mid X_i=x]-\mathop{\mathrm{\mathbbm{E}}}[P^a(Z_i,X_i)\mid X_i=x]}{\Delta P(x;a,a')}
= 1,
\end{align*}
hence the weight function integrates to one.

By Theorem \ref{Theorem2}, the marginal treatment effect is identified. The identified treatment-selection model, together with the policy-counterfactual conditions, determines the counterfactual policy propensity scores and hence the induced participation share \(\Delta P(x;a,a')\) and the policy weight function \(h_{PR}(v\mid x;a,a')\). Since the interaction matrix \(\mathbf W\) is known, the exposure mapping identifies the mean neighborhood exposure of policy-induced participants \(\bar d_{\mathcal N}^{a,S}(v,x)\). Therefore, \(\mathop{\mathrm{PRDE}}(x;a,a')\) is identified.

\medskip

\textbf{Spillover component}

By definition,
\[
\mathop{\mathrm{PRSE}}(x;a,a') = \frac{
\mathop{\mathrm{\mathbbm{E}}}\!\left[
Y_i^{(D_i^{a'})}(\bar D_{\mathcal N i}^{a'}) - Y_i^{(D_i^{a'})}(\bar D_{\mathcal N i}^{a})
\mid X_i=x \right]
}{\Delta P(x;a,a')}.
\]

Using the marginal spillover effect,
\[
\mathop{\mathrm{MSE}}(d,\bar d_{\mathcal N},v,x)
=
\frac{\partial}{\partial\bar d_{\mathcal N}}
\mathop{\mathrm{\mathbbm{E}}}\!\left[
Y_i^{(d)}(\bar d_{\mathcal N})
\mid V_i=v,X_i=x
\right],
\]
the spillover component can be represented as
\[
\mathop{\mathrm{PRSE}}(x;a,a')
=
\frac{
\mathop{\mathrm{\mathbbm{E}}}\!\left[
\int_{\bar D_{\mathcal N i}^{a}}^{\bar D_{\mathcal N i}^{a'}}
\mathop{\mathrm{MSE}}\!\left(
D_i^{a'},
\bar d_{\mathcal N},
V_i,
X_i
\right)
d\bar d_{\mathcal N}
\mid X_i=x
\right]
}{
\Delta P(x;a,a')
}.
\]

By Theorem \ref{Theorem2},
\[
\mathop{\mathrm{MSE}}(d,\bar d_{\mathcal N},v,x)
=
\delta^{(d)},
\]
which does not depend on \(\bar d_{\mathcal N}\). Therefore,
\[
\int_{\bar D_{\mathcal N i}^{a}}^{\bar D_{\mathcal N i}^{a'}}
\mathop{\mathrm{MSE}}\!\left(
D_i^{a'},
\bar d_{\mathcal N},
V_i,
X_i
\right)
d\bar d_{\mathcal N}
=
\delta^{(D_i^{a'})}
\left(
\bar D_{\mathcal N i}^{a'}
-
\bar D_{\mathcal N i}^{a}
\right),
\]
where
\[
\delta^{(D_i^{a'})}
=
D_i^{a'}\delta^{(1)}
+
(1-D_i^{a'})\delta^{(0)}.
\]
Thus,
\[
\mathop{\mathrm{PRSE}}(x;a,a')
=
\frac{
\mathop{\mathrm{\mathbbm{E}}}\!\left[
\delta^{(D_i^{a'})}
\left(
\bar D_{\mathcal N i}^{a'}
-
\bar D_{\mathcal N i}^{a}
\right)
\mid X_i=x
\right]
}{
\Delta P(x;a,a')
}.
\]
The parameters \(\delta^{(1)}\) and \(\delta^{(0)}\) are identified by Theorem \ref{Theorem1}. Under the policy-counterfactual conditions, the identified policy propensity scores determine the counterfactual
treatment assignments under policies \(a\) and \(a'\). Since the interaction matrix \(\mathbf W\) is known, the corresponding counterfactual neighborhood exposures \(\bar D_{\mathcal N i}^{a}\) and \(\bar D_{\mathcal N i}^{a'}\) are determined by the exposure mapping. Therefore, the expectation
\[
\mathop{\mathrm{\mathbbm{E}}}\!\left[
\delta^{(D_i^{a'})}
\left(
\bar D_{\mathcal N i}^{a'}
-
\bar D_{\mathcal N i}^{a}
\right)
\mid X_i=x
\right]
\]
is identified, and hence so is \(\mathop{\mathrm{PRSE}}(x;a,a')\).

\medskip

\textbf{Total Effect}

Finally, using the decomposition
\[ 
Y_i^{a'}-Y_i^a = \Big( Y_i^{(D_i^{a'})}(\bar D_{\mathcal N i}^{a'}) - Y_i^{(D_i^{a'})}(\bar D_{\mathcal N i}^{a}) \Big) + \Big( Y_i^{(D_i^{a'})}(\bar D_{\mathcal N i}^{a}) - Y_i^{(D_i^{a})}(\bar D_{\mathcal N i}^{a}) \Big), 
\]
taking conditional expectations given \(X_i = x\), and dividing both sides by \(\Delta P(x;a,a')\) yields
\[ 
\mathop{\mathrm{PRTOT}}(x;a,a') = \mathop{\mathrm{PRSE}}(x;a,a') + \mathop{\mathrm{PRDE}}(x;a,a'). 
\]
Since both \(\mathop{\mathrm{PRDE}}(x;a,a')\) and \(\mathop{\mathrm{PRSE}}(x;a,a')\) are identified, it follows that \(\mathop{\mathrm{PRTOT}}(x;a,a')\) is also identified.
\qed

\section{Implementation Details}\label{BCIES-APDX_sampling}

For completeness, this appendix provides the likelihood and full conditional distributions underlying the parameter-expanded Gibbs sampler described in Section \ref{BCIES-Section3}.

\subsection{Complete-data likelihood}\label{complete-data-likelihood}

Let
\[
\mathbf L^*
\coloneqq
\begin{bmatrix}
\mathbf D^*\\
\mathbf Y^{(1)}\\
\mathbf Y^{(0)}
\end{bmatrix},
\qquad
\mathbf R \coloneqq
\begin{bmatrix}
\mathbf P&0&0\\
0&\mathbf Q&0\\
0&0&\mathbf Q
\end{bmatrix},
\]
so that
\[
\mathbf L^*=\mathbf R\boldsymbol{\theta}
+\boldsymbol{\varepsilon}.
\]

Conditional on the mixture allocations, \(\mathbf{c}\), the error vector \(\boldsymbol{\varepsilon}\) has mean zero and covariance matrix \(\boldsymbol{\Omega}(\mathbf{c})\). Since each individual \(i\) belongs to exactly one component \(g\), i.e.~\(\sum_{g=1}^G c_{ig}=1\), the covariance matrix of \(\boldsymbol{\varepsilon}\) is a block-diagonal matrix satisfying
\begin{equation}
\label{eq:Omega_c}
\boldsymbol{\Omega}(\mathbf{c}) = \sum_{g=1}^G \text{diag}(c_{1g},\ldots,c_{ng}) \otimes \boldsymbol{\Sigma}_g.
\end{equation}
The complete-data likelihood with the augmented latent outcome \(\mathbf{L}^*\) and mixture indicators \(\mathbf{c}\) is
\[
p(\mathbf Y,\mathbf D,\mathbf L^*,\mathbf c\mid \boldsymbol{\theta},\{\mathbf\Sigma_g\},\boldsymbol{\pi})
= p(\mathbf Y,\mathbf D\mid \mathbf L^*) \cdot p(\mathbf L^* \mid \boldsymbol{\theta},\{\mathbf\Sigma_g\},\mathbf c) \cdot p(\mathbf c \mid \boldsymbol{\pi}).
\]
From noting that knowing \((Y^{(1)}_{i},Y^{(0)}_{i})\) as well as the sign of \(D_i^\ast\) perfectly predicts the value of \((Y_i,D_i)\), we can derive the conditional likelihood function as follows
\begin{equation}
\begin{split}
p(\mathbf{Y}, \mathbf{D}&\mid \theta, \mathbf{\Sigma}, \mathbf{L}^\ast) \\
&= \prod_{i=1}^n \left[\mathbbm{1}(D_i^\ast>0)\mathbbm{1}(D_i = 1)\mathbbm{1}(Y_i = Y^{(1)}_{i}) + \mathbbm{1}(D_i^*\leq0)\mathbbm{1}(D_i = 0)\mathbbm{1}(Y_i = Y^{(0)}_{i})\right].\\
&= \prod_{\{i: D_i = 1\}} p\left(Y^{(1)}_{i}, D_i^* >0\right) \prod_{\{i: D_i = 0\}} p\left(Y^{(0)}_{i}, D_i^\ast \leq 0\right) \\
&= \prod_{\{i: D_i = 1\}} \int_{0}^{\infty} p\left(Y^{(1)}_{i},D_i^\ast\right)dD_i^\ast \prod_{\{i: D_i = 0\}} \int_{-\infty}^0 p\left(Y^{(0)}_{i}, D_i^\ast\right)  dD_i^\ast\\
&= \prod_{\{i: D_i = 1\}} \int_{0}^{\infty} p\left(D_i^\ast\mid Y^{(1)}_{i}\right)p\left(Y^{(1)}_{i}\right)dD_i^\ast \prod_{\{i: D_i = 0\}} \int_{-\infty}^0 p\left(D_i^\ast\mid Y^{(0)}_{i}\right)p\left(Y^{(0)}_{i}\right)dD_i^\ast\\
&= \prod_{\{i: D_i = 1\}} \Phi \left(\frac{u_{Di} + \rho_{1D} u_{1i}}{(1-\rho_{1D}^2)^{-1/2}}\right) \frac{1}{\sigma_1} \phi(u_{1i}) \prod_{\{i: D_i = 0\}} \left [ 1 - \Phi \left(\frac{u_{Di} + \rho_{0D} u_{0i}}{(1-\rho_{0D}^2)^{-1/2}}\right) \right] \frac{1}{\sigma_0} \phi(u_{0i}),
\end{split}
\end{equation}
where
\[
u_{Di} \coloneqq \mathbf{P}^\top_i\mathbf{\gamma}; \quad  u_{1i} \coloneqq \frac{Y^{(1)}_{i}-\mathbf{Q}^\top_i \mathbf{\kappa}_1}{\sigma_1}; \quad u_{0i} \coloneqq \frac{Y^{(0)}_{i}-\mathbf{Q}^\top_i \mathbf{\kappa}_0}{\sigma_0}.
\]
We can derive the second term from the representation of latent variables in
\[
p(\mathbf L^* \mid \boldsymbol{\theta},\{\mathbf\Sigma_g\},\mathbf c) 
= \prod_{i=1}^n\prod_{g=1}^G \left[ 2\pi \left\vert\mathbf{\Sigma}\right\vert^{-1/2}\exp\left\{ -\frac{1}{2} (\mathbf{L}^* - \mathbf{R}\mathbf{\theta})^\top\mathbf{\Sigma}_g^{-1}(\mathbf{L}^* - \mathbf{R}\mathbf{\theta}) \right\} \right]^{c_{ig}}.
\]
By definition, \(c_i \mid \boldsymbol{\pi} \overset{iid}{\sim} \mathcal{M}ult(1, \boldsymbol{\pi})\), thus
\[
p(\mathbf c \mid \boldsymbol{\pi}) = \prod_{i=1}^n\prod_{g=1}^G \pi_g^{c_{ig}}.
\]

\subsection{Priors and parameter expansion}\label{BCIES-APDX_PX}

This section derives the auxiliary distribution used in the parameter-expansion covariance update. For notational simplicity, we suppress the mixture-component index \(g\); the derivation applies separately to each component.

Recall that the covariance matrix in the identified parameterization is
\[
\mathbf\Sigma
=
\begin{bmatrix}
1
&
\rho_{1D}\sigma_1
&
\rho_{0D}\sigma_0
\\
\rho_{1D}\sigma_1
&
\sigma_1^2
&
\rho_{10}\sigma_1\sigma_0
\\
\rho_{0D}\sigma_0
&
\rho_{10}\sigma_1\sigma_0
&
\sigma_0^2
\end{bmatrix},
\]
where the normalization \(\Sigma_{11}=1\) fixes the scale of the latent treatment-selection equation. Introduce the positive expansion parameter \(\tau>0\) and define
\[
\mathbf A
\coloneqq
\mathop{\mathrm{diag}}(\tau,1,1),
\qquad
\widetilde{\mathbf\Sigma}
\coloneqq
\mathbf A\mathbf\Sigma\mathbf A.
\]
Hence,
\[
\widetilde{\mathbf\Sigma}
=
\begin{bmatrix}
\tau^2
&
\tau\rho_{1D}\sigma_1
&
\tau\rho_{0D}\sigma_0
\\
\tau\rho_{1D}\sigma_1
&
\sigma_1^2
&
\rho_{10}\sigma_1\sigma_0
\\
\tau\rho_{0D}\sigma_0
&
\rho_{10}\sigma_1\sigma_0
&
\sigma_0^2
\end{bmatrix}.
\]

Let \(\Sigma_{ij}\) and \(\widetilde{\Sigma}_{ij}\) denote the \((i,j)\)th elements of \(\mathbf\Sigma\) and \(\widetilde{\mathbf\Sigma}\), respectively. Because \(\Sigma_{11}=1\), the transformation from
\[
(\Sigma_{12},\Sigma_{13},\Sigma_{22},
\Sigma_{23},\Sigma_{33},\tau^2)
\]
to
\[
(\widetilde{\Sigma}_{11},\widetilde{\Sigma}_{12},
\widetilde{\Sigma}_{13},\widetilde{\Sigma}_{22},
\widetilde{\Sigma}_{23},\widetilde{\Sigma}_{33})
\]
has Jacobian matrix
\[
\mathcal J
=
\frac{
\partial(
\widetilde{\Sigma}_{11},
\widetilde{\Sigma}_{12},
\widetilde{\Sigma}_{13},
\widetilde{\Sigma}_{22},
\widetilde{\Sigma}_{23},
\widetilde{\Sigma}_{33})
}{
\partial(
\Sigma_{12},
\Sigma_{13},
\Sigma_{22},
\Sigma_{23},
\Sigma_{33},
\tau^2)
}
=
\begin{bmatrix}
0&0&0&0&0&1\\
\tau&0&0&0&0&\Sigma_{12}/(2\tau)\\
0&\tau&0&0&0&\Sigma_{13}/(2\tau)\\
0&0&1&0&0&0\\
0&0&0&1&0&0\\
0&0&0&0&1&0
\end{bmatrix}.
\]
Therefore,
\begin{equation}
\label{eq:jacobian-PX}
|\mathcal J|
=
\tau^2.
\end{equation}

We assign the expanded covariance matrix the inverse-Wishart prior
\begin{equation}
\label{eq:IW-PX-apdx}
\widetilde{\mathbf\Sigma}
\sim
\mathcal W^{-1}(\mathbf I_3,\nu_0),
\end{equation}
with density
\[
p(\widetilde{\mathbf\Sigma})
\propto
|\widetilde{\mathbf\Sigma}|^{-(\nu_0+4)/2}
\exp
\left\{
-\frac12
\operatorname{tr}
(\widetilde{\mathbf\Sigma}^{-1})
\right\}.
\]

Because
\[
|\widetilde{\mathbf\Sigma}|
=
|\mathbf A|^2|\mathbf\Sigma|
=
\tau^2|\mathbf\Sigma|,
\]
and
\[
\widetilde{\mathbf\Sigma}^{-1}
=
\mathbf A^{-1}
\mathbf\Sigma^{-1}
\mathbf A^{-1},
\]
the induced joint density of
\((\mathbf\Sigma,\tau^2)\) follows from the change-of-variables formula
\begin{equation}
\label{eq:joint-prior-PX-1}
p(\mathbf\Sigma,\tau^2)
=
p(\widetilde{\mathbf\Sigma})|\mathcal J| \propto\;
(\tau^2)^{-\left(\nu_0/2+1\right)}
|\mathbf\Sigma|^{-(\nu_0+4)/2} 
\exp
\left\{
-\frac12\mathop{\mathrm{tr}}(\mathbf A^{-1}\mathbf\Sigma^{-1}\mathbf A^{-1})
\right\}.
\end{equation}

To express this density compactly, define
\begin{equation}
\label{eq:Delta-PX}
\Delta
\coloneqq
1
+2\rho_{10}\rho_{1D}\rho_{0D}
-\rho_{10}^2
-\rho_{1D}^2
-\rho_{0D}^2.
\end{equation}
Positive definiteness of \(\mathbf\Sigma\) requires \(\Delta>0\). The relevant diagonal elements of \(\mathbf\Sigma^{-1}\) are
\begin{align*}
(\mathbf\Sigma^{-1})_{11}
&=
\frac{1-\rho_{10}^2}{\Delta},
\\
(\mathbf\Sigma^{-1})_{22}
&=
\frac{1-\rho_{0D}^2}
{\sigma_1^2\Delta},
\\
(\mathbf\Sigma^{-1})_{33}
&=
\frac{1-\rho_{1D}^2}
{\sigma_0^2\Delta}.
\end{align*}
Consequently,
\[
\mathop{\mathrm{tr}}(\mathbf A^{-1}\mathbf\Sigma^{-1}\mathbf A^{-1})
=
\frac{1}{\tau^2}
\frac{1-\rho_{10}^2}{\Delta}
+
\frac{1-\rho_{0D}^2}{\sigma_1^2\Delta}
+
\frac{1-\rho_{1D}^2}{\sigma_0^2\Delta}.
\]

Combining these expressions with
\eqref{eq:joint-prior-PX-1}, the induced joint prior is
\begin{equation}
\label{eq:joint-prior-PX-2}
\begin{split}
p(\mathbf\Sigma,\tau^2)
\propto\;&
(\tau^2)^{-\left(\nu_0/2+1\right)}
|\mathbf\Sigma|^{-(\nu_0+4)/2}
\\
&\times
\exp
\left\{
-\frac{1}{2\tau^2}
\frac{1-\rho_{10}^2}{\Delta}
\right\}\times
\exp
\left\{
-\frac12
\left[
\frac{1-\rho_{0D}^2}{\sigma_1^2\Delta}
+
\frac{1-\rho_{1D}^2}{\sigma_0^2\Delta}
\right]
\right\}.
\end{split}
\end{equation}

It follows from \eqref{eq:joint-prior-PX-2} that, conditional on
\(\mathbf\Sigma\),
\begin{equation}
\label{APDX-tausq}
\tau^2\mid\mathbf\Sigma
\sim
\mathcal{I}\mathcal{G}
\left(
\frac{\nu_0}{2},
\frac{1-\rho_{10}^2}{2\Delta}
\right),
\end{equation}
where the inverse-Gamma distribution is parameterized by shape and scale. Equivalently,
\begin{equation}
\label{APDX-tausq-chisq}
\tau^2\mid\mathbf\Sigma
\overset{d}{=}
\frac{1-\rho_{10}^2}
{\Delta\,\chi_{\nu_0}^2}.
\end{equation}

For completeness, integrating \(\tau^2\) out of
\eqref{eq:joint-prior-PX-2} yields the induced marginal prior
\begin{equation}
\label{APDX-Sigma}
\begin{split}
p(\mathbf\Sigma)
\propto\;&
|\mathbf\Sigma|^{-(\nu_0+4)/2}
\left(
\frac{1-\rho_{10}^2}{\Delta}
\right)^{-\nu_0/2} 
\times
\exp
\left\{
-\frac12
\left[
\frac{1-\rho_{0D}^2}{\sigma_1^2\Delta}
+
\frac{1-\rho_{1D}^2}{\sigma_0^2\Delta}
\right]
\right\},
\end{split}
\end{equation}
over the positive-definite covariance matrices satisfying
\(\Sigma_{11}=1\).

Equation \eqref{APDX-tausq} provides the auxiliary distribution used in the parameter-expansion step. For each mixture component \(g\), an auxiliary draw of \(\tau_g^2\) is used only to construct the transformed residual cross-product matrix entering the expanded covariance update. After drawing \(\widetilde{\mathbf\Sigma}_g\), the normalization scale is reset as
\[
\tau_g^2
=
\widetilde{\Sigma}_{g,11},
\]
and \(\mathbf\Sigma_g\) is recovered by rescaling. Thus, the auxiliary draw used to form the expanded update and the scale implied by the subsequent covariance draw play distinct sequential roles in the parameter-expanded Gibbs sampler.

\subsection{Full conditional distributions}\label{BCIES-APDX_MCMC}

Posterior simulation alternates between imputing the latent quantities and updating the model parameters conditional on the resulting complete data.

\subsubsection{Latent-data augmentation}\label{latent-data-augmentation}

(a1) Missing potential outcome \(\mathbf{Y}^{miss}\)
\begin{equation}
\label{missingY}
Y_i^{miss} \mid \Theta_{-Y_i^{miss}}, \mathbf{Y},\mathbf{D} \overset{ind}{\sim} \mathcal{N} \left( (1-D_i)\mu_{1i,g} + D_i\mu_{0i,g}, (1-D_i)V_{1i,g} + D_iV_{0i,g}  \right),
\end{equation}

\vspace{-2mm}

where \(g\) is the current component assgined to \(i\), and
\begin{align*}
\mu_{1i,g} &\coloneqq \mathbf{Q}^\top_i \mathbf{\kappa}_1 + (D_i^* - \mathbf{P}^\top_i\mathbf{\gamma}) \left[ \frac{\sigma_{0,g}^2\sigma_{1D,g} - \sigma_{10,g}\sigma_{0D,g} }{\sigma_{0,g}^2 - \sigma_{0D,g}^2} \right] + (Y_i - \mathbf{Q}^\top_i \mathbf{\kappa}_0) \left[ \frac{\sigma_{10,g} - \sigma_{0D,g}\sigma_{1D,g} }{\sigma_{0,g}^2 - \sigma_{0D,g}^2} \right],\\
\mu_{0i,g} &\coloneqq \mathbf{Q}^\top_i \mathbf{\kappa}_0 + (D_i^* - \mathbf{P}^\top_i\mathbf{\gamma}) \left[ \frac{\sigma_{1,g}^2\sigma_{0D,g} - \sigma_{10,g}\sigma_{1D,g} }{\sigma_{1,g}^2 - \sigma_{1D,g}^2} \right] + (Y_i - \mathbf{Q}^\top_i \mathbf{\kappa}_1) \left[ \frac{\sigma_{10,g} - \sigma_{0D,g}\sigma_{1D,g} }{\sigma_{1,g}^2 - \sigma_{1D,g}^2} \right],\\
V_{1i,g} &\coloneqq \sigma_{1,g}^2 - \frac{\sigma_{1D,g}^2 \sigma_{0,g}^2 - 2\sigma_{10,g}\sigma_{0D,g}\sigma_{1D,g} + \sigma_{10,g}^2}{\sigma_{0,g}^2 - \sigma_{0D,g}^2},\\
V_{0i,g} &\coloneqq \sigma_{0,g}^2 - \frac{\sigma_{0D,g}^2 \sigma_{1,g}^2 - 2\sigma_{10,g}\sigma_{0D,g}\sigma_{1D,g} + \sigma_{10,g}^2}{\sigma_{1,g}^2 - \sigma_{1D,g}^2}.
\end{align*}

(a2) Latent utility \(\mathbf{D}^*\)
\begin{equation}
\label{latentD}
D_i^* \mid \Theta_{-D_i^*}, \mathbf{Y},\mathbf{D} \overset{ind}{\sim} \left\{\begin{array}{c}
T\mathcal{N}_{(0,+\infty)} (\mu_{Di,g}, V_{Di,g}) \quad \text{if } D_i = 1
\\
T\mathcal{N}_{(-\infty,0]} (\mu_{Di,g}, V_{Di,g}) \quad \text{if } D_i = 0
\end{array}\right.,
\end{equation}
where \(g\) is the current component assgined to \(i\), and
\begin{align*}
\mu_{Di,g} &\coloneqq \mathbf{P}^\top_i\mathbf{\gamma} + \left[D_i Y_i + (1-D_i)Y_i^{miss} - \mathbf{Q}^\top_i \mathbf{\kappa}_1 \right] \left[ \frac{\sigma_{0,g}^2 \sigma_{1D,g} - \sigma_{10,g} \sigma_{0D,g}}{\sigma_{1,g}^2\sigma_{0,g}^2 -\sigma_{10,g}^2}\right],
\\
&\quad + \left[D_i Y_i^{miss} + (1-D_i)Y_i - \mathbf{Q}^\top_i \mathbf{\kappa}_0 \right] \left[ \frac{\sigma_{1,g}^2 \sigma_{0D,g} - \sigma_{10,g} \sigma_{1D,g}}{\sigma_{1,g}^2\sigma_{0,g}^2 -\sigma_{10,g}^2}\right],
\\
V_{Di,g} &\coloneqq 1 - \frac{\sigma_{1D,g}^2\sigma_{0,g}^2 - 2\sigma_{10,g}\sigma_{0D,g}\sigma_{1D,g} + \sigma_{1,g}^2\sigma_{0D,g}^2}{\sigma_{1,g}^2\sigma_{0,g}^2 - \sigma_{10,g}^2}.
\end{align*}

(a3) Component indicators \(\mathbf{c}\)
\begin{align}
\label{componentid}
c_i \mid \Theta_{-c_i}, \mathbf{Y}, \mathbf{D} &\overset{ind}{\sim} \mathcal{M}ult \left(1,\left[\omega_{i1},\ldots,\omega_{iG} \right]^\top\right),
\end{align}
where for \(i = 1,\ldots,N\) and \(g = 1,\ldots,G\)
\begin{align*}
\omega_{ig} \coloneqq \frac{\pi_g  \phi(\mathbf{L}_i^*; \mathbf{R}_i\mathbf{\theta}, \mathbf{\Sigma}_g)}{\sum_{h=1}^{G}\pi_h\phi(\mathbf{L}_i^\ast; \mathbf{R}_i\mathbf{\theta}, \mathbf{\Sigma}_h)}.
\end{align*}

\subsubsection{Parameter updates}\label{parameter-updates}

After imputing the missing data, we can infer the posterior distribution of remaining parameters, conditioning on the complete data. Given the component indicators \(\mathbf{c}\), we denote \(\mathcal{I}_g\) as the set of observations belonging to component \(g\) and \(n_g\) as the corresponding cardinality:
\[
\mathcal{I}_g\coloneqq \{i:c_{ig}=1\}, \quad n_g \coloneqq \vert\mathcal{I}_g\vert = \sum_{i=1}^n \mathbbm{1}(c_{ig}=1).
\]

(b1) Regression coefficients \(\mathbf{\theta}\)
\begin{equation}
\label{posttheta}
\mathbf{\theta} \mid \Theta_{-\mathbf{\theta}}, \mathbf{Y},\mathbf{D} \sim \mathcal{N}(\bar{\mathbf{\mu}}_\theta, \overline{\mathbf{V}}_\theta),
\end{equation}
where
\[
\bar{\mathbf{\mu}}_\theta \coloneqq \overline{\mathbf{V}}_\theta [\mathbf{R}^\top\boldsymbol{\Omega}(\mathbf{c})^{-1}\mathbf{L}^* + \underline{\mathbf{V}}_{\theta}^{-1}\underline{\mathbf{\mu}}_{\theta}, \quad
\overline{\mathbf{V}}_\theta \coloneqq [\mathbf{R}^\top\boldsymbol{\Omega}(\mathbf{c})^{-1}\mathbf{R} + \underline{\mathbf{V}}_{\theta}^{-1}]^{-1}.
\]

(b2) The component weights \(\boldsymbol{\pi}\)
\begin{align}
\boldsymbol{\pi} \mid \Theta_{-\mathbf{\pi}}, \mathbf{Y}, \mathbf{D} &\sim \mathcal{D}ir(\underline{\omega}_1 + n_1, \ldots, \underline{\omega}_G + n_G),
\end{align}

(b3) Update component-specific covariance matrices \(\{\boldsymbol{\Sigma}_g\}\) via parameter expansion.

For each component \(g = 1,\ldots,G\), we update the component-specific covariance matrix using only those observations currently assigned to component \(g\). For \(i\in \mathcal{I}_g\), collect the corresponding latent and observed quantities into component-specific subvectors and submatrices:
\[
\mathbf D_g^*= \{D_i^*\}_{i\in\mathcal I_g},
\qquad \mathbf Y_g^{(1)}=\{Y_i^{(1)}\}_{i\in\mathcal I_g},
\qquad \mathbf Y_g^{(0)}=\{Y_i^{(0)}\}_{i\in\mathcal I_g},
\]
\[
\mathbf P_g=\{\mathbf P_i^\top\}_{i\in\mathcal I_g},
\qquad \mathbf Q_g=\{\mathbf Q_i^\top\}_{i\in\mathcal I_g}.
\]
Define the corresponding component-specific residual vectors
\[
  \boldsymbol{\varepsilon}_{D,g}=\mathbf D_g^*-\mathbf P_g\boldsymbol{\gamma},
  \qquad
  \boldsymbol{\varepsilon}_{1,g}=\mathbf Y_g^{(1)}-\mathbf Q_g\boldsymbol{\kappa}_1,
  \qquad
  \boldsymbol{\varepsilon}_{0,g}=\mathbf Y_g^{(0)}-\mathbf Q_g\boldsymbol{\kappa}_0.
\]
Utilizing the parameter expansion, let \(\tau_g>0\) denote the component-specific expansion parameter randomly draw from its prior according to \eqref{APDX-tausq-chisq}. The associated transformed residual cross-product matrix is
\[
\widetilde{\mathbf{M}}_g = \begin{bmatrix}
\tau_g^2\mathbf{\varepsilon}_{D,g}^\top\mathbf{\varepsilon}_{D,g} & \tau_g\mathbf{\varepsilon}_{D,g}^\top\mathbf{\varepsilon}_{1,g}   & \tau_g\mathbf{\varepsilon}_{D,g}^\top\mathbf{\varepsilon}_{0,g} \\
\tau_g\mathbf{\varepsilon}_{1,g}^\top\mathbf{\varepsilon}_{D,g}   & \mathbf{\varepsilon}_{1,g}^\top\mathbf{\varepsilon}_{1,g} & \mathbf{\varepsilon}_{1,g}^\top\mathbf{\varepsilon}_{0,g} \\
\tau_g\mathbf{\varepsilon}_{0,g}^\top\mathbf{\varepsilon}_{D,g}   & \mathbf{\varepsilon}_{0,g}^\top\mathbf{\varepsilon}_{1,g}   & \mathbf{\varepsilon}_{0,g}^\top\mathbf{\varepsilon}_{0,g}
\end{bmatrix}
\]

Given the inverse-Wishart prior \(\widetilde{\mathbf\Sigma}_g\sim \mathcal W^{-1}(\mathbf{I}_3,\underline{\nu})\), the conditional posterior distribution of the unconstrained covariance matrix \(\widetilde{\boldsymbol{\Sigma}}_g\) is of the form
\begin{align}
\widetilde{\boldsymbol{\Sigma}}_g \mid \mathbf{\theta},\mathbf{L}^*,\mathbf{c},\mathbf{Y},\mathbf{D} \sim \mathcal{W}^{-1}(\widetilde{\mathbf{M}}_g + \mathbf{I}_3, n_g+\underline{\nu}).
\end{align}

After drawing \(\widetilde{\mathbf{\Sigma}}_g\), we set \(\tau_g^2 = \widetilde{\mathbf{\Sigma}}_{g,11}\) and recover the normalized covariance matrix \(\mathbf{\Sigma}_g\) as
\begin{equation}
\label{recoverSigma}
\mathbf{\Sigma}_g = 
\begin{bmatrix}
1/\tau_g &0 &0
\\
0 &1 &0
\\
0 &0 &1
\end{bmatrix}
\times \widetilde{\mathbf{\Sigma}}_g \times
\begin{bmatrix}
1/\tau_g &0 &0
\\
0 &1 &0
\\
0 &0 &1
\end{bmatrix}.
\end{equation}

This covariance update is performed separately for each \(g=1,\ldots,G\).

Algorithm \ref{alg:GGRM_mixture_long} collects these conditional updates in implementation order. It is the detailed counterpart of the parameter-expanded Gibbs sampler summarized in the main text. Posterior inference for the causal estimands is obtained by evaluating each estimand at every retained posterior draw.

\medskip

\begin{spacing}{1}
\begin{algorithm*}[H]
\caption{Implementation of the parameter-expanded Gibbs sampler}
\label{alg:GGRM_mixture_long}
Initialize: $\boldsymbol{\beta}^{[0]}, \{\boldsymbol{\Sigma}_g^{[0]}\}_{g=1}^G, \boldsymbol{\pi}^{[0]}, \boldsymbol{c}^{[0]}$, set $s = 0$.

\While{$s < S$} {
\medskip
  \begin{itemize}
    \item[a.] \textbf{Latent-data augmentation}
    \begin{itemize}
      \item[(a1)] Sample $Y_i^{\text{miss}} \mid \mathbin{\vcenter{\hbox{\scalebox{.5}{$\bullet$}}}}\sim \mathcal{N}(\overline{\mu}_{Y_i},\overline{V}_{Y_i})$;
      \item[(a2)] Sample $D_i^{*} \mid \mathbin{\vcenter{\hbox{\scalebox{.5}{$\bullet$}}}}\sim T\mathcal{N}_{\mathcal{S}}(\overline{\mu}_{D_i},\overline{V}_{D_i})$;
      \item[(a3)] Sample the mixture indicator  $c_i \mid \mathbin{\vcenter{\hbox{\scalebox{.5}{$\bullet$}}}}\sim \mathcal{M}ult \left(1,\left[\omega_{i1},\ldots,\omega_{iG} \right]^\top\right)$.
    \end{itemize}
    \item[b.] \normalsize\textbf{Parameter updates}
    \begin{itemize}
      \item[(b1)] Sample $\boldsymbol{\beta}\mid \mathbin{\vcenter{\hbox{\scalebox{.5}{$\bullet$}}}}\sim \mathcal{N}(\overline{\boldsymbol{\mu}}_\beta,\overline{\mathbf{V}}_\beta)$;
      \item[(b2)] Sample mixture weights $\boldsymbol{\pi}\mid \mathbin{\vcenter{\hbox{\scalebox{.5}{$\bullet$}}}}\sim \mathcal{D}ir(\underline{\omega}_1 + n_1, \ldots, \underline{\omega}_G + n_G)$;
      \item[(b3)] For each $g$:
      \begin{itemize}\small
          \item[$\bullet$] sample expansion scale $\tau_g^2$ from auxiliary prior,
          \item[$\bullet$] sample $\widetilde{\boldsymbol{\Sigma}}_g$ given transformed residuals in component $g$,
          \item[$\bullet$] set $\boldsymbol{\Sigma}_g$ using  $\widetilde{\boldsymbol{\Sigma}}_g$ and $\tau_g^2$.
      \end{itemize}
    \end{itemize}
  \end{itemize}
  $s \leftarrow s+1$
}
\Return retained posterior draws $\{\boldsymbol{\beta}^{[s]}, \boldsymbol{\pi}^{[s]},\{\boldsymbol{\Sigma}_g^{[s]}\}_{g=1}^G\}_{s=nb+1}^S$.
\end{algorithm*}
\end{spacing}

\section{Additional Simulation Results}\label{BCIES-APDX_Simulations}

\subsection{Complete MTE Results}\label{complete-mte-results}

\vspace{-0.25\baselineskip}

\begin{table}[H]
\centering
\caption{\label{tab:tab-SMCresults-Mar26-DGP2-MTE1s}Simulation Results for Marginal Treatment Effects at Low Exposure}
\centering
\resizebox{\ifdim\width>\linewidth\linewidth\else\width\fi}{!}{
\fontsize{8}{10}\selectfont
\begin{threeparttable}
\begin{tabular}[t]{ccccccccccc}
\toprule
\multicolumn{2}{c}{ } & \multicolumn{3}{c}{n = 500} & \multicolumn{3}{c}{n = 1000} & \multicolumn{3}{c}{n = 2000} \\
\cmidrule(l{3pt}r{3pt}){3-5} \cmidrule(l{3pt}r{3pt}){6-8} \cmidrule(l{3pt}r{3pt}){9-11}
Model & Grid & Bias & RMSE & Coverage & Bias & RMSE & Coverage & Bias & RMSE & Coverage\\
\midrule
 & 0.1 & 0.372 & 0.429 & 0.603 & 0.386 & 0.416 & 0.267 & 0.390 & 0.406 & 0.044\\

 & 0.2 & 0.384 & 0.422 & 0.396 & 0.392 & 0.411 & 0.101 & 0.395 & 0.405 & 0.001\\

 & 0.3 & 0.394 & 0.421 & 0.233 & 0.397 & 0.410 & 0.028 & 0.398 & 0.405 & 0.000\\

 & 0.4 & 0.402 & 0.423 & 0.140 & 0.401 & 0.411 & 0.008 & 0.401 & 0.407 & 0.000\\

 & 0.5 & 0.409 & 0.428 & 0.103 & 0.405 & 0.414 & 0.006 & 0.404 & 0.408 & 0.000\\

 & 0.6 & 0.416 & 0.434 & 0.104 & 0.409 & 0.417 & 0.005 & 0.407 & 0.411 & 0.000\\

 & 0.7 & 0.424 & 0.444 & 0.132 & 0.413 & 0.422 & 0.011 & 0.410 & 0.415 & 0.000\\

 & 0.8 & 0.434 & 0.458 & 0.236 & 0.418 & 0.430 & 0.032 & 0.413 & 0.420 & 0.001\\

\multirow{-9}{*}{\centering\arraybackslash NSRM} & 0.9 & 0.446 & 0.482 & 0.452 & 0.424 & 0.443 & 0.154 & 0.418 & 0.428 & 0.012\\
\cmidrule{1-11}
 & 0.1 & -0.036 & 0.237 & 0.961 & -0.024 & 0.176 & 0.958 & -0.008 & 0.124 & 0.951\\

 & 0.2 & -0.018 & 0.203 & 0.964 & -0.014 & 0.149 & 0.960 & -0.002 & 0.105 & 0.955\\

 & 0.3 & -0.005 & 0.185 & 0.962 & -0.006 & 0.134 & 0.960 & 0.002 & 0.094 & 0.958\\

 & 0.4 & 0.005 & 0.174 & 0.967 & 0.000 & 0.124 & 0.959 & 0.006 & 0.088 & 0.963\\

 & 0.5 & 0.016 & 0.168 & 0.963 & 0.006 & 0.118 & 0.961 & 0.009 & 0.085 & 0.964\\

 & 0.6 & 0.026 & 0.168 & 0.969 & 0.012 & 0.117 & 0.966 & 0.013 & 0.084 & 0.961\\

 & 0.7 & 0.037 & 0.174 & 0.969 & 0.019 & 0.120 & 0.969 & 0.016 & 0.087 & 0.964\\

 & 0.8 & 0.050 & 0.188 & 0.975 & 0.026 & 0.129 & 0.971 & 0.021 & 0.094 & 0.961\\

\multirow{-9}{*}{\centering\arraybackslash SRM} & 0.9 & 0.068 & 0.217 & 0.974 & 0.037 & 0.149 & 0.969 & 0.026 & 0.110 & 0.953\\
\bottomrule
\end{tabular}
\begin{tablenotes}[para]
\item \scriptsize \textit{Notes:} This table displays the average bias (Bias), the Root Mean Squared Error (RMSE), and the coverage rate (Coverage) across $R=1000$ replicates. The rows contain results for models with/without spatial interference and for various sample size $n$.
\end{tablenotes}
\end{threeparttable}}
\end{table}

\vspace{-0.75\baselineskip}

\begin{table}[H]
\centering
\caption{\label{tab:tab-SMCresults-Mar26-DGP2-MTE2s}Simulation Results for Marginal Treatment Effects at Medium Exposure}
\centering
\resizebox{\ifdim\width>\linewidth\linewidth\else\width\fi}{!}{
\fontsize{8}{10}\selectfont
\begin{threeparttable}
\begin{tabular}[t]{ccccccccccc}
\toprule
\multicolumn{2}{c}{ } & \multicolumn{3}{c}{n = 500} & \multicolumn{3}{c}{n = 1000} & \multicolumn{3}{c}{n = 2000} \\
\cmidrule(l{3pt}r{3pt}){3-5} \cmidrule(l{3pt}r{3pt}){6-8} \cmidrule(l{3pt}r{3pt}){9-11}
Model & Grid & Bias & RMSE & Coverage & Bias & RMSE & Coverage & Bias & RMSE & Coverage\\
\midrule
 & 0.1 & -0.028 & 0.217 & 0.958 & -0.014 & 0.155 & 0.948 & -0.010 & 0.113 & 0.947\\

 & 0.2 & -0.016 & 0.175 & 0.944 & -0.008 & 0.123 & 0.943 & -0.005 & 0.089 & 0.944\\

 & 0.3 & -0.006 & 0.150 & 0.947 & -0.003 & 0.104 & 0.945 & -0.002 & 0.075 & 0.942\\

 & 0.4 & 0.002 & 0.134 & 0.941 & 0.001 & 0.091 & 0.947 & 0.001 & 0.066 & 0.946\\

 & 0.5 & 0.009 & 0.126 & 0.946 & 0.005 & 0.084 & 0.955 & 0.004 & 0.061 & 0.952\\

 & 0.6 & 0.016 & 0.125 & 0.954 & 0.009 & 0.083 & 0.965 & 0.007 & 0.060 & 0.953\\

 & 0.7 & 0.024 & 0.133 & 0.963 & 0.013 & 0.089 & 0.961 & 0.010 & 0.065 & 0.959\\

 & 0.8 & 0.034 & 0.151 & 0.970 & 0.018 & 0.103 & 0.971 & 0.013 & 0.076 & 0.960\\

\multirow{-9}{*}{\centering\arraybackslash NSRM} & 0.9 & 0.046 & 0.187 & 0.974 & 0.024 & 0.131 & 0.970 & 0.018 & 0.097 & 0.959\\
\cmidrule{1-11}
 & 0.1 & -0.034 & 0.204 & 0.963 & -0.021 & 0.148 & 0.951 & -0.010 & 0.108 & 0.949\\

 & 0.2 & -0.016 & 0.163 & 0.964 & -0.011 & 0.117 & 0.953 & -0.004 & 0.085 & 0.953\\

 & 0.3 & -0.003 & 0.138 & 0.961 & -0.003 & 0.097 & 0.954 & 0.000 & 0.071 & 0.956\\

 & 0.4 & 0.008 & 0.122 & 0.963 & 0.003 & 0.084 & 0.962 & 0.004 & 0.061 & 0.957\\

 & 0.5 & 0.018 & 0.114 & 0.958 & 0.009 & 0.076 & 0.961 & 0.007 & 0.055 & 0.960\\

 & 0.6 & 0.028 & 0.113 & 0.959 & 0.015 & 0.075 & 0.963 & 0.010 & 0.053 & 0.970\\

 & 0.7 & 0.039 & 0.121 & 0.964 & 0.022 & 0.080 & 0.965 & 0.014 & 0.057 & 0.968\\

 & 0.8 & 0.052 & 0.139 & 0.970 & 0.029 & 0.094 & 0.964 & 0.018 & 0.066 & 0.973\\

\multirow{-9}{*}{\centering\arraybackslash SRM} & 0.9 & 0.070 & 0.175 & 0.968 & 0.040 & 0.121 & 0.962 & 0.024 & 0.086 & 0.961\\
\bottomrule
\end{tabular}
\begin{tablenotes}[para]
\item \scriptsize \textit{Notes:} This table displays the average bias (Bias), the Root Mean Squared Error (RMSE), and the coverage rate (Coverage) across $R=1000$ replicates. The rows contain results for models with/without spatial interference and for various sample size $n$.
\end{tablenotes}
\end{threeparttable}}
\end{table}

\vspace{-0.75\baselineskip}

\begin{table}[H]
\centering
\caption{\label{tab:tab-SMCresults-Mar26-DGP2-MTE3s}Simulation Results for Marginal Treatment Effects at High Exposure}
\centering
\fontsize{8}{10}\selectfont
\begin{threeparttable}
\begin{tabular}[t]{ccccccccccc}
\toprule
\multicolumn{2}{c}{ } & \multicolumn{3}{c}{n = 500} & \multicolumn{3}{c}{n = 1000} & \multicolumn{3}{c}{n = 2000} \\
\cmidrule(l{3pt}r{3pt}){3-5} \cmidrule(l{3pt}r{3pt}){6-8} \cmidrule(l{3pt}r{3pt}){9-11}
Model & Grid & Bias & RMSE & Coverage & Bias & RMSE & Coverage & Bias & RMSE & Coverage\\
\midrule
 & 0.1 & -0.428 & 0.479 & 0.514 & -0.414 & 0.442 & 0.276 & -0.410 & 0.425 & 0.053\\

 & 0.2 & -0.416 & 0.450 & 0.357 & -0.408 & 0.426 & 0.123 & -0.405 & 0.415 & 0.007\\

 & 0.3 & -0.406 & 0.433 & 0.229 & -0.403 & 0.416 & 0.042 & -0.402 & 0.408 & 0.001\\

 & 0.4 & -0.398 & 0.420 & 0.159 & -0.399 & 0.409 & 0.014 & -0.399 & 0.404 & 0.000\\

 & 0.5 & -0.391 & 0.410 & 0.117 & -0.395 & 0.404 & 0.004 & -0.396 & 0.401 & 0.000\\

 & 0.6 & -0.384 & 0.403 & 0.140 & -0.391 & 0.400 & 0.004 & -0.393 & 0.398 & 0.000\\

 & 0.7 & -0.376 & 0.398 & 0.220 & -0.387 & 0.397 & 0.013 & -0.390 & 0.396 & 0.000\\

 & 0.8 & -0.366 & 0.395 & 0.381 & -0.382 & 0.396 & 0.069 & -0.387 & 0.394 & 0.000\\

\multirow{-9}{*}{\centering\arraybackslash NSRM} & 0.9 & -0.354 & 0.398 & 0.610 & -0.376 & 0.397 & 0.237 & -0.382 & 0.394 & 0.030\\
\cmidrule{1-11}
 & 0.1 & -0.032 & 0.247 & 0.955 & -0.018 & 0.175 & 0.960 & -0.012 & 0.129 & 0.931\\

 & 0.2 & -0.014 & 0.214 & 0.954 & -0.008 & 0.150 & 0.953 & -0.006 & 0.109 & 0.938\\

 & 0.3 & -0.001 & 0.195 & 0.950 & 0.000 & 0.136 & 0.950 & -0.002 & 0.098 & 0.942\\

 & 0.4 & 0.010 & 0.183 & 0.954 & 0.006 & 0.127 & 0.949 & 0.001 & 0.090 & 0.945\\

 & 0.5 & 0.020 & 0.177 & 0.948 & 0.012 & 0.122 & 0.954 & 0.005 & 0.086 & 0.951\\

 & 0.6 & 0.030 & 0.176 & 0.955 & 0.018 & 0.122 & 0.955 & 0.008 & 0.084 & 0.954\\

 & 0.7 & 0.041 & 0.181 & 0.956 & 0.025 & 0.126 & 0.957 & 0.012 & 0.086 & 0.954\\

 & 0.8 & 0.054 & 0.193 & 0.955 & 0.032 & 0.135 & 0.956 & 0.016 & 0.092 & 0.959\\

\multirow{-9}{*}{\centering\arraybackslash SRM} & 0.9 & 0.072 & 0.220 & 0.959 & 0.043 & 0.156 & 0.960 & 0.022 & 0.106 & 0.961\\
\bottomrule
\end{tabular}
\begin{tablenotes}[para]
\item \scriptsize \textit{Notes:} This table displays the average bias (Bias), the Root Mean Squared Error (RMSE), and the coverage rate (Coverage) across $R=1000$ replicates. The rows contain results for models with/without spatial interference and for various sample size $n$.
\end{tablenotes}
\end{threeparttable}
\end{table}

\vspace{-0.25\baselineskip}

Tables \ref{tab:tab-SMCresults-Mar26-DGP2-MTE1s}---\ref{tab:tab-SMCresults-Mar26-DGP2-MTE3s} report the complete Monte Carlo results for the marginal treatment effect over the full resistance grid \(v\in \{0.1,\ldots,0.9\}\) at low, medium, and high neighborhood exposure, respectively. These tables reinforce the representative results reported in the main text. Under the correctly specified SRM, MTE bias remains small over most of the resistance and exposure space, RMSE decreases as the sample size increases, and empirical coverage of the 95\% posterior credible intervals is generally close to the nominal level. The largest finite-sample deviations occur near the boundaries of the resistance distribution, where estimation is naturally less precise. In contrast, the NSRM exhibits persistent misspecification bias whenever neighborhood exposure differs materially from the level at which the omitted spillover component has little effect. The distortion is especially pronounced at low and high exposure. In these regions, increasing \(n\) reduces sampling uncertainty around an incorrectly specified estimand, causing empirical coverage to deteriorate rather than improve. At medium exposure, the NSRM performs comparatively well in this particular design because the omitted spillover component produces little discrepancy at that exposure level. Thus, satisfactory performance at a single exposure value should not be interpreted as robustness to interference.

\subsection{Additional Robustness Designs}\label{additional-robustness-designs}

We consider three additional designs that depart from the baseline simulation. The first removes spillovers altogether, the second introduces latent heterogeneity through a finite mixture of Gaussian disturbances, and the third replaces Gaussian disturbances with a heavy-tailed multivariate Student-\(t\) distribution. These exercises assess whether the proposed estimator remains well behaved when the spillover channel is inactive and when the latent disturbance distribution is more heterogeneous or heavy-tailed than in the baseline design.

\subsubsection{Scenario (i): No Spillovers}\label{scenario-i-no-spillovers}

We first suppress neighborhood spillovers by setting the exposure parameters to zero
\[
\delta^{(1)} = \delta^{(0)} = 0,
\]
while leaving all other elements of the baseline data-generating process in subsection \ref{BCIES-Section4_1} unchanged. Under this design, neighborhood exposure has no effect on either potential outcome, meaning interference channel is inactive.

The corresponding Monte Carlo results in Table \ref{tab:tab-SMCresults-Mar26-DGP1-SI} show that the SRM continues to recover the principal structural parameters accurately even though the true spillover coefficients are zero. Bias is generally small, RMSE decreases with \(n\), and empirical coverage remains close to the nominal level for the principal parameters. In particular, estimates of \(\delta^{(1)}\) and \(\delta^{(0)}\) concentrate around zero as the sample size increases.

The MTE results in \ref{tab:tab-SMCresults-Mar26-DGP1-MTEs} exhibit a similar pattern. Because spillovers are absent, the true MTE does not vary with neighborhood exposure, and the estimated MTE surfaces remain close to their true values across the exposure grid. Thus, allowing for spillovers in estimation does not by itself generate economically meaningful spillover effects when the data-generating process contains none. This exercise shows that the more general SRM remains well behaved when the no-spillover model is nested within it.

\subsection{Scenario (ii). Finite-Mixture Latent Heterogeneity}\label{scenario-ii.-finite-mixture-latent-heterogeneity}

The second robustness design introduces latent heterogeneity through a two-component mixture of multivariate Gaussian disturbances
\[
\begin{aligned}
\epsilon_i = \begin{bmatrix}
\epsilon_i^{(D)} , \epsilon_i^{(1)} , \epsilon_i^{(0)}
\end{bmatrix}^\top &\overset{ind}{\sim} \frac{1}{3} \mathcal{N}\left(0, \Sigma_1 \right) + \frac{2}{3} \mathcal{N}\left(0, \Sigma_2 \right),\\ 
\end{aligned}
\]
where the two covariance matrices differ substantially in scale, creating heterogeneous dispersion across individuals
\[
\Sigma_1 = \begin{bmatrix}  
1.00 & 0.85 & 0.20 \\
& 1.20& 0.40 \\
& & 1.00
\end{bmatrix} \quad \text{and} \quad \Sigma_2 = \begin{bmatrix}  
1.00 & 0.25 & 0.75 \\
& 0.90& 0.30 \\
& & 1.10
\end{bmatrix}.
\]
This specification captures environments in which unobserved productivity or selection costs vary across latent subpopulations--for example, high- and low-ability individuals responding differently to the same policy incentives.

We estimate the SRM using the corresponding two-component mixture specification, \(G=2\). In Table \ref{tab:tab-SMCresults-Mar26-DGP3-SI}, the spillover coefficients \(\delta^{(1)}\) and \(\delta^{(0)}\), together with their difference, exhibit small bias across sample sizes, while RMSE generally declines as \(n\) increases. Coverage for most identified structural parameters remains reasonably close to the nominal level.

The MTE results in Table \ref{tab:tab-SMCresults-Mar26-DGP3-MTEs} show that the heterogeneous treatment-effect surface is also recovered well over the resistance and exposure grids. Bias is small relative to the magnitude of the true effects, precision improves with the sample size, and coverage is generally close to 95\%. These results indicate that the proposed mixture specification can accommodate substantial latent heterogeneity without materially degrading estimation of the principal causal objects.

\subsubsection{Scenario (iii). Heavy-Tailed Non-Normal Disturbances}\label{scenario-iii.-heavy-tailed-non-normal-disturbances}

The third robustness design introduces heavy-tailed non-normality by drawing disturbances from a multivariate Student-\(t\) distribution,
\[
\epsilon_i \sim t_\nu(0,\Sigma), \quad \nu = 5.
\]
This specification generates occasional extreme realizations of the latent selection and outcome shocks while preserving the same dependence structure as in the baseline Gaussian design.

Estimation continues to use the mixture-normal SRM. This experiment therefore evaluates performance under distributional misspecification rather than simply under a richer correctly specified mixture model. Overall, the estimator remains reasonably stable under heavy-tailed disturbances. As reported in Table \ref{tab:tab-SMCresults-Mar26-DGP4-SI}, bias for the principal spillover and treatment-response parameters remains limited, and RMSE declines with the sample size. Coverage is somewhat less accurate than in the baseline Gaussian design for some parameters, reflecting the more demanding distributional environment.

The MTE results in Table \ref{tab:tab-SMCresults-Mar26-DGP4-MTEs} show a similar pattern. Estimated effects remain close to the true MTE surface across low, medium, and high exposure, although finite-sample uncertainty is larger than under Gaussian disturbances. The deterioration is therefore primarily one of precision and interval calibration rather than a systematic failure to recover the heterogeneous treatment-effect pattern.

Across all robustness designed, the proposed framework consistently recovers the structure of heterogeneous treatment effects and spillover responses, supporting the reliability of subsequent estimation of policy-relevant effects.

\newpage

\begin{table}[H]
\centering
\caption{\label{tab:tab-SMCresults-Mar26-DGP1-SI}Simulation Results for Scenario (i)}
\centering
\resizebox{\ifdim\width>\linewidth\linewidth\else\width\fi}{!}{
\begin{threeparttable}
\begin{tabular}[t]{lcccccccccccccl}
\toprule
\multicolumn{2}{c}{ } & \multicolumn{4}{c}{Quantities of Interest} & \multicolumn{9}{c}{Other Parameters} \\
\cmidrule(l{3pt}r{3pt}){3-6} \cmidrule(l{3pt}r{3pt}){7-15}
Metric & n & $\delta^{(1)}$ & $\delta^{(0)}$ & $\delta^{(1)}-\delta^{(0)}$ & $\sigma_{1D}-\sigma_{0D}$ & $\alpha$ & $\beta^{(D)}$ & $\beta_1^{(1)}$ & $\beta_1^{(0)}$ & $\sigma_1^2$ & $\sigma_0^2$ & $\rho_{1D}$ & $\rho_{0D}$ & $\rho_{10}$\\
\midrule
 & True Value & 0.000 & 0.000 & 0.000 & 0.200 & 1.500 & 0.000 & 2.000 & 1.000 & 1.000 & 1.000 & 0.900 & 0.700 & 0.600\\
\cmidrule{1-15}
 & 500 & 0.005 & -0.002 & 0.007 & -0.041 & 0.064 & -0.020 & 0.013 & -0.001 & 0.010 & 0.007 & -0.046 & -0.006 & 0.105\\

 & 1000 & 0.008 & 0.000 & 0.008 & -0.024 & 0.029 & -0.012 & 0.005 & 0.000 & 0.002 & 0.002 & -0.024 & -0.001 & 0.110\\

\multirow{-3}{*}[1\dimexpr\aboverulesep+\belowrulesep+\cmidrulewidth]{\raggedright\arraybackslash Bias} & 2000 & -0.006 & -0.001 & -0.005 & -0.014 & 0.017 & -0.005 & 0.010 & 0.000 & 0.000 & 0.002 & -0.012 & 0.000 & 0.107\\
\cmidrule{1-15}
 & 500 & 0.219 & 0.240 & 0.325 & 0.119 & 0.142 & 0.076 & 0.130 & 0.148 & 0.105 & 0.098 & 0.061 & 0.082 & 0.126\\

 & 1000 & 0.152 & 0.176 & 0.234 & 0.087 & 0.085 & 0.054 & 0.089 & 0.108 & 0.067 & 0.070 & 0.035 & 0.058 & 0.123\\

\multirow{-3}{*}[1\dimexpr\aboverulesep+\belowrulesep+\cmidrulewidth]{\raggedright\arraybackslash RMSE} & 2000 & 0.108 & 0.119 & 0.163 & 0.063 & 0.059 & 0.037 & 0.064 & 0.072 & 0.048 & 0.049 & 0.023 & 0.042 & 0.117\\
\cmidrule{1-15}
 & 500 & 0.965 & 0.957 & 0.952 & 0.960 & 0.893 & 0.929 & 0.965 & 0.956 & 0.969 & 0.957 & 0.919 & 0.960 & 0.979\\

 & 1000 & 0.942 & 0.941 & 0.952 & 0.948 & 0.920 & 0.942 & 0.954 & 0.936 & 0.977 & 0.953 & 0.934 & 0.961 & 0.963\\

\multirow{-3}{*}[1\dimexpr\aboverulesep+\belowrulesep+\cmidrulewidth]{\raggedright\arraybackslash Coverage} & 2000 & 0.943 & 0.951 & 0.952 & 0.949 & 0.922 & 0.943 & 0.957 & 0.953 & 0.982 & 0.954 & 0.952 & 0.961 & 0.967\\
\bottomrule
\end{tabular}
\begin{tablenotes}[para]
\item \textit{Notes:} This table displays the average bias (Bias), the Root Mean Squared Error (RMSE), and the coverage rate (Coverage) across $R=1000$ replicates. The rows contain results for various sample size $N$.
\end{tablenotes}
\end{threeparttable}}
\end{table}

\begin{table}[!h]
\centering
\caption{\label{tab:tab-SMCresults-Mar26-DGP1-MTEs}Simulation Results for Marginal Treatment Effects}
\centering
\resizebox{\ifdim\width>\linewidth\linewidth\else\width\fi}{!}{
\fontsize{8}{10}\selectfont
\begin{threeparttable}
\begin{tabular}[t]{cccccccccccc}
\toprule
\multicolumn{3}{c}{ } & \multicolumn{3}{c}{n = 500} & \multicolumn{3}{c}{n = 1000} & \multicolumn{3}{c}{n = 2000} \\
\cmidrule(l{3pt}r{3pt}){4-6} \cmidrule(l{3pt}r{3pt}){7-9} \cmidrule(l{3pt}r{3pt}){10-12}
Exposure & Grid $v$ & True Value & Bias & RMSE & Coverage & Bias & RMSE & Coverage & Bias & RMSE & Coverage\\
\midrule
 & 0.1 & 1.256 & -0.038 & 0.237 & 0.968 & -0.024 & 0.176 & 0.954 & -0.009 & 0.124 & 0.948\\

 & 0.2 & 1.168 & -0.020 & 0.204 & 0.964 & -0.014 & 0.150 & 0.959 & -0.002 & 0.105 & 0.954\\

 & 0.3 & 1.105 & -0.007 & 0.185 & 0.964 & -0.006 & 0.134 & 0.958 & 0.002 & 0.094 & 0.958\\

 & 0.4 & 1.051 & 0.004 & 0.174 & 0.969 & 0.000 & 0.125 & 0.956 & 0.006 & 0.088 & 0.960\\

 & 0.5 & 1.000 & 0.015 & 0.168 & 0.964 & 0.006 & 0.119 & 0.958 & 0.009 & 0.085 & 0.962\\

 & 0.6 & 0.949 & 0.025 & 0.168 & 0.965 & 0.012 & 0.117 & 0.962 & 0.013 & 0.084 & 0.961\\

 & 0.7 & 0.895 & 0.036 & 0.174 & 0.970 & 0.018 & 0.120 & 0.968 & 0.016 & 0.087 & 0.962\\

 & 0.8 & 0.832 & 0.050 & 0.188 & 0.967 & 0.026 & 0.129 & 0.969 & 0.020 & 0.095 & 0.959\\

\multirow{-9}{*}{\centering\arraybackslash Low} & 0.9 & 0.744 & 0.068 & 0.217 & 0.968 & 0.036 & 0.149 & 0.967 & 0.027 & 0.110 & 0.953\\
\cmidrule{1-12}
 & 0.1 & 1.256 & -0.035 & 0.204 & 0.969 & -0.021 & 0.149 & 0.950 & -0.010 & 0.108 & 0.947\\

 & 0.2 & 1.168 & -0.017 & 0.163 & 0.967 & -0.011 & 0.117 & 0.953 & -0.004 & 0.085 & 0.953\\

 & 0.3 & 1.105 & -0.004 & 0.138 & 0.963 & -0.003 & 0.098 & 0.952 & 0.000 & 0.071 & 0.960\\

 & 0.4 & 1.051 & 0.007 & 0.122 & 0.965 & 0.003 & 0.084 & 0.961 & 0.004 & 0.061 & 0.957\\

 & 0.5 & 1.000 & 0.018 & 0.114 & 0.956 & 0.009 & 0.077 & 0.961 & 0.007 & 0.055 & 0.964\\

 & 0.6 & 0.949 & 0.028 & 0.113 & 0.960 & 0.015 & 0.075 & 0.961 & 0.011 & 0.053 & 0.964\\

 & 0.7 & 0.895 & 0.039 & 0.121 & 0.958 & 0.022 & 0.080 & 0.957 & 0.014 & 0.057 & 0.968\\

 & 0.8 & 0.832 & 0.052 & 0.140 & 0.965 & 0.029 & 0.093 & 0.961 & 0.019 & 0.066 & 0.968\\

\multirow{-9}{*}{\centering\arraybackslash Medium} & 0.9 & 0.744 & 0.070 & 0.176 & 0.968 & 0.040 & 0.120 & 0.966 & 0.025 & 0.086 & 0.959\\
\cmidrule{1-12}
 & 0.1 & 1.256 & -0.032 & 0.247 & 0.955 & -0.018 & 0.175 & 0.958 & -0.012 & 0.129 & 0.941\\

 & 0.2 & 1.168 & -0.014 & 0.214 & 0.955 & -0.008 & 0.150 & 0.958 & -0.006 & 0.109 & 0.940\\

 & 0.3 & 1.105 & -0.001 & 0.195 & 0.954 & 0.000 & 0.136 & 0.949 & -0.002 & 0.098 & 0.941\\

 & 0.4 & 1.051 & 0.010 & 0.183 & 0.952 & 0.006 & 0.127 & 0.947 & 0.002 & 0.090 & 0.944\\

 & 0.5 & 1.000 & 0.020 & 0.177 & 0.949 & 0.012 & 0.123 & 0.951 & 0.005 & 0.085 & 0.952\\

 & 0.6 & 0.949 & 0.031 & 0.176 & 0.951 & 0.018 & 0.122 & 0.954 & 0.009 & 0.084 & 0.956\\

 & 0.7 & 0.895 & 0.042 & 0.181 & 0.957 & 0.025 & 0.126 & 0.957 & 0.012 & 0.085 & 0.960\\

 & 0.8 & 0.832 & 0.055 & 0.193 & 0.961 & 0.032 & 0.135 & 0.956 & 0.017 & 0.092 & 0.960\\

\multirow{-9}{*}{\centering\arraybackslash High} & 0.9 & 0.744 & 0.073 & 0.220 & 0.958 & 0.043 & 0.156 & 0.958 & 0.023 & 0.106 & 0.963\\
\bottomrule
\end{tabular}
\begin{tablenotes}[para]
\item \textit{Notes:} This table displays the average bias (Bias), the Root Mean Squared Error (RMSE), and the coverage rate (Coverage) across $R=1000$ replicates.
\end{tablenotes}
\end{threeparttable}}
\end{table}

\newpage

\begin{table}[H]
\centering
\caption{\label{tab:tab-SMCresults-Mar26-DGP3-SI}Simulation Results for Scenario (ii)}
\centering
\resizebox{\ifdim\width>\linewidth\linewidth\else\width\fi}{!}{
\begin{threeparttable}
\begin{tabular}[t]{lcccccccccccccl}
\toprule
\multicolumn{2}{c}{ } & \multicolumn{4}{c}{Quantities of Interest} & \multicolumn{9}{c}{Other Parameters} \\
\cmidrule(l{3pt}r{3pt}){3-6} \cmidrule(l{3pt}r{3pt}){7-15}
Metric & n & $\delta^{(1)}$ & $\delta^{(0)}$ & $\delta^{(1)}-\delta^{(0)}$ & $\sigma_{1D}-\sigma_{0D}$ & $\alpha$ & $\beta^{(D)}$ & $\beta_1^{(1)}$ & $\beta_1^{(0)}$ & $\sigma_1^2$ & $\sigma_0^2$ & $\rho_{1D}$ & $\rho_{0D}$ & $\rho_{10}$\\
\midrule
 & True Value & 1.500 & 0.500 & 1.000 & 0.200 & 1.500 & 0.000 & 2.000 & 1.000 & 1.000 & 1.000 & 0.900 & 0.700 & 0.600\\
\cmidrule{1-15}
 & 500 & -0.001 & 0.010 & -0.011 & 0.001 & 0.056 & 0.001 & 0.006 & -0.012 & 0.014 & 0.016 & -0.027 & -0.028 & 0.214\\

 & 1000 & -0.001 & 0.014 & -0.015 & 0.008 & 0.030 & 0.000 & 0.005 & -0.016 & 0.003 & 0.005 & -0.016 & -0.024 & 0.274\\

\multirow{-3}{*}[1\dimexpr\aboverulesep+\belowrulesep+\cmidrulewidth]{\raggedright\arraybackslash Bias} & 2000 & 0.001 & 0.003 & -0.002 & -0.003 & 0.018 & 0.002 & 0.010 & -0.007 & -0.003 & 0.000 & -0.019 & -0.016 & 0.300\\
\cmidrule{1-15}
 & 500 & 0.264 & 0.264 & 0.366 & 0.190 & 0.150 & 0.080 & 0.157 & 0.157 & 0.101 & 0.107 & 0.132 & 0.122 & 0.257\\

 & 1000 & 0.182 & 0.189 & 0.260 & 0.133 & 0.096 & 0.052 & 0.111 & 0.117 & 0.067 & 0.071 & 0.093 & 0.088 & 0.296\\

\multirow{-3}{*}[1\dimexpr\aboverulesep+\belowrulesep+\cmidrulewidth]{\raggedright\arraybackslash RMSE} & 2000 & 0.136 & 0.133 & 0.194 & 0.098 & 0.066 & 0.038 & 0.082 & 0.080 & 0.049 & 0.052 & 0.069 & 0.060 & 0.327\\
\cmidrule{1-15}
 & 500 & 0.957 & 0.952 & 0.958 & 0.950 & 0.904 & 0.934 & 0.965 & 0.963 & 0.951 & 0.962 & 0.966 & 0.953 & 0.992\\

 & 1000 & 0.957 & 0.956 & 0.955 & 0.944 & 0.929 & 0.959 & 0.963 & 0.938 & 0.954 & 0.951 & 0.949 & 0.938 & 0.956\\

\multirow{-3}{*}[1\dimexpr\aboverulesep+\belowrulesep+\cmidrulewidth]{\raggedright\arraybackslash Coverage} & 2000 & 0.938 & 0.948 & 0.941 & 0.937 & 0.941 & 0.954 & 0.940 & 0.950 & 0.943 & 0.960 & 0.940 & 0.944 & 0.825\\
\bottomrule
\end{tabular}
\begin{tablenotes}[para]
\item \textit{Notes:} This table displays the average bias (Bias), the Root Mean Squared Error (RMSE), and the coverage rate (Coverage) across $R=1000$ replicates. The rows contain results for various sample size $n$.
\end{tablenotes}
\end{threeparttable}}
\end{table}

\begin{table}[!h]
\centering
\caption{\label{tab:tab-SMCresults-Mar26-DGP3-MTEs}Simulation Results for Marginal Treatment Effects}
\centering
\resizebox{\ifdim\width>\linewidth\linewidth\else\width\fi}{!}{
\fontsize{8}{10}\selectfont
\begin{threeparttable}
\begin{tabular}[t]{cccccccccccc}
\toprule
\multicolumn{3}{c}{ } & \multicolumn{3}{c}{n = 500} & \multicolumn{3}{c}{n = 1000} & \multicolumn{3}{c}{n = 2000} \\
\cmidrule(l{3pt}r{3pt}){4-6} \cmidrule(l{3pt}r{3pt}){7-9} \cmidrule(l{3pt}r{3pt}){10-12}
Exposure & Grid $v$ & True Value & Bias & RMSE & Coverage & Bias & RMSE & Coverage & Bias & RMSE & Coverage\\
\midrule
 & 0.1 & 0.951 & 0.017 & 0.306 & 0.964 & 0.029 & 0.216 & 0.956 & 0.013 & 0.159 & 0.929\\

 & 0.2 & 1.002 & 0.017 & 0.248 & 0.966 & 0.026 & 0.176 & 0.957 & 0.014 & 0.129 & 0.943\\

 & 0.3 & 1.039 & 0.017 & 0.216 & 0.966 & 0.024 & 0.154 & 0.952 & 0.015 & 0.112 & 0.942\\

 & 0.4 & 1.070 & 0.016 & 0.200 & 0.965 & 0.021 & 0.143 & 0.952 & 0.016 & 0.104 & 0.943\\

 & 0.5 & 1.100 & 0.016 & 0.196 & 0.967 & 0.019 & 0.141 & 0.953 & 0.017 & 0.102 & 0.942\\

 & 0.6 & 1.130 & 0.016 & 0.203 & 0.965 & 0.017 & 0.146 & 0.948 & 0.017 & 0.105 & 0.947\\

 & 0.7 & 1.161 & 0.016 & 0.223 & 0.965 & 0.015 & 0.160 & 0.944 & 0.018 & 0.115 & 0.947\\

 & 0.8 & 1.198 & 0.015 & 0.257 & 0.964 & 0.013 & 0.184 & 0.936 & 0.019 & 0.133 & 0.948\\

\multirow{-9}{*}{\centering\arraybackslash Low} & 0.9 & 1.250 & 0.015 & 0.318 & 0.960 & 0.009 & 0.226 & 0.943 & 0.020 & 0.164 & 0.942\\
\cmidrule{1-12}
 & 0.1 & 1.350 & 0.013 & 0.274 & 0.953 & 0.024 & 0.190 & 0.954 & 0.012 & 0.140 & 0.944\\

 & 0.2 & 1.402 & 0.013 & 0.205 & 0.952 & 0.020 & 0.142 & 0.955 & 0.013 & 0.104 & 0.942\\

 & 0.3 & 1.439 & 0.012 & 0.165 & 0.955 & 0.018 & 0.113 & 0.955 & 0.014 & 0.082 & 0.953\\

 & 0.4 & 1.470 & 0.012 & 0.141 & 0.957 & 0.015 & 0.096 & 0.960 & 0.015 & 0.070 & 0.954\\

 & 0.5 & 1.500 & 0.012 & 0.134 & 0.957 & 0.013 & 0.092 & 0.959 & 0.016 & 0.066 & 0.955\\

 & 0.6 & 1.530 & 0.012 & 0.144 & 0.955 & 0.011 & 0.099 & 0.955 & 0.017 & 0.071 & 0.953\\

 & 0.7 & 1.561 & 0.011 & 0.169 & 0.962 & 0.009 & 0.118 & 0.949 & 0.017 & 0.085 & 0.950\\

 & 0.8 & 1.598 & 0.011 & 0.211 & 0.959 & 0.007 & 0.148 & 0.949 & 0.018 & 0.107 & 0.948\\

\multirow{-9}{*}{\centering\arraybackslash Medium} & 0.9 & 1.650 & 0.011 & 0.281 & 0.956 & 0.003 & 0.198 & 0.950 & 0.020 & 0.144 & 0.939\\
\cmidrule{1-12}
 & 0.1 & 1.750 & 0.009 & 0.315 & 0.945 & 0.018 & 0.218 & 0.945 & 0.011 & 0.160 & 0.938\\

 & 0.2 & 1.802 & 0.008 & 0.257 & 0.949 & 0.014 & 0.176 & 0.950 & 0.012 & 0.130 & 0.936\\

 & 0.3 & 1.839 & 0.008 & 0.224 & 0.954 & 0.012 & 0.153 & 0.956 & 0.013 & 0.113 & 0.940\\

 & 0.4 & 1.870 & 0.007 & 0.207 & 0.955 & 0.009 & 0.141 & 0.958 & 0.014 & 0.104 & 0.939\\

 & 0.5 & 1.900 & 0.007 & 0.202 & 0.956 & 0.007 & 0.137 & 0.953 & 0.015 & 0.102 & 0.936\\

 & 0.6 & 1.930 & 0.007 & 0.207 & 0.960 & 0.005 & 0.142 & 0.954 & 0.016 & 0.105 & 0.946\\

 & 0.7 & 1.961 & 0.007 & 0.225 & 0.956 & 0.003 & 0.155 & 0.955 & 0.016 & 0.114 & 0.943\\

 & 0.8 & 1.998 & 0.006 & 0.257 & 0.957 & 0.001 & 0.178 & 0.955 & 0.018 & 0.132 & 0.935\\

\multirow{-9}{*}{\centering\arraybackslash High} & 0.9 & 2.050 & 0.006 & 0.316 & 0.950 & -0.003 & 0.220 & 0.952 & 0.019 & 0.163 & 0.936\\
\bottomrule
\end{tabular}
\begin{tablenotes}[para]
\item \textit{Notes:} This table displays the average bias (Bias), the Root Mean Squared Error (RMSE), and the coverage rate (Coverage) across $R=1000$ replicates.
\end{tablenotes}
\end{threeparttable}}
\end{table}

\newpage

\begin{table}[H]
\centering
\caption{\label{tab:tab-SMCresults-Mar26-DGP4-SI}Simulation Results for Scenario (iii)}
\centering
\resizebox{\ifdim\width>\linewidth\linewidth\else\width\fi}{!}{
\begin{threeparttable}
\begin{tabular}[t]{lcccccccccccccl}
\toprule
\multicolumn{2}{c}{ } & \multicolumn{4}{c}{Quantities of Interest} & \multicolumn{9}{c}{Other Parameters} \\
\cmidrule(l{3pt}r{3pt}){3-6} \cmidrule(l{3pt}r{3pt}){7-15}
Metric & n & $\delta^{(1)}$ & $\delta^{(0)}$ & $\delta^{(1)}-\delta^{(0)}$ & $\sigma_{1D}-\sigma_{0D}$ & $\alpha$ & $\beta^{(D)}$ & $\beta_1^{(1)}$ & $\beta_1^{(0)}$ & $\sigma_1^2$ & $\sigma_0^2$ & $\rho_{1D}$ & $\rho_{0D}$ & $\rho_{10}$\\
\midrule
 & True Value & 1.500 & 0.500 & 1.000 & 0.200 & 1.500 & 0.000 & 2.000 & 1.000 & 1.000 & 1.067 & 0.450 & 0.549 & 0.323\\
\cmidrule{1-15}
 & 500 & 0.003 & -0.002 & 0.006 & 0.027 & -0.262 & -0.020 & -0.031 & 0.035 & 0.739 & 0.713 & -0.025 & 0.004 & 0.143\\

 & 1000 & 0.001 & -0.001 & 0.002 & 0.038 & -0.288 & -0.010 & -0.034 & 0.036 & 0.712 & 0.706 & -0.007 & 0.010 & 0.146\\

\multirow{-3}{*}[1\dimexpr\aboverulesep+\belowrulesep+\cmidrulewidth]{\raggedright\arraybackslash Bias} & 2000 & -0.003 & 0.001 & -0.003 & 0.042 & -0.295 & -0.004 & -0.031 & 0.035 & 0.696 & 0.708 & 0.000 & 0.011 & 0.147\\
\cmidrule{1-15}
 & 500 & 0.268 & 0.315 & 0.420 & 0.204 & 0.291 & 0.072 & 0.168 & 0.204 & 0.810 & 0.777 & 0.054 & 0.107 & 0.173\\

 & 1000 & 0.188 & 0.233 & 0.298 & 0.163 & 0.300 & 0.048 & 0.121 & 0.154 & 0.748 & 0.742 & 0.033 & 0.082 & 0.166\\

\multirow{-3}{*}[1\dimexpr\aboverulesep+\belowrulesep+\cmidrulewidth]{\raggedright\arraybackslash RMSE} & 2000 & 0.135 & 0.154 & 0.207 & 0.134 & 0.301 & 0.033 & 0.089 & 0.109 & 0.710 & 0.749 & 0.024 & 0.062 & 0.161\\
\cmidrule{1-15}
 & 500 & 0.954 & 0.954 & 0.952 & 0.895 & 0.309 & 0.946 & 0.955 & 0.953 & 0.020 & 0.017 & 0.907 & 0.875 & 0.844\\

 & 1000 & 0.958 & 0.932 & 0.941 & 0.849 & 0.038 & 0.957 & 0.942 & 0.920 & 0.000 & 0.000 & 0.902 & 0.834 & 0.808\\

\multirow{-3}{*}[1\dimexpr\aboverulesep+\belowrulesep+\cmidrulewidth]{\raggedright\arraybackslash Coverage} & 2000 & 0.941 & 0.950 & 0.945 & 0.799 & 0.001 & 0.944 & 0.933 & 0.924 & 0.000 & 0.000 & 0.869 & 0.796 & 0.773\\
\bottomrule
\end{tabular}
\begin{tablenotes}[para]
\item \textit{Notes:} This table displays the average bias (Bias), the Root Mean Squared Error (RMSE), and the coverage rate (Coverage) across $R=1000$ replicates. The rows contain results for various sample size $n$.
\end{tablenotes}
\end{threeparttable}}
\end{table}

\begin{table}[!h]
\centering
\caption{\label{tab:tab-SMCresults-Mar26-DGP4-MTEs}Simulation Results for Marginal Treatment Effects}
\centering
\resizebox{\ifdim\width>\linewidth\linewidth\else\width\fi}{!}{
\fontsize{8}{10}\selectfont
\begin{threeparttable}
\begin{tabular}[t]{cccccccccccc}
\toprule
\multicolumn{3}{c}{ } & \multicolumn{3}{c}{n = 500} & \multicolumn{3}{c}{n = 1000} & \multicolumn{3}{c}{n = 2000} \\
\cmidrule(l{3pt}r{3pt}){4-6} \cmidrule(l{3pt}r{3pt}){7-9} \cmidrule(l{3pt}r{3pt}){10-12}
Exposure & Grid $v$ & True Value & Bias & RMSE & Coverage & Bias & RMSE & Coverage & Bias & RMSE & Coverage\\
\midrule
 & 0.1 & 1.423 & -0.097 & 0.386 & 0.910 & -0.085 & 0.308 & 0.862 & -0.074 & 0.241 & 0.846\\

 & 0.2 & 1.290 & -0.064 & 0.322 & 0.920 & -0.056 & 0.250 & 0.879 & -0.052 & 0.192 & 0.867\\

 & 0.3 & 1.210 & -0.056 & 0.285 & 0.927 & -0.057 & 0.215 & 0.907 & -0.055 & 0.164 & 0.877\\

 & 0.4 & 1.155 & -0.063 & 0.262 & 0.929 & -0.067 & 0.201 & 0.906 & -0.062 & 0.151 & 0.884\\

 & 0.5 & 1.098 & -0.064 & 0.254 & 0.926 & -0.073 & 0.192 & 0.910 & -0.069 & 0.144 & 0.871\\

 & 0.6 & 1.047 & -0.070 & 0.258 & 0.924 & -0.078 & 0.190 & 0.904 & -0.071 & 0.142 & 0.875\\

 & 0.7 & 0.987 & -0.071 & 0.269 & 0.910 & -0.079 & 0.193 & 0.900 & -0.080 & 0.149 & 0.861\\

 & 0.8 & 0.915 & -0.072 & 0.298 & 0.908 & -0.081 & 0.210 & 0.904 & -0.084 & 0.163 & 0.853\\

\multirow{-9}{*}{\centering\arraybackslash Low} & 0.9 & 0.783 & -0.039 & 0.335 & 0.922 & -0.061 & 0.236 & 0.911 & -0.060 & 0.181 & 0.881\\
\cmidrule{1-12}
 & 0.1 & 1.823 & -0.094 & 0.348 & 0.902 & -0.084 & 0.278 & 0.839 & -0.076 & 0.231 & 0.823\\

 & 0.2 & 1.690 & -0.062 & 0.272 & 0.901 & -0.055 & 0.214 & 0.863 & -0.053 & 0.177 & 0.826\\

 & 0.3 & 1.610 & -0.054 & 0.228 & 0.907 & -0.056 & 0.178 & 0.865 & -0.056 & 0.145 & 0.835\\

 & 0.4 & 1.555 & -0.061 & 0.203 & 0.897 & -0.067 & 0.159 & 0.859 & -0.063 & 0.128 & 0.827\\

 & 0.5 & 1.498 & -0.062 & 0.192 & 0.896 & -0.072 & 0.148 & 0.862 & -0.070 & 0.118 & 0.797\\

 & 0.6 & 1.447 & -0.068 & 0.193 & 0.894 & -0.078 & 0.146 & 0.857 & -0.072 & 0.111 & 0.814\\

 & 0.7 & 1.387 & -0.069 & 0.207 & 0.888 & -0.078 & 0.150 & 0.873 & -0.081 & 0.119 & 0.805\\

 & 0.8 & 1.315 & -0.070 & 0.240 & 0.893 & -0.080 & 0.171 & 0.874 & -0.085 & 0.138 & 0.793\\

\multirow{-9}{*}{\centering\arraybackslash Medium} & 0.9 & 1.183 & -0.037 & 0.284 & 0.903 & -0.060 & 0.208 & 0.902 & -0.061 & 0.158 & 0.851\\
\cmidrule{1-12}
 & 0.1 & 2.223 & -0.092 & 0.386 & 0.920 & -0.083 & 0.297 & 0.871 & -0.077 & 0.249 & 0.848\\

 & 0.2 & 2.090 & -0.059 & 0.318 & 0.920 & -0.054 & 0.239 & 0.889 & -0.054 & 0.199 & 0.865\\

 & 0.3 & 2.010 & -0.052 & 0.281 & 0.924 & -0.055 & 0.213 & 0.903 & -0.058 & 0.169 & 0.886\\

 & 0.4 & 1.955 & -0.058 & 0.265 & 0.932 & -0.066 & 0.196 & 0.896 & -0.065 & 0.154 & 0.878\\

 & 0.5 & 1.899 & -0.059 & 0.257 & 0.927 & -0.071 & 0.188 & 0.903 & -0.071 & 0.145 & 0.883\\

 & 0.6 & 1.847 & -0.066 & 0.254 & 0.928 & -0.077 & 0.186 & 0.901 & -0.074 & 0.136 & 0.890\\

 & 0.7 & 1.787 & -0.067 & 0.264 & 0.927 & -0.077 & 0.190 & 0.905 & -0.083 & 0.141 & 0.886\\

 & 0.8 & 1.715 & -0.067 & 0.287 & 0.914 & -0.080 & 0.208 & 0.907 & -0.087 & 0.158 & 0.864\\

\multirow{-9}{*}{\centering\arraybackslash High} & 0.9 & 1.583 & -0.035 & 0.325 & 0.921 & -0.059 & 0.244 & 0.903 & -0.062 & 0.175 & 0.885\\
\bottomrule
\end{tabular}
\begin{tablenotes}[para]
\item \textit{Notes:} This table displays the average bias (Bias), the Root Mean Squared Error (RMSE), and the coverage rate (Coverage) across $R=1000$ replicates.
\end{tablenotes}
\end{threeparttable}}
\end{table}

\section{On Empirical Application}\label{BCIES-APDX_EmpiricalApp}

This appendix supplements the Opportunity Zones (OZ) application in Section \ref{BCIES-Section5}, which reports the empirical design, key posterior estimates, heterogeneous treatment effects, and policy counterfactuals. Here we provide supporting data documentation, full model estimates, and numerical results underlying Figures \ref{fig:fig-OZ-MTE-combined}--\ref{fig:fig-OZ-PRSE-shifta-bygroup}.

\subsection{Data, Variable Construction, and Spatial Setting}\label{BCIES-APDX_EmpiricalApp-data}

As described in Section \ref{BCIES-Section5_1}, the analysis uses \(3{,}699\) OZ-eligible California census tracts: 727 designated QOZs and \(2{,}972\) eligible but non-designated tracts. The outcome is housing-unit growth between 2017 and 2022, and treatment is QOZ designation among eligible tracts. Neighborhood exposure is the row-normalized share of adjacent tracts designated as QOZs. Figure \ref{fig:fig-OZ-California} provides the corresponding spatial distribution, including a Downtown Los Angeles inset that illustrates the local intermixing of designated and non-designated tracts.

\begin{figure}[H]

{\centering \includegraphics[width=1\linewidth,height=0.35\textheight]{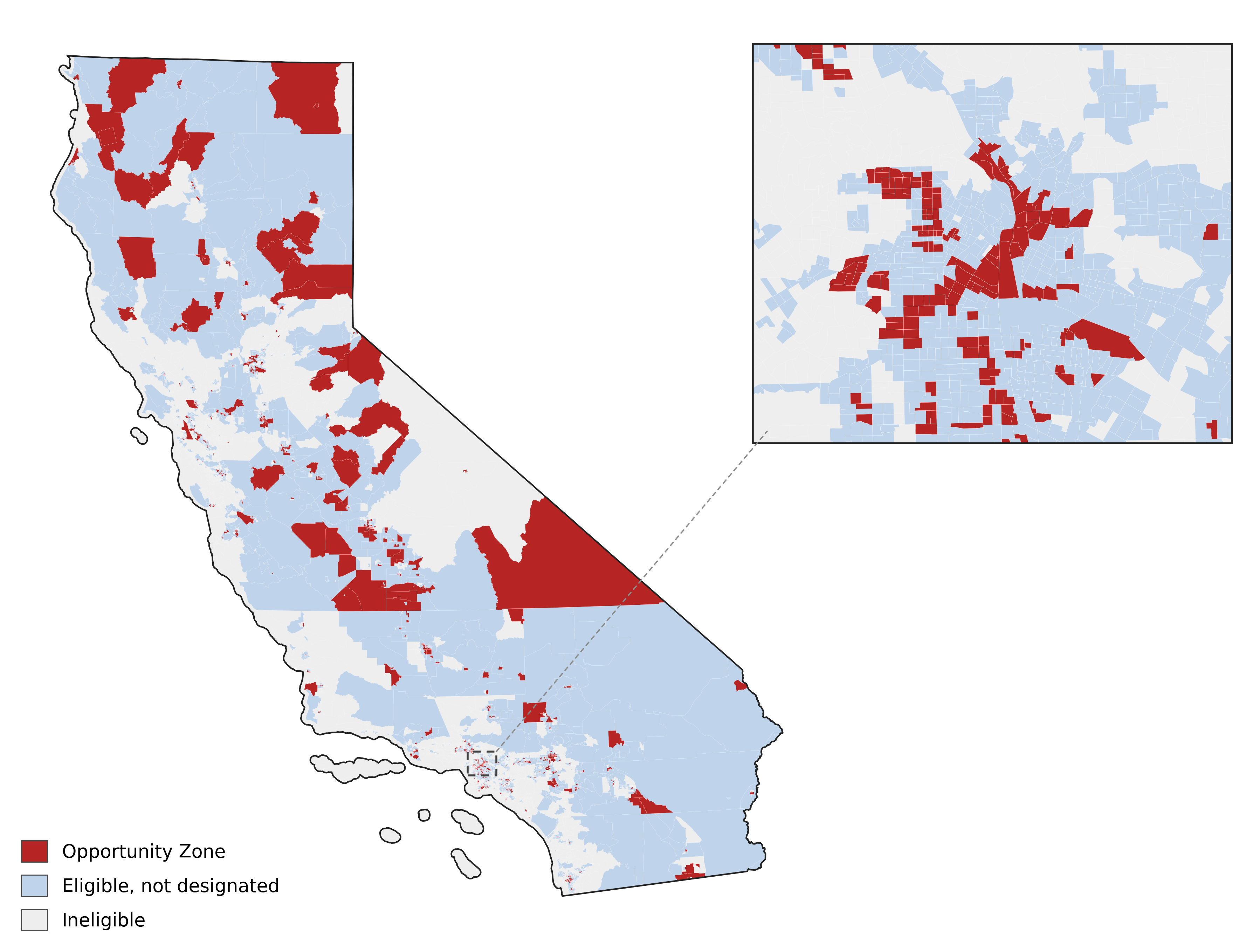} 

}

\caption{Opportunity Zone Designation across California Census Tracts. Census tracts are classified as designated Opportunity Zones (red), eligible but not designated (light blue), or ineligible (gray). The inset displays Downtown Los Angeles, illustrating the fine spatial intermixing of designated and non-designated tracts that motivates allowing for localized spillover effects.}\label{fig:fig-OZ-California}
\end{figure}

\begin{table}[htb]
        \caption{Variable Description}
        \label{tab-definition}
        \centering
\resizebox{\textwidth}{!}{%
\begin{tabular}{@{}ll@{}}
\toprule
\multicolumn{1}{c}{Variable} & \multicolumn{1}{c}{Description}                                                                         \\ \midrule
Housing Unit Growth           & Growth of total housing units between 2017-2022                                                            \\
QOZ                           & An indicator equal to one if an eligible tract was selected as an Opportunity Zone,  or zero otherwise. \\
Political Affiliation &
  \begin{tabular}[c]{@{}l@{}}An indicator equal to one if a tract’s representative to the state’s lower house is  of the same political \\ party as the state’s governor, and zero otherwise.\end{tabular} \\
Poverty Rate &
  \begin{tabular}[c]{@{}l@{}}The proportion of residents in a tract whose ratio of income to the poverty threshold is less than or \\ equal to 0.99, scaled by the number of residents.\end{tabular} \\
Median Earnings               & Logarithm of the median earnings in a tract.                                                            \\
Employment Rate &
  \begin{tabular}[c]{@{}l@{}}The number of individuals in the labor force in a tract that are working, either in civilian or \\ Armed Forces, scaled by the total labor force of the tract.\end{tabular} \\
\% White                      & The proportion of non-Hispanic white residents in a tract.                                              \\
\% Native hc covered          & The proportion of native-born individuals covered by health insurance in a tract.                       \\
\% Higher ed.                 & The proportion of the population in a tract with at least a high school education.                      \\
\% Rent                       & The proportion of rental unit in a tract.                                                               \\
Population                    & The total population in a tract from the 2010 Census.                                                   \\ \bottomrule
\end{tabular}%
}
\end{table}

\begin{table}[htb]
        \caption{Data Sources}
        \label{tab-source}
        \centering
\resizebox{\textwidth}{!}{%
\begin{tabular}{@{}lll@{}}
\toprule
\multicolumn{1}{c}{Source} & \multicolumn{1}{c}{URL} 
\\ \midrule
\begin{tabular}[c]{@{}l@{}}IRS (Internal Revenue Service) data: \\ List of Qualified Opportunity Zones\end{tabular} &
  \url{https://www.irs.gov/credits-deductions/businesses/opportunity-zones} &
   \\
Urban Institute's data       & \url{https://www.urban.org/policy-centers/metropolitan-housing-and-communities-policy-center/projects/opportunity-zones} &  \\
\begin{tabular}[c]{@{}l@{}}American Community Survey (ACS) \\ 5-Year Data (2009-2022)\end{tabular} &
  \url{https://www.census.gov/data/developers/data-sets/acs-5year.html} &
   \\
California State Legislature & \url{https://ballotpedia.org/California\_State\_Legislature}                                                               &  \\
TIGER Geographic Shapefiles  & \url{https://www.census.gov/geographies/mapping-files/time-series/geo/tiger-line-file.html}                              &  \\
SLDU and SLDL Blocks Splits  & \url{https://www.census.gov/geographies/mapping-files/2018/dec/rdo/2018-state-legislative-bef.html}                      &  \\ \bottomrule
\end{tabular}%
}
\end{table}

\begin{table}[H]
\centering\centering
\caption{\label{tab:tab-OZ-summarystats}Summary Statistics and Balancing Tests}
\centering
\fontsize{9}{11}\selectfont
\begin{threeparttable}
\begin{tabular}[t]{lcccr>{\raggedright\arraybackslash}p{1mm}c}
\toprule
\multicolumn{1}{c}{ } & \multicolumn{1}{c}{All tracts (n=3699)} & \multicolumn{1}{c}{QOZs (n=727)} & \multicolumn{1}{c}{Non-QOZs (n=2972)} & \multicolumn{3}{c}{QOZs -- Non-QOZs} \\
\cmidrule(l{3pt}r{3pt}){2-2} \cmidrule(l{3pt}r{3pt}){3-3} \cmidrule(l{3pt}r{3pt}){4-4} \cmidrule(l{3pt}r{3pt}){5-7}
\multicolumn{1}{c}{Variables} & \multicolumn{1}{c}{Mean (std)} & \multicolumn{1}{c}{Mean (std)} & \multicolumn{1}{c}{Mean (std)} & \multicolumn{1}{c}{Diff.Mean} & \multicolumn{1}{c}{} & \multicolumn{1}{c}{t-statistic}\\
\midrule
\addlinespace[0.3em]
\multicolumn{7}{l}{\textbf{Outcome}}\\
\hspace{1em}Housing Unit Growth & 0.03 (0.17) & 0.04 (0.14) & 0.03 (0.17) & 0.02 & * & 2.57\\
\addlinespace[0.3em]
\multicolumn{7}{l}{\textbf{Observed Characteristics}}\\
\hspace{1em}Political Affiliation & 0.79 (0.41) & 0.82 (0.38) & 0.78 (0.42) & 0.04 & ** & 2.72\\
\hspace{1em}Poverty Rate & 0.19 (0.09) & 0.27 (0.09) & 0.17 (0.08) & 0.09 & *** & 24.64\\
\hspace{1em}Median Earnings & 10.17 (0.31) & 10.01 (0.26) & 10.21 (0.30) & -0.21 & *** & -18.60\\
\hspace{1em}Median Rent & 7.08 (0.27) & 6.92 (0.26) & 7.11 (0.26) & -0.19 & *** & -17.79\\
\hspace{1em}Employment Rate & 0.29 (0.07) & 0.26 (0.07) & 0.29 (0.07) & -0.04 & *** & -12.87\\
\hspace{1em}\% White & 0.56 (0.21) & 0.53 (0.20) & 0.56 (0.21) & -0.04 & *** & -4.35\\
\hspace{1em}\% Native & 0.90 (0.04) & 0.89 (0.04) & 0.91 (0.04) & -0.02 & *** & -9.05\\
\hspace{1em}\% Higher ed. & 0.15 (0.09) & 0.11 (0.07) & 0.16 (0.09) & -0.05 & *** & -15.23\\
\hspace{1em}\% Rent & 0.57 (0.21) & 0.67 (0.19) & 0.54 (0.21) & 0.13 & *** & 16.47\\
Population & 4509.55 (1613.93) & 4305.31 (1476.18) & 4559.51 (1642.24) & -254.20 & *** & -4.07\\
\bottomrule
\end{tabular}
\begin{tablenotes}[para]
\item \textit{Notes:} This table presents summary statistics at the census tract level in California. All tracts refer to the entire sample of eligible census tracts for Opportunity Zones, which consist of selected tracts (QOZs) and eligible, not selected tracts (Non-QOZs). Two-sample t-statistics of tests for differences in mean values between two subsamples are reported. The asterisks *, **, and *** indicate statistical significance at the $10\%$, $5\%$, and $1\%$, respectively.
\end{tablenotes}
\end{threeparttable}
\end{table}

\subsection{Full Model Estimates and Specification Sensitivity}\label{BCIES-APDX_EmpiricalApp-estimates}

For completeness, the empirical specification maps the notation of Section \ref{BCIES-Section2_2} to \(D_i=QOZ_i\), \(Z_i=Political_i\), \(X_i=Demographic_i\), and \(\bar D_{\mathcal N i}=\overline{QOZ}_{\mathcal N i}\). The resulting selection and regime-specific outcome equations are
\begin{equation}
\begin{split}
QOZ_i &= \mathbbm{1}\left\{\alpha Political_i+X_i\beta^{(D)}+\epsilon_i^{(D)}>0\right\},\\
Y_i^{(d)} &=
\delta^{(d)}\overline{QOZ}_{\mathcal N i}+X_i\beta^{(d)}+\epsilon_i^{(d)},\qquad d\in\{0,1\}.
\end{split}
\end{equation}

\vspace{-0.75\baselineskip}

\begin{table}[H]
\centering
\caption{\label{tab:tab-OZ-est-results-full}Estimation Results}
\centering
\resizebox{\ifdim\width>\linewidth\linewidth\else\width\fi}{!}{
\fontsize{9}{11}\selectfont
\begin{threeparttable}
\begin{tabular}[t]{lllll}
\toprule
\multicolumn{1}{c}{ } & \multicolumn{2}{c}{(I) No Controls} & \multicolumn{2}{c}{(II) Additional Controls} \\
\cmidrule(l{3pt}r{3pt}){2-3} \cmidrule(l{3pt}r{3pt}){4-5}
 & Mean (std) & CI90 & Mean (std) & CI90\\
\midrule
\addlinespace[0.3em]
\multicolumn{5}{l}{\textbf{Treatment Decision Equation}}\\
\hspace{1em}$Political\ Affiliation\ (\alpha)$ & 0.125 (0.043) & {}[0.056, 0.198] & 0.162 (0.070) & {}[0.049, 0.277]\\
\hspace{1em}$Intercept\ (\beta^{(D)}_0)$ & -0.827 (0.042) & {}[-0.896, -0.759] & 0.011 (1.167) & {}[-1.943, 1.987]\\
\hspace{1em}$Poverty\ Rate\ (\beta^{(D)}_1)$ & -- & -- & 5.594 (0.348) & {}[5.022, 6.176]\\
\hspace{1em}$Median\ Income\ (\beta^{(D)}_2)$ & -- & -- & -0.194 (0.116) & {}[-0.389, 0.000]\\
\hspace{1em}$Employment\ Rate\ (\beta^{(D)}_3)$ & -- & -- & -0.761 (0.471) & {}[-1.532, 0.017]\\
\addlinespace[0.3em]
\multicolumn{5}{l}{\textbf{Outcome Equation for QOZs}}\\
\hspace{1em}$Neighborhood\ Treatment\ (\delta^{(1)})$ & 0.034 (0.017) & {}[0.005, 0.062] & 0.032 (0.014) & {}[0.009, 0.055]\\
\hspace{1em}$Intercept\ (\beta^{(1)}_0)$ & -0.206 (0.065) & {}[-0.255, -0.004] & -0.444 (0.233) & {}[-0.836, -0.067]\\
\hspace{1em}$Poverty\ Rate\ (\beta^{(1)}_1)$ & -- & -- & 0.304 (0.094) & {}[0.146, 0.454]\\
\hspace{1em}$Median\ Income\ (\beta^{(1)}_2)$ & -- & -- & 0.035 (0.023) & {}[-0.002, 0.073]\\
\hspace{1em}$Employment\ Rate\ (\beta^{(1)}_3)$ & -- & -- & 0.034 (0.077) & {}[-0.095, 0.160]\\
\addlinespace[0.3em]
\multicolumn{5}{l}{\textbf{Outcome Equation for Non-QOZs}}\\
\hspace{1em}$Neighborhood\ Treatment\ (\delta^{(0)})$ & 0.016 (0.013) & {}[-0.006, 0.038] & 0.009 (0.009) & {}[-0.006, 0.024]\\
\hspace{1em}$Intercept\ (\beta^{(0)}_0)$ & -0.025 (0.004) & {}[-0.032, -0.018] & -0.127 (0.089) & {}[-0.272, 0.020]\\
\hspace{1em}$Poverty\ Rate\ (\beta^{(0)}_1)$ & -- & -- & 0.039 (0.037) & {}[-0.023, 0.099]\\
\hspace{1em}$Median\ Income\ (\beta^{(0)}_2)$ & -- & -- & 0.015 (0.009) & {}[0.001, 0.030]\\
\hspace{1em}$Employment\ Rate\ (\beta^{(0)}_3)$ & -- & -- & -0.052 (0.033) & {}[-0.107, 0.003]\\
\addlinespace[0.3em]
\multicolumn{5}{l}{\textbf{Correlations and Variances}}\\
\hspace{1em}$\sigma_1^2$ & 0.047 (0.009) & {}[0.022, 0.058] & 0.018 (0.003) & {}[0.015, 0.023]\\
\hspace{1em}$\sigma_0^2$ & 0.038 (0.001) & {}[0.036, 0.040] & 0.045 (0.015) & {}[0.029, 0.071]\\
\hspace{1em}$\rho_{1D}$ & 0.839 (0.206) & {}[0.220, 0.920] & 0.182 (0.134) & {}[-0.063, 0.412]\\
\hspace{1em}$\rho_{0D}$ & -0.842 (0.011) & {}[-0.858, -0.824] & -0.130 (0.052) & {}[-0.218, -0.045]\\
\hspace{1em}$\rho_{10}$ & -0.708 (0.181) & {}[-0.794, -0.145] & -0.317 (0.101) & {}[-0.457, -0.150]\\
\addlinespace[0.3em]
\multicolumn{5}{l}{\textbf{Criteria}}\\
\hspace{1em}Log marginal likelihood & 471.819 (23.806) & {}[429.253, 507.149] & 1458.008 (6.315) & {}[1446.534, 1466.931]\\
\hspace{1em}AICM & 189.775 & -- & -2836.262 & --\\
\hspace{1em}Observations & 3699 & -- & 3699 & --\\
\addlinespace[0.3em]
\multicolumn{5}{l}{\textbf{Quantities of Interest}}\\
\hspace{1em}$\delta^{(1)}$ & 0.034 (0.017) & {}[0.005, 0.062] & 0.032 (0.014) & {}[0.009, 0.055]\\
\hspace{1em}$\delta^{(0)}$ & 0.016 (0.013) & {}[-0.006, 0.038] & 0.009 (0.009) & {}[-0.006, 0.024]\\
\hspace{1em}$\Delta_{\delta^{(1)}-\delta^{(0)}}$ & 0.017 (0.022) & {}[-0.018, 0.054] & 0.023 (0.016) & {}[-0.004, 0.051]\\
\hspace{1em}$\Delta_{\sigma_{1D}-\sigma_{0D}}$ & 0.349 (0.049) & {}[0.197, 0.386] & 0.052 (0.022) & {}[0.018, 0.092]\\
\bottomrule
\end{tabular}
\begin{tablenotes}[para]
\item \textit{Notes:} This table presents estimation results from Spillover Roy model. Posterior means, standard deviations, as well as lower and upper bounds of 90\% credible intervals are reported. Baseline specification (I) employs no control variables. Specification (II), which uses demographic characteristics as controls, provides robust results for quantities of interest and achieves better log marginal likelihood and Akaike's Information Criterion (AICM; Raftery \textit{et al.}, 2007).
\end{tablenotes}
\end{threeparttable}}
\end{table}

\vspace{-0.75\baselineskip}

Table \ref{tab:tab-OZ-est-results-full} provides the full posterior estimates underlying Table \ref{tab:tab-OZ-est-results} and compares the preferred specification with a parsimonious model that omits demographic controls. The preferred specification fits the data substantially better according to the reported criteria. The key patterns emphasized in Section \ref{BCIES-Section5_2} are robust across the two specifications: partisan alignment is positively associated with designation, and neighborhood exposure has a larger estimated effect for QOZs than for non-QOZs. The preferred specification includes pre-treatment demographic controls. Relative to the specification without controls, it provides substantially improved model fit according to the reported marginal-likelihood and AICM criteria and, more importantly, allows treatment selection to depend on observed socioeconomic characteristics that are relevant to OZ designation. Conditioning on observed tract characteristics substantially changes the estimated dependence structure.

\subsubsection*{Marginal Treatment Effects}\label{marginal-treatment-effects}
\addcontentsline{toc}{subsubsection}{Marginal Treatment Effects}

Table \ref{tab:tab-OZ-est-MTE} gives the numerical posterior estimates underlying Figure \ref{fig:fig-OZ-MTE-combined}. As discussed in Section \ref{BCIES-Section5_2}, the MTE declines with latent resistance at each exposure level, while higher neighborhood exposure shifts the MTE upward modestly. The table is included to report posterior standard deviations and 90\% credible intervals for the full resistance grid.

\vspace{-0.25\baselineskip}

\begin{table}[H]
\centering
\caption{\label{tab:tab-OZ-est-MTE}Marginal Treatment Effects by Resistance Level and Exposure.}
\centering
\resizebox{\ifdim\width>\linewidth\linewidth\else\width\fi}{!}{
\fontsize{10}{12}\selectfont
\begin{threeparttable}
\begin{tabular}[t]{ccccccc}
\toprule
\multicolumn{1}{c}{ } & \multicolumn{2}{c}{Low exposure (10th percentile)} & \multicolumn{2}{c}{Mean exposure} & \multicolumn{2}{c}{High exposure (90th percentile)} \\
\cmidrule(l{3pt}r{3pt}){2-3} \cmidrule(l{3pt}r{3pt}){4-5} \cmidrule(l{3pt}r{3pt}){6-7}
Resistance level ($v$) & Mean (std) & CI90 & Mean (std) & CI90 & Mean (std) & CI90\\
\midrule
0.1 & 0.027 (0.020) & {}[-0.006, 0.061] & 0.031 (0.020) & {}[-0.000, 0.064] & 0.041 (0.020) & {}[0.008, 0.074]\\
0.2 & 0.004 (0.019) & {}[-0.028, 0.036] & 0.008 (0.019) & {}[-0.022, 0.040] & 0.018 (0.019) & {}[-0.013, 0.050]\\
0.3 & -0.013 (0.021) & {}[-0.049, 0.021] & -0.008 (0.021) & {}[-0.043, 0.025] & 0.001 (0.021) & {}[-0.034, 0.035]\\
0.4 & -0.027 (0.025) & {}[-0.070, 0.014] & -0.022 (0.024) & {}[-0.065, 0.018] & -0.013 (0.024) & {}[-0.055, 0.027]\\
0.5 & -0.041 (0.029) & {}[-0.093, 0.008] & -0.036 (0.028) & {}[-0.088, 0.013] & -0.027 (0.028) & {}[-0.078, 0.022]\\
0.6 & -0.054 (0.033) & {}[-0.115, 0.003] & -0.049 (0.033) & {}[-0.110, 0.008] & -0.040 (0.033) & {}[-0.100, 0.017]\\
0.7 & -0.068 (0.038) & {}[-0.139, -0.003] & -0.063 (0.038) & {}[-0.134, 0.003] & -0.054 (0.038) & {}[-0.123, 0.012]\\
0.8 & -0.085 (0.044) & {}[-0.167, -0.009] & -0.080 (0.044) & {}[-0.162, -0.004] & -0.071 (0.044) & {}[-0.151, 0.005]\\
0.9 & -0.108 (0.053) & {}[-0.205, -0.017] & -0.103 (0.053) & {}[-0.200, -0.011] & -0.094 (0.053) & {}[-0.190, -0.003]\\
\bottomrule
\end{tabular}
\begin{tablenotes}[para]
\item \textit{Notes:} This table presents estimation results for Marginal Treatment Effects, evaluated at nine grid values of the unmeasured resistance level $v$ and different values of the neighborhood treatment $\bar{d}_\mathcal{N}$. Posterior means, standard deviations, as well as $90\%$ credible intervals for specification (II) with control variables are reported.
\end{tablenotes}
\end{threeparttable}}
\end{table}

\vspace{-0.25\baselineskip}

\subsubsection*{Average Effects under the Realized OZ Assignment}\label{average-effects-under-the-realized-oz-assignment}
\addcontentsline{toc}{subsubsection}{Average Effects under the Realized OZ Assignment}

Table \ref{tab:tab-OZ-est-CP} reports the numerical average effects summarized in Section \ref{BCIES-Section5_2}. It confirms that the estimated average direct and average total effects are positive for treated tracts, whereas the average spillover effect for untreated tracts is small and its 90\% credible interval includes zero.

\vspace{-0.25\baselineskip}

\begin{table}[H]
\centering
\caption{\label{tab:tab-OZ-est-CP}Average Effects under the Realized OZ Assignment}
\centering
\fontsize{10}{12}\selectfont
\begin{threeparttable}
\begin{tabular}[t]{cccc}
\toprule
ADT & AST & ATOT & ASUT\\
\midrule
0.045 [0.014, 0.078] & 0.014 [0.004, 0.023] & 0.059 [0.026, 0.096] & 0.001 [-0.001, 0.004]\\
\bottomrule
\end{tabular}
\begin{tablenotes}[para]
\item \textit{Notes:} This table presents estimation results for the average direct effect on the treated (ADT), the average spillover effect on the treated (AST), the average total effect on the treated (ATOT), and the average spillover effect on the untreated (ASUT) from the Spillover Roy model. Posterior means and $90\%$ credible intervals for specification (II) with control variables are reported.
\end{tablenotes}
\end{threeparttable}
\end{table}

\vspace{-0.25\baselineskip}

\subsection{Policy Counterfactual Analysis}\label{BCIES-APDX_EmpiricalApp-policy}

Section \ref{BCIES-Section5_3} defines the counterfactual expansion
\[
P_i^{a,\tau}=P_i^a+\tau(1-P_i^a),\qquad \tau\in[0,1],
\]
and Figure \ref{fig:fig-OZ-PRTE-shifta} displays the resulting policy-relevant effects. Table \ref{tab:tab-OZ-PRTE-shifta} provides the corresponding numerical posterior estimates and induced-treatment shares. Consistent with Figure \ref{fig:fig-OZ-PRTE-shifta}, the posterior mean PRDE declines with expansion, whereas the PRSE increases. The posterior mean PRTOT consequently exhibits diminishing returns and crosses zero for larger expansions. The 90\% credible intervals for PRDE and PRTOT include zero throughout the reported range, so Table \ref{tab:tab-OZ-PRTE-shifta} is also useful for assessing the uncertainty around the point-estimate patterns emphasized in the main text.

Table \ref{tab:tab-OZ-PRSE-shifta-bygroup} is the numerical counterpart to Figure \ref{fig:fig-OZ-PRSE-shifta-bygroup}. It shows that estimated spillover gains are largest for induced entrants and are also positive for always-treated tracts. For never-treated tracts, the posterior means are much smaller and the 90\% credible intervals include zero at every reported policy shift. These numerical results support the conclusion that spillover benefits are concentrated among tracts treated under the counterfactual policy.

\vspace{-0.25\baselineskip}

\begin{table}[H]
\centering
\caption{\label{tab:tab-OZ-PRTE-shifta}Policy-Relevant Effects under Policy Expansion}
\centering
\resizebox{\ifdim\width>\linewidth\linewidth\else\width\fi}{!}{
\fontsize{9}{11}\selectfont
\begin{threeparttable}
\begin{tabular}[t]{lcccc}
\toprule
Policy & Share induced & PRDE & PRSE & PRTOT\\
\midrule
$\tau = $0.05 & \makecell[c]{0.041 \\ {\scriptsize [0.040, 0.042]}} & \makecell[c]{0.012 \\ {\scriptsize [-0.019, 0.045]}} & \makecell[c]{0.014 \\ {\scriptsize [0.001, 0.026]}} & \makecell[c]{0.026 \\ {\scriptsize [-0.004, 0.060]}}\\
$\tau = $0.15 & \makecell[c]{0.122 \\ {\scriptsize [0.121, 0.124]}} & \makecell[c]{0.003 \\ {\scriptsize [-0.030, 0.034]}} & \makecell[c]{0.016 \\ {\scriptsize [0.003, 0.028]}} & \makecell[c]{0.019 \\ {\scriptsize [-0.012, 0.052]}}\\
$\tau = $0.25 & \makecell[c]{0.203 \\ {\scriptsize [0.200, 0.205]}} & \makecell[c]{-0.005 \\ {\scriptsize [-0.039, 0.028]}} & \makecell[c]{0.018 \\ {\scriptsize [0.005, 0.030]}} & \makecell[c]{0.013 \\ {\scriptsize [-0.020, 0.048]}}\\
$\tau = $0.35 & \makecell[c]{0.283 \\ {\scriptsize [0.279, 0.286]}} & \makecell[c]{-0.011 \\ {\scriptsize [-0.048, 0.025]}} & \makecell[c]{0.020 \\ {\scriptsize [0.007, 0.033]}} & \makecell[c]{0.008 \\ {\scriptsize [-0.028, 0.047]}}\\
$\tau = $0.45 & \makecell[c]{0.363 \\ {\scriptsize [0.359, 0.368]}} & \makecell[c]{-0.018 \\ {\scriptsize [-0.057, 0.020]}} & \makecell[c]{0.022 \\ {\scriptsize [0.008, 0.035]}} & \makecell[c]{0.004 \\ {\scriptsize [-0.037, 0.045]}}\\
$\tau = $0.55 & \makecell[c]{0.444 \\ {\scriptsize [0.438, 0.449]}} & \makecell[c]{-0.024 \\ {\scriptsize [-0.067, 0.016]}} & \makecell[c]{0.024 \\ {\scriptsize [0.008, 0.038]}} & \makecell[c]{-0.000 \\ {\scriptsize [-0.043, 0.043]}}\\
$\tau = $0.65 & \makecell[c]{0.525 \\ {\scriptsize [0.518, 0.531]}} & \makecell[c]{-0.030 \\ {\scriptsize [-0.077, 0.014]}} & \makecell[c]{0.025 \\ {\scriptsize [0.008, 0.042]}} & \makecell[c]{-0.004 \\ {\scriptsize [-0.051, 0.041]}}\\
$\tau = $0.75 & \makecell[c]{0.605 \\ {\scriptsize [0.597, 0.612]}} & \makecell[c]{-0.036 \\ {\scriptsize [-0.088, 0.011]}} & \makecell[c]{0.027 \\ {\scriptsize [0.009, 0.045]}} & \makecell[c]{-0.008 \\ {\scriptsize [-0.059, 0.039]}}\\
\bottomrule
\end{tabular}
\begin{tablenotes}[para]
\item \textit{Notes:} This table presents estimation results for the Policy-Relevant Direct Effect (PRDE), Policy-Relevant Spillover Effect (PRSE), and Total Policy-Relevant Total Effect (PRTOT) from the Spillover Roy model. Posterior means and $90\%$ credible intervals for specification (II) with control variables are reported.
\end{tablenotes}
\end{threeparttable}}
\end{table}

\vspace{-0.25\baselineskip}

\begin{table}[H]
\centering
\caption{\label{tab:tab-OZ-PRSE-shifta-bygroup}Group-Specific Spillover Effects under Policy Expansion}
\centering
\resizebox{\ifdim\width>\linewidth\linewidth\else\width\fi}{!}{
\fontsize{9}{11}\selectfont
\begin{threeparttable}
\begin{tabular}[t]{lccc}
\toprule
Policy & Always treated & Induced entrants & Never treated\\
\midrule
$\tau = $0.05 & \makecell[c]{0.0012 \\ {\scriptsize [0.0004, 0.0020]}} & \makecell[c]{0.0013 \\ {\scriptsize [0.0004, 0.0023]}} & \makecell[c]{0.0004 \\ {\scriptsize [-0.0002, 0.0010]}}\\
$\tau = $0.15 & \makecell[c]{0.0036 \\ {\scriptsize [0.0011, 0.0061]}} & \makecell[c]{0.0040 \\ {\scriptsize [0.0012, 0.0067]}} & \makecell[c]{0.0011 \\ {\scriptsize [-0.0007, 0.0029]}}\\
$\tau = $0.25 & \makecell[c]{0.0060 \\ {\scriptsize [0.0018, 0.0101]}} & \makecell[c]{0.0066 \\ {\scriptsize [0.0020, 0.0110]}} & \makecell[c]{0.0018 \\ {\scriptsize [-0.0012, 0.0048]}}\\
$\tau = $0.35 & \makecell[c]{0.0084 \\ {\scriptsize [0.0025, 0.0141]}} & \makecell[c]{0.0092 \\ {\scriptsize [0.0027, 0.0154]}} & \makecell[c]{0.0025 \\ {\scriptsize [-0.0017, 0.0066]}}\\
$\tau = $0.45 & \makecell[c]{0.0108 \\ {\scriptsize [0.0032, 0.0182]}} & \makecell[c]{0.0118 \\ {\scriptsize [0.0035, 0.0199]}} & \makecell[c]{0.0033 \\ {\scriptsize [-0.0022, 0.0085]}}\\
$\tau = $0.55 & \makecell[c]{0.0132 \\ {\scriptsize [0.0039, 0.0222]}} & \makecell[c]{0.0145 \\ {\scriptsize [0.0043, 0.0244]}} & \makecell[c]{0.0040 \\ {\scriptsize [-0.0027, 0.0104]}}\\
$\tau = $0.65 & \makecell[c]{0.0155 \\ {\scriptsize [0.0046, 0.0262]}} & \makecell[c]{0.0171 \\ {\scriptsize [0.0051, 0.0288]}} & \makecell[c]{0.0047 \\ {\scriptsize [-0.0031, 0.0123]}}\\
$\tau = $0.75 & \makecell[c]{0.0179 \\ {\scriptsize [0.0053, 0.0302]}} & \makecell[c]{0.0197 \\ {\scriptsize [0.0059, 0.0332]}} & \makecell[c]{0.0054 \\ {\scriptsize [-0.0036, 0.0141]}}\\
\bottomrule
\end{tabular}
\begin{tablenotes}[para]
\item \textit{Notes:} This table presents estimation results for the Policy-Relevant Direct Effect (PRDE), Policy-Relevant Spillover Effect (PRSE), and Policy-Relevant Total Effect (PRTOT) from the Spillover Roy model. Posterior means and $90\%$ credible intervals for specification (II) with control variables are reported.
\end{tablenotes}
\end{threeparttable}}
\end{table}

\end{document}